\documentclass[a4paper,12pt]{article}

\usepackage{latexsym}
\usepackage{amsmath}
\usepackage{amsthm}
\usepackage{amssymb}
\usepackage{enumerate}
\usepackage{multicol}
\usepackage{graphicx}

\allowdisplaybreaks[4]

\theoremstyle{plain}
\newtheorem{theorem}{Theorem}[section]
\newtheorem{proposition}[theorem]{Proposition}
\newtheorem{lemma}[theorem]{Lemma}
\newtheorem{corollary}[theorem]{Corollary}

\theoremstyle{definition}
\newtheorem{definition}[theorem]{Definition}

\newtheorem{remark}[theorem]{Remark}
\numberwithin{equation}{section}

\def\cB{{\mathcal{B}}}

\def\C{{\mathbb{C}}}
\def\cC{{\mathcal{C}}}

\def\cD{{\mathcal{D}}}

\def\be{{\mathbf{e}}}

\def\cF{{\mathcal{F}}}

\def\cH{{\mathcal{H}}}

\def\bk{{\mathbf{k}}}

\def\cL{{\mathcal{L}}}

\def\cM{{\mathcal{M}}}

\def\hm{{\hat{m}}}

\def\N{{\mathbb{N}}}
\def\cN{{\mathcal{N}}}

\def\cO{{\mathcal{O}}}

\def\cP{{\mathcal{P}}}
\def\bp{{\mathbf{p}}}

\def\R{{\mathbb{R}}}

\def\S{{\mathbb{S}}}

\def\bs{{\mathbf{s}}}

\def\T{{\mathbb{T}}}

\def\cU{{\mathcal{U}}}
\def\bu{{\mathbf{u}}}

\def\cV{{\mathcal{V}}}
\def\bv{{\mathbf{v}}}

\def\bW{{\mathbf{W}}}

\def\bx{{\mathbf{x}}}
\def\bX{{\mathbf{X}}}

\def\hcX{{\hat{\mathcal{X}}}}
\def\hbx{{\hat{\mathbf{x}}}}
\def\hbX{{\hat{\mathbf{X}}}}

\def\by{{\mathbf{y}}}
\def\bY{{\mathbf{Y}}}

\def\hcY{{\hat{\mathcal{Y}}}}
\def\hby{{\hat{\mathbf{y}}}}
\def\hbY{{\hat{\mathbf{Y}}}}

\def\Z{{\mathbb{Z}}}

\def\bZ{{\mathbf{Z}}}

\def\opsi{{\overline{\psi}}}

\def\spin{{\{\uparrow,\downarrow\}}}
\def\ua{{\uparrow}}
\def\da{{\downarrow}}
\def\la{{\lambda}}

\def\O{{\Omega}}
\def\o{{\omega}}

\def\eps{{\varepsilon}}
\def\g{{\gamma}}
\def\G{{\Gamma}}
\def\s{{\sigma}}
\def\hs{\hat{\sigma}}
\def\htau{\hat{\tau}}
\def\hrho{\hat{\rho}}
\def\heta{\hat{\eta}}
\def\D{{\Delta}}

\def\<{{\langle}}
\def\>{{\rangle}}

\def\Tr{\mathop{\mathrm{Tr}}}
\def\Arg{\mathop{\mathrm{Arg}}}

\def\dis{\mathop{\mathrm{dis}}\nolimits}

\def\sgn{\mathop{\mathrm{sgn}}\nolimits}

\def\Im{\mathop{\mathrm{Im}}}
\def\Re{\mathop{\mathrm{Re}}}

\def\b0{{\mathbf{0}}}

\def\ec{{\epsilon_c^{\sigma}}}
\def\eo{{\epsilon_o^{\sigma}}}

\begin{document}

\title{Exponential Decay of 
Equal-Time Four-Point Correlation Functions
 in the Hubbard Model on the Copper-Oxide Lattice}
\author{Yohei Kashima \medskip \\
Department of Mathematical Sciences, University of Tokyo\\
Komaba, Tokyo, 153-8914, Japan\\ 
kashima@ms.u-tokyo.ac.jp}
\date{}

\maketitle

\begin{abstract}
\noindent
For the Hubbard model on the two-dimensional
 copper-oxide lattice, equal-time four-point correlation functions at positive
 temperature are proved to decay exponentially in the thermodynamic 
 limit if the magnitude of the on-site interactions is smaller than some power of
 temperature. This result especially implies that the equal-time
 correlation functions for singlet Cooper pairs of various symmetries
 decay exponentially in the distance between the Cooper pairs in high temperatures or in low-temperature
 weak-coupling regimes. The proof is based on a multi-scale integration
 over the Matsubara frequency.
\end{abstract}

\section{Introduction}\label{sec_introduction}
\subsection{Introductory remarks}
In order to explain high-temperature superconductivity in ceramic copper
oxide materials, several tight-binding models for the charge carriers in
2 dimensional plane have been proposed with the consensus that the
superconducting pairing mechanism should be understood by focusing on
the conducting $\text{CuO}_2$ plane first. In the hierarchy of the
well-known 2D models (see, e.g, \cite{D}) the three-band Hubbard model
on the copper-oxide lattice (\cite{E}), or the CuO Hubbard model in short, is
believed to be the closest to the reality since it explicitly
distinguishes one relevant electron orbital of the copper and those of
the oxygens surrounding the copper in the unit cell. Being more
realistic also means being more complex. Rigorous mathematical methods
need to be developed to explore the relatively involved structure of the CuO
Hubbard model in depth. 

In this paper we prove that equal-time 4-point correlation functions in
the CuO Hubbard model at positive temperature decay exponentially in the
thermodynamic limit if the coupling constants on both the copper
and the oxygen sites are smaller than some power of
temperature. The result will be fully stated in Subsection
\ref{subsec_main_result}.  One direct consequence of this theorem is
the exponential decay of pairing-pairing correlation functions 
in the distance between the center of 2 electrons and that
of 2 holes, excluding long range correlations between singlet Cooper
pairs in high temperatures or in low-temperature weak-coupling regimes. 

It has been proved in \cite{K2} that finite-temperature equal-time correlation functions for
many-electron models, including the Hubbard model as one instance, on the
hyper-cubic lattice of arbitrary dimension decay exponentially if the
interaction is smaller than some power of temperature. The proof of
\cite{K2} essentially uses the volume-, temperature-independent determinant
bound on the covariance matrix established by Pedra and Salmhofer
(\cite{PS}). The exponential decay of the correlation functions in the CuO
Hubbard model cannot be deduced as an immediate corollary of the theorems in
\cite{K2}, since Pedra-Salmhofer's determinant bound in its original form
\cite[\mbox{Theorem 2.4}]{PS} does not apply to the covariance for
multi-band many-Fermion models such as the CuO Hubbard model. Thus one
has to alter the way to achieve the goal. As a way out we expand the
covariance over the Matsubara frequency through the Fourier transform
this time and try to control the correlation function analytically by means
of a multi-scale expansion along the segments of the large Matsubara
frequency. The dispersion relation for the free particle hopping to the
nearest neighbor sites on the CuO lattice can be a square root of cosine
of the momentum variable, which is, unlike in the single-band models
treated in \cite{K1}, \cite{K2}, non-analytic. Once transformed into the
Matsubara sum, however, the covariance appears to contain only the
square of the dispersion relation. Thus the covariance in the Matsubara
sum representation explicitly shows its analytic property with respect
to the momentum variable.  As in \cite{K2} the
analyticity of the covariance enables us to reformulate the correlation function multiplied
by the distance between the electrons and the holes into a multi-contour
integral of the correlation function with respect to new complex
variables inserted in the covariance. The practical role of the
multi-scale integration over the Matsubara frequency in this paper is to
establish a volume-independent upper bound on the perturbed correlation
function inside the multi-contour integral. Due to a self-contained
nature of the multi-scale Matsubara expansion, the proofs in this paper
merely rely on the repeated use of the tree formula for logarithm of the
Grassmann Gaussian integral. 

More precisely speaking, the correlation function of our original interest is
expressed as a well-defined finite dimensional Grassmann integral
during the intermediate technical construction. In the major part of
this paper we deal with the Grassmann integral formulation, which is
flexible to mathematical manipulations, as the rigorous counterpart of
the correlation function. This is the same stance as taken in \cite{K1},
\cite{K2}, or more generally in the constructive Fermionic quantum field
theory (see, e.g, \cite{FKT}). Finally by sending the finite dimensional
formulation to the limit we withdraw the conclusion on the original
correlation function defined by trace operations over the Fermionic Fock space.

This paper is not the first to consider multi-scale analysis over
the large Matsubara frequency. On the contrary, a number of papers have
already discussed qualitatively similar problems to the Matsubara
ultraviolet problem posed in this paper. See, e.g, \cite{BGM}, \cite{BGPS},
\cite{G}, \cite{GM} by one of the pioneering groups of the subject. 
One of the purposes of this paper is set to provide readers with an
alternative method to solve the Matsubara ultraviolet problem.
In order to help the readers to properly comprehend the purpose of this
paper, let us summarize the main differences between the methods used in this
article and those in the preceding papers. First, this paper uses a
version of the finite dimensional Grassmann integral formulation reported in \cite{K1}.
 The reduction to the finite dimensionality in this formulation is based
 on the discretization of the interval of temperature in the perturbative
 expansion of the partition function. Accordingly the basis of Grassmann algebra is indexed by the finite space-time variables and the
 step size of the discretization explicitly appears in the
 characterization of the covariance as a parameter, changing the face of
 the covariance from the well-known free propagator. This paper does not
 introduce Grassmann algebra indexed by the momentum variables. In
 \cite{BGM}, \cite{G}, \cite{GM} the derivation of the finite
 dimensional Grassmann integral formulation is based on the cut-off on
 the Matsubara frequency. As a result the basis of Grassmann algebra is
 indexed by the finite momentum variables. Secondly, the multi-scale
 analysis in this paper is completed by the induction on the scale
 level, which assumes a norm bound on the input and then proves the
 relevant norm bound on the output produced by the single-scale
 integration. The papers \cite{BGM}, \cite{BGPS}, \cite{G}, \cite{GM}
 use a family of trees called the Gallavotti-Nicol\`o trees to organize
 the multi-scale integration process, achieving collective descriptions
 of the theory all through the integration levels. This paper's concept
 of finding a norm bound on the output of the integration at one scale is
 closer to the rigorous analysis on finite dimensional Grassmann algebra
 established in \cite{FKT1}, \cite{FKT}. However, the paper \cite{FKT1}
 and the book \cite{FKT} apply a representation theorem developed by
 themselves to expand logarithm
 of the Grassmann Gaussian integral, while this paper as well as the
 papers \cite{BGM}, \cite{BGPS}, \cite{G}, \cite{GM} use the tree
 expansion for the same purpose. Thirdly, this paper
 derives equal-time 4-point correlation functions by substituting an artificial
 quartic term into the original Hamiltonian and differentiating the free
 energy governed by the modified Hamiltonian with respect to the
 coefficient of the artificial term. The papers \cite{BGM}, \cite{BGPS},
 \cite{G}, \cite{GM} derive correlation functions by inserting the
 source Grassmann variables into the Grassmann integral formulation and then
 letting the Grassmann derivatives act on the modified Grassmann
 integral formulation called the generating function.

Though this paper involves a multi-scale analysis concerning the
Matsubara sum as the main technical ingredient, it does not treat any
infrared multi-scale analysis around zero points of the dispersion
relation. Accordingly this paper has no improvement on the temperature
dependency of the allowed magnitude of the interaction over the single-scale analysis \cite{K1}, 
\cite{K2} and cannot study the behavior of correlation functions at zero
temperature. In recent years infrared multi-scale integration techniques
have been intensively applied to describe the zero-temperature limit of
thermal expectation values of various observables in the Hubbard model
on the honeycomb lattice by Giuliani and Mastropietro (\cite{GM}) and by
Giuliani, Mastropietro and Porta (\cite{GMP}). In connection with the
main result of this article we should remark that the many-electron model of graphene
studied in \cite{GM}, \cite{GMP} also has a matrix-valued kinetic
energy, so the single-scale analysis previously reported in \cite{K2}
does not prove the exponential decay of the finite-temperature
correlation functions in the system. However, it is straightforward to
adapt the proofs in this article to conclude the same result for the Hubbard model
on the honeycomb lattice as claimed for the CuO Hubbard model.

This paper is outlined as follows. In the following subsections we
define the CuO Hubbard model and
state the main result of this paper. In Section \ref{sec_formulation} we
characterize the correlation function as a limit of the finite
dimensional Grassmann integral and derive the
contour integral formulation. In Section
\ref{sec_preliminaries} we prepare some necessary tools for the
multi-scale integration such as the cut-off function and the sliced
covariances. In Section \ref{sec_multiscale_integration} we carry out
the multi-scale integration over the Matsubara frequency and prove the
main theorem. In Appendix \ref{app_covariance} we derive the covariance
governed by the free Hamiltonian on the CuO lattice. Appendix \ref{app_formulation_convergence}
provides a sketch of how to prove the convergence property of the
Grassmann integral formulation. In Appendix \ref{app_log_grassmann} we
prove a general formula for logarithm of Grassmann polynomials, which is
necessary for the multi-scale
integration. Finally Appendix \ref{app_thermodynamic_limit} shows
that the correlation function converges to a finite value in the
thermodynamic limit if the coupling constants obey the smallness
condition under which the multi-scale analysis is performed.

\subsection{The Hubbard model on the CuO lattice}\label{subsec_model_hamiltonian}
Here we define the model Hamiltonian operator. For $L\in\N$ let
$\G:=(\Z/L\Z)^2$. The CuO lattice consists of 3 separate lattices, each
of which is isomorphic to $\G$ (see Figure \ref{fig_CuO_lattice}).
\begin{figure}
\begin{center}
\begin{picture}(60,60)(0,0)
\put(0,0){$\bullet$}
\put(40,0){$\bullet$}
\put(0,40){$\bullet$}
\put(40,40){$\bullet$}
\put(20,0){$\circ$}
\put(60,0){$\circ$}
\put(20,40){$\circ$}
\put(60,40){$\circ$}
\put(0,20){$\diamond$}
\put(40,20){$\diamond$}
\put(0,60){$\diamond$}
\put(40,60){$\diamond$}
\end{picture}
 \caption{The CuO lattice for $L=2$, where `$\bullet$' denotes Cu sites,
 `$\circ$' denotes O sites and `$\diamond$' denotes the other O sites.}\label{fig_CuO_lattice}
\end{center}
\end{figure}
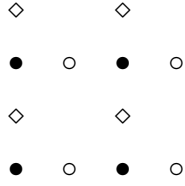
For $\bx\in \G$ let $(1,\bx)$ represent a Cu site, $(2,\bx)$ represent the
O site right to $(1,\bx)$, and $(3,\bx)$ denote the O site above
$(1,\bx)$ (see Figure \ref{fig_site_index}).
\begin{figure}
\begin{center}
\begin{picture}(55,55)(0,0)
\put(5,15){$\bullet$}
\put(0,0){$(1,\bx)$}
\put(45,15){$\circ$}
\put(40,0){$(2,\bx)$}
\put(5,55){$\diamond$}
\put(0,40){$(3,\bx)$}
\end{picture}
 \caption{Labeling each site.}\label{fig_site_index}
\end{center}
\end{figure}
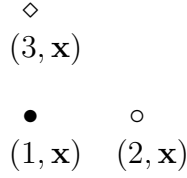
The CuO lattice is viewed as the union of $\{(\rho,\bx)\ |\ \bx\in\G\}$
$(\rho=1,2,3)$. The model Hamiltonian is defined as a self-adjoint
operator on the Fermionic Fock space
$F_f(L^2(\{1,2,3\}\times\G\times\spin))$. See, e.g, \cite[\mbox{Appendix
A}]{K1} for a brief description of the Fermionic Fock space defined on a finite
lattice. The CuO Hubbard model was originally designed to govern the total energy
of holes moving and interacting on the $\text{CuO}_2$ plane (see \cite{E}). Thus
the vacuum of $F_f(L^2(\{1,2,3\}\times\G\times\spin))$ should be
interpreted as the state where every site of $\{1,2,3\}\times \G$ is
occupied by an electron-pair. 

For $(\rho,\bx,\s)\in \{1,2,3\}\times\G\times\spin$ let $\psi_{\rho \bx
\s}$ be the annihilation operator defined on
$F_f(L^2(\{1,2,3\}\times\G\times\spin))$. The physical role of $\psi_{\rho \bx
\s}$ is to annihilate a hole with spin $\s$ at the site $(\rho,\bx)$. We
write the adjoint operator of $\psi_{\rho \bx \s}$ as $\psi_{\rho \bx
\s}^*$. The operator $\psi_{\rho \bx \s}^*$ is called the creation
operator and physically considered to be creating a hole with spin $\s$
at the site $(\rho,\bx)$. The CuO Hubbard model $H$ is defined as
follows. 
\begin{align*}
&H:=H_0+V,\\
&H_0:=t\sum_{(\bx,\s)\in\G\times\spin}(\psi_{1\bx\s}^*\psi_{2\bx\s}+\psi_{1\bx\s}^*\psi_{2(\bx-\be_1)\s}+\psi_{1\bx\s}^*\psi_{3\bx\s}+\psi_{1\bx\s}^*\psi_{3(\bx-\be_2)\s}+\text{h.c})\\
&\qquad\quad+\sum_{(\bx,\s)\in\G\times\spin}\ec\psi_{1\bx\s}^*\psi_{1\bx\s}+\sum_{(\rho,\bx,\s)\in\{2,3\}\times\G\times\spin}\eo\psi_{\rho\bx\s}^*\psi_{\rho\bx\s},\\
&V:=U_c\sum_{\bx\in\G}\psi_{1\bx\ua}^*\psi_{1\bx\da}^*\psi_{1\bx\da}\psi_{1\bx\ua}+U_o\sum_{(\rho,\bx)\in \{2,3\}\times
 \G}\psi_{\rho\bx\ua}^*\psi_{\rho\bx\da}^*\psi_{\rho\bx\da}\psi_{\rho\bx\ua},
\end{align*}
where $\be_1:=(1,0),\be_2:=(0,1)\in\Z^2$ and the terminology
``Hermitian conjugate'' is shortened to ``h.c'', meaning that the
adjoint operators of the operators in front are placed. The parameters $t$,
$U_c$, $U_o$, $\ec$, $\eo$ ($\s\in\spin$) are initially set to be
real. The parameter $t$ is the hopping amplitude between a Cu
site and the neighboring O sites. The parameters $\ec$ and $\eo$
represent the on-site energy minus the hole chemical potential for the Cu
sites and the O sites, respectively. We assume that the quadratic
Hamiltonian $H_0$ may contain the contribution from the magnetic field
such as
$h_c\sum_{\bx\in\G}S_{1,\bx}^z+h_o\sum_{(\rho,\bx)\in\{2,3\}\times\G}S_{\rho,\bx}^z$
with $h_c, h_o\in\R$, 
$S_{\rho,\bx}^z:=\frac{1}{2}(\psi_{\rho\bx\ua}^*\psi_{\rho\bx\ua}-\psi_{\rho\bx\da}^*\psi_{\rho\bx\da})$
$(\rho\in \{1,2,3\},\bx\in\G)$. This is the reason why $\ec$ and $\eo$ are defined to be
spin-dependent. The strength of the on-site interaction is expressed
by $U_c$ on the Cu sites and by $U_o$ on the O sites. 

Let $\beta>0$ denote the inverse of temperature times the Boltzmann
constant. The thermal expectation value of an observable $O$ is defined
as $\Tr(e^{-\beta H}O)/\Tr e^{-\beta H}$, where the trace is taken over
the Fock space $F_f(L^2(\{1,2,3\}\times\G\times\spin))$. For conciseness
we write $\<O\>_L$ in place of $\Tr(e^{-\beta H}O)/\Tr e^{-\beta H}$.
   
\subsection{Exponential decay property of the correlation
  functions}\label{subsec_main_result}
Let $\|\cdot\|_{\R^2}$ denote the Euclidean norm of $\R^2$ and $e(\approx
2.71828)$ be the base of the natural logarithms. This paper is devoted
to establish the following theorem.
\begin{theorem}\label{thm_exponential_decay}
There exist non-decreasing positive functions $f_1(\cdot)$,
 $f_2(\cdot):\R_{\ge 1}\to\R_{>0}$ such that if
\begin{equation}\label{eq_theorem_constraint}
|U_c|,|U_o|\le
 \frac{1}{f_1\left(\max_{\s\in\spin}\{1,|t|,|\ec|,|\eo|\}\right)\max\{1,\beta^{16}\}\beta},
\end{equation}
$\lim_{L\to \infty\atop L\in
 \N}\<\psi_{\hrho_1\hbx_1\hs_1}^*\psi_{\hrho_2\hbx_2\hs_2}^*\psi_{\heta_2\hby_2\htau_2}\psi_{\heta_1\hby_1\htau_1}+\text{h.c}\>_L$ 
exists and satisfies that
\begin{align}
&\left|\lim_{L\to \infty\atop L\in
 \N}\<\psi_{\hrho_1\hbx_1\hs_1}^*\psi_{\hrho_2\hbx_2\hs_2}^*\psi_{\heta_2\hby_2\htau_2}\psi_{\heta_1\hby_1\htau_1}+\text{h.c}\>_L\right|\le
 f_2\left(\max_{\s\in\spin}\{1,|t|,|\ec|,|\eo|\}\right)\max\{1,\beta^{16}\}\notag\\
&\quad\cdot\left(\frac{1}{\max\{1,t^2\}\max\{\beta,\beta^2\}}+1\right)^{-\frac{1}{8e}\|\sum_{j=1}^2(\hat{s}(\hs_j)\hbx_j-\hat{s}(\htau_j)\hby_j)\|_{\R^2}},\label{eq_theorem_decay_bound}
\end{align}
for any $(\hrho_j,\hbx_j,\hs_j),(\heta_j,\hby_j,\htau_j)\in
 \{1,2,3\}\times\Z^2\times\spin$ $(j=1,2)$,
 $t,\ec,\eo\in\R$ $(\s\in\spin)$, $\beta\in\R_{>0}$ and any map $\hat{s}(\cdot):\spin\to
 \{1,-1\}$. 
\end{theorem}

\begin{remark}
 The correlation function $\<\psi_{\hrho_1\hbx_1\hs_1}^*\psi_{\hrho_2\hbx_2\hs_2}^*\psi_{\heta_2\hby_2\htau_2}\psi_{\heta_1\hby_1\htau_1}+\text{h.c}\>_L$
 is defined for $\hbx_1,\hbx_2,\hby_1,\hby_2\in\Z^2$ by considering $\hbx_1,\hbx_2,\hby_1,\hby_2$ as the corresponding sites in $\G$ by periodicity.
\end{remark}

\begin{remark}
As a result of our proof, the growth rates of $f_1(\cdot)$, $f_2(\cdot)$
 are estimated as $f_1(x)=O(x^{44})$, $f_2(x)=O(x^{36})$ ($x\to\infty$). However, since it is not the main aim of our analysis, these
 orders are not quantitatively optimized. 
\end{remark}

\begin{remark}
The theorem provides decay bounds on the thermodynamic limit of the
 correlation functions for singlet Cooper pairs. For instance let us
 define the s-wave pairing operator $\D_s(\rho,\bx)$, the extended s-wave
 pairing operator $\D_{s^*}(\rho,\bx)$ and the $d_{x^2-y^2}$-wave pairing operator $\D_{d_{x^2-y^2}}(\rho,\bx)$ as follows. For $(\rho,\bx)\in\{1,2,3\}\times\G$,
\begin{align*}
&\D_s(\rho,\bx):=\psi_{\rho\bx\da}\psi_{\rho\bx\ua},\\
&\D_{s^*}(\rho,\bx):=\frac{1}{2}(\psi_{\rho(\bx+\be_1)\da}\psi_{\rho\bx\ua}+\psi_{\rho(\bx-\be_1)\da}\psi_{\rho\bx\ua}+\psi_{\rho(\bx+\be_2)\da}\psi_{\rho\bx\ua}+\psi_{\rho(\bx-\be_2)\da}\psi_{\rho\bx\ua}),\\
&\D_{d_{x^2-y^2}}(\rho,\bx):=\frac{1}{2}(\psi_{\rho(\bx+\be_1)\da}\psi_{\rho\bx\ua}+\psi_{\rho(\bx-\be_1)\da}\psi_{\rho\bx\ua}-\psi_{\rho(\bx+\be_2)\da}\psi_{\rho\bx\ua}-\psi_{\rho(\bx-\be_2)\da}\psi_{\rho\bx\ua}).
\end{align*}
If the map $\hat{s}(\cdot):\spin\to \{1,-1\}$ is identically 1, the
 theorem shows that \\
$|\lim_{L\to
 \infty,L\in\N}\<\D_a(\hat{\rho},\hbx)^*\D_a(\hat{\eta},\hby)+\text{h.c}\>_L|$ decays
 exponentially with $\|\hbx-\hby\|_{\R^2}$ for
 $\hat{\rho},\hat{\eta}\in\{1,2,3\}$, 
 $a=s,s^*,d_{x^2-y^2}$. If we take $\hat{s}(\cdot)$ to obey
 $\hat{s}(\ua)=-\hat{s}(\da)$, on the other hand, the theorem also implies
 exponential decay of spin-spin correlation
 functions of the form
 $\lim_{L\to \infty,L\in\N}\<S^x_{\hrho,\hbx}S^x_{\heta,\hby}+S^y_{\hrho,\hbx}S^y_{\heta,\hby}\>_L$
 with $\|\hbx-\hby\|_{\R^2}$, where the spin
 operators $S_{\rho,\bx}^x$, $S_{\rho,\bx}^y$ are defined by
 $S_{\rho,\bx}^x:=\frac{1}{2}(\psi_{\rho\bx\ua}^*\psi_{\rho\bx\da}+\psi_{\rho\bx\da}^*\psi_{\rho\bx\ua})$,
 $S_{\rho,\bx}^y:=\frac{1}{2}(-i\psi_{\rho\bx\ua}^*\psi_{\rho\bx\da}+i\psi_{\rho\bx\da}^*\psi_{\rho\bx\ua})$
 $((\rho,\bx)\in \{1,2,3\}\times \G)$.
 \end{remark}

\begin{remark}
The coupling constants $U_c$, $U_o$ satisfying \eqref{eq_theorem_constraint}
can be taken arbitrarily large as
 $\beta\searrow 0$. This means that the theorem generally 
proves exponential decay of the correlation functions in high temperatures.
\end{remark}

\begin{remark}\label{rem_half_filled}
Consider the case that $\ec=-\frac{1}{2}U_c$ and $\eo=-\frac{1}{2}U_o$ $(\forall
 \s\in\spin)$. The Hamiltonian $H$ becomes invariant under the transform
$\psi_{1\bx\s}\to \psi_{1\bx\s}^*$, $\psi_{1\bx\s}^*\to \psi_{1\bx\s}$, 
$\psi_{\rho\bx\s}\to -\psi_{\rho\bx\s}^*$, $\psi_{\rho\bx\s}^*\to
 -\psi_{\rho\bx\s}$ ($\rho\in \{2,3\},(\bx,\s)\in\G\times
 \spin$). This invariance implies that
 $\<\psi_{\rho\bx\s}^*\psi_{\rho\bx \s}\>_L=\frac{1}{2}$ $(\forall
 (\rho,\bx,\s)\in\{1,2,3\}\times\G\times\spin)$ and thus the system is
 half-filled. According to our construction, $f_1(1)>1$. If $\beta>1$,
 the constraint \eqref{eq_theorem_constraint} implies $|U_c|$,
 $|U_o|<1$. Therefore, we can claim the theorem for $\beta>1$ by eliminating $\ec,\eo$
 $(\s\in\spin)$ in the right-hand sides of \eqref{eq_theorem_constraint}
 and \eqref{eq_theorem_decay_bound}. On the other hand, for arbitrarily large
 $|U_c|$, $|U_o|$ there exists $\beta\le 1$ such that
 \eqref{eq_theorem_constraint} holds. Thus, the theorem concludes the
 exponential decay of correlation functions with the strong couplings if the temperature is high enough. \end{remark}

\begin{remark}
A power-law decay property of equal-time 4-point correlation
 functions can be proved by exactly following the argument of
 \cite{KT}. One result is that
\begin{equation*}
\limsup_{L\to \infty\atop L\in
 \N}|\<\psi_{\hrho_1\hbx\hs_1}^*\psi_{\hrho_2\hbx\hs_2}^*\psi_{\heta_2\hby\htau_2}\psi_{\heta_1\hby\htau_1}+\text{h.c}\>_L|
\le 2\|\hbx-\hby\|_{\R^2}^{-\tilde{c}f(\beta)},
\end{equation*}
for any $\hbx,\hby\in\Z^2$ with sufficiently large $\|\hbx-\hby\|_{\R^2}$, $(\hrho_j,\hs_j),(\heta_j,\htau_j)\in
 \{1,2,3\}\times\spin$ $(j=1,2)$ and $\beta \in\R_{>0}$, where $\tilde{c}>0$ is a
 constant, the function $f(\cdot):\R_{>0}\to \R_{>0}$ is
 decreasing and asymptotically behaves as $f(\beta)=O(\beta^{-1})$
$(\beta\to\infty)$, $O(|\log\beta|)$ $(\beta\searrow 0)$. An advantage
 of the framework \cite{KT}, apart from its conciseness, is that it
 requires no constraint on the magnitude of the interactions. However,
 it has not been applied to prove exponential decay of correlations in 2D
 many-electron systems, to the author's knowledge.
\end{remark} 

\section{Formulation}\label{sec_formulation}
In this section we formulate the
correlation function by using the notion of Grassmann integral
and show that the Grassmann integral representation of the correlation function
multiplied by the distance between the holes and the electrons is
transformed into a contour integral of the Grassmann integral.  This
procedure is essentially the same as we did in \cite{K1},
\cite{K2}. In order to avoid unnecessary repetition we present the
proofs at a minimum.

Let us introduce notations which are used throughout the paper. For
simplicity set $E_{max}:=\max_{\s\in\spin}\{1,|t|,|\ec|,|\eo|\}$. 
The sites on which the 4-point correlation function
 is defined are fixed to be $(\hrho_1,\hbx_1,\hs_1)$, $(\hrho_2,\hbx_2,\hs_2)$,
$(\heta_1,\hby_1,\htau_1)$, $(\heta_2,\hby_2,\htau_2)\in \{1,2,3\}\times
\Z^2\times \spin$. We simply write $\hcX_j$, $\hcY_j$
 $(j=1,2)$ instead of $(\hrho_j,\hbx_j,\hs_j)$,
 $(\heta_j,\hby_j,\htau_j)$ $(j=1,2)$, respectively. We also fix a map
 $\hat{s}(\cdot):\spin\to \{1,-1\}$. Let us accept that a site of $\Z^2$ is identified as the corresponding
 site of $\G$ whenever we consider a problem in $\G$.  
For $\bx=(x_1,x_2,\cdots,x_n)$, $\by=(y_1,y_2,\cdots,y_n)\in \C^n$,
$\<\bx,\by\>:=\sum_{j=1}^nx_jy_j$,
$\<\bx,\by\>_{\C^n}:=\sum_{j=1}^nx_j\overline{y_j}$ and
$\|\bx\|_{\C^n}:=\sqrt{\<\bx,\bx\>_{\C^n}}$. 
For $x\in \R$, $\lfloor x\rfloor$ denotes the largest integer
which does not exceed $x$.
Let $1_P:=1$ if the proposition $P$ is true, $1_P:=0$ otherwise. 
For any subset $\cO$ of a topological space let $\cO^i$ denote the
interior of $\cO$. 
Let $\S_n$ be the set of all permutations over
$\{1,2,\cdots,n\}$ $(n\in\N)$. 
It will be convenient to use the
function $\cF_{t,\beta}(\cdot):\R\to\R$ defined by
$$
\cF_{t,\beta}(x):=\frac{1}{2}\sinh^{-1}\left(\frac{x\pi^2}{8\max\{1,t^2\}\max\{\beta,\beta^2\}}\right).
$$
Here recall that $\sinh^{-1}(x)=\log(x+\sqrt{x^2+1})$.

The correlation function will be formulated as a limit of
Grassmann integration over a finite dimensional Grassmann algebra. The
reduction to the finite dimensional problem is done by discretizing the
integrals over the interval $[0,\beta)$ in the perturbative expansion of the partition function. For this purpose, take a parameter $h\in 2\N/\beta$ and set $[0,\beta)_h:=
\{0,1/h,2/h,\cdots,\beta-1/h\}$, 
$[-\beta,\beta)_h:= \{-\beta,-\beta+1/h,\cdots,-1/h\}\cup
[0,\beta)_h$. Note that $\sharp [0,\beta)_h=\beta h$, $\sharp
[-\beta,\beta)_h=2\beta h$. We have seen in \cite[\mbox{Appendix
C}]{K1} that taking the parameter $h$ from $2\N/\beta$ rather than from
$\N/\beta$ is convenient for the discretization of $[0,\beta)$ and 
$[-\beta,\beta)$. Set
$I_{L,h}:=\{1,2,3\}\times\G\times\spin\times[0,\beta)_h$ 
and $N_{L,h}:=\sharp I_{L,h}=6L^2\beta h$. We define the lattice of
the momentum variable $\G^*$ and the subset of the Matsubara frequency
$\cM_h$ by $\G^*:=(\frac{2\pi}{L}\Z/(2\pi\Z))^2$ and $\cM_h:=\{\o\in
\pi(2\Z+1)/\beta\ |\ |\o|<\pi h\}$.
 
\subsection{The Grassmann Gaussian integral}\label{subsec_grassmann_integral}
Here let us summarize the notion of Grassmann Gaussian integral. For a finite 
dimensional complex vector space $W$ and $n\in\N$, let $\bigwedge^nW$
denote the $n$-fold anti-symmetric tensor product of $W$ and $\bigwedge^0W:=\C$. Moreover, set $\bigwedge
W:=\bigoplus_{n=0}^{\dim W}\bigwedge^nW$. 

Let $\cV$, $\cV^+$, $\cV^-$, $\cV_p$ $(p\in \N)$ be the complex vector
spaces spanned by the basis $\{\opsi_X,\psi_X\}_{X\in I_{L,h}}$,
$\{\opsi_X\}_{X\in I_{L,h}}$, $\{\psi_X\}_{X\in I_{L,h}}$,
$\{\opsi_X^p,\psi_X^p\}_{X\in I_{L,h}}$ $(p\in\N)$, respectively. 
This paper concerns various problems formulated in the Grassmann
algebras $\bigwedge \cV$, $\bigwedge \cV^+$, $\bigwedge \cV^-$,
$\bigwedge \cV_p$ $(p\in \N)$. Remark that there is a vector space
isomorphism between $\bigwedge \cV$ and $\left(\bigwedge \cV^+\right)\otimes \left(\bigwedge \cV^-\right)$,
the tensor product of $\bigwedge \cV^+$ and $\bigwedge \cV^-$. Then,
let $\cP_n:\bigwedge \cV\to \left(\bigwedge^n\cV^+\right)\otimes
\left(\bigwedge^n\cV^-\right)$ denote the standard projection $(
n\in \{0,1,\cdots,N_{L,h}\})$. 

Let us give a number from 1 to $N_{L,h}$ to each element of $I_{L,h}$ so
that we can write $I_{L,h}=\{X_{o,j}\}_{j=1}^{N_{L,h}}$. Set
$\psi:=(\opsi_{X_{o,1}},\cdots,\opsi_{X_{o,N_{L,h}}},\psi_{X_{o,1}},\cdots,\psi_{X_{o,N_{L,h}}})$,
$\psi^p:=(\opsi_{X_{o,1}}^p,\cdots,\opsi_{X_{o,N_{L,h}}}^p,\psi_{X_{o,1}}^p,\cdots,\psi_{X_{o,N_{L,h}}}^p)$
$(p\in \N)$. Take $p$, $q_1,\cdots,q_n\in\N$ with $p\neq
q_j$ $(\forall j\in\{1,\cdots,n\})$. The Grassmann Gaussian integral
$\int\cdot d\mu_C(\psi^p)$ with a covariance $(C(X,Y))_{X,Y\in I_{L,h}}$
is a linear map from $\bigwedge \left(\left(\bigoplus_{j=1}^n
\cV_{q_j}\right)\bigoplus\cV_p\right)$ to $\bigwedge
\left(\bigoplus_{j=1}^n\cV_{q_j}\right)$ defined as follows. For 
$f\in \bigwedge \left(\bigoplus_{j=1}^n\cV_{q_j}\right)$ and
$X_1,\cdots,X_a,Y_1,\cdots,Y_b\in I_{L,h}$,
 \begin{align*}
&\int f d\mu_{C}(\psi^p):=f,\\
&\int f\opsi_{X_1}^p\cdots
 \opsi_{X_a}^p\psi_{Y_b}^p\cdots\psi_{Y_1}^pd\mu_C(\psi^p):=\left\{\begin{array}{ll}\det(C(X_j,Y_k))_{1\le j,k\le a}f &\text{ if }a=b,\\
0 &\text{ if }a\neq b.
\end{array}\right.
\end{align*}
Then for any $g\in\bigwedge \left(\left(\bigoplus_{j=1}^n
\cV_{q_j}\right)\bigoplus\cV_p\right)$, $\int gd\mu_C(\psi^p)$ can be
defined by linearity and anti-symmetry. 

Though it is not used during the formulation in this section, let us
recall the notion of left derivative at this stage for later use. For $X'\in I_{L,h}$ the left
derivative $\partial/ \partial \psi_{X'}^p$ is a linear operator on
$\bigwedge \left(\left(\bigoplus_{j=1}^n
\cV_{q_j}\right)\bigoplus\cV_p\right)$. By letting $\cV_p'$ be the vector
space with the basis $\{\opsi_X^p,\psi_X^p\}_{X\in I_{L,h}}\backslash
\{\psi_{X'}^p\}$,
$$
\frac{\partial}{\partial \psi_{X'}^p}(f\psi_{X'}^pg):=(-1)^mfg,\quad \frac{\partial}{\partial \psi_{X'}^p}g:=0,
$$
for $f\in \bigwedge^m \left(\left(\bigoplus_{j=1}^n
\cV_{q_j}\right)\bigoplus\cV_p'\right)$ $(m\in\N\cup\{0\})$, $g\in \bigwedge \left(\left(\bigoplus_{j=1}^n
\cV_{q_j}\right)\bigoplus\cV_p'\right)$. Then, $(\partial/\partial
\psi_{X'}^p)g$ can be defined for any $g\in \bigwedge \left(\left(\bigoplus_{j=1}^n
\cV_{q_j}\right)\bigoplus\cV_p\right)$ by linearity. The definition of
the left derivative $\partial/\partial
\opsi_{X'}^p$ is parallel to that of $\partial/\partial\psi_{X'}^p$.

\subsection{The covariance}\label{subsec_covariance}
In our formulation the covariance is given as a 2-point correlation
function governed by the free Hamiltonian $H_0$. For
$(\rho,\bx,\s,x),(\eta,\by,\tau,y)\in \{1,2,3\}\times \G\times \spin \times [0,\beta)$,
$$
\cC(\rho\bx\s x, \eta \by \tau y):=\frac{\Tr(e^{-\beta
H_0}T(\psi_{\rho\bx\s}^*(x)\psi_{\eta\by\tau}(y)))}{\Tr e^{-\beta H_0}},
$$
where $\psi_{\rho\bx\s}^*(x):=e^{xH_0}\psi_{\rho\bx\s}^*e^{-xH_0}$,
$\psi_{\eta\by\tau}(y):=e^{yH_0}\psi_{\eta\by\tau}e^{-yH_0}$,
$T(\psi_{\rho\bx\s}^*(x)\psi_{\eta\by\tau}(y)):=1_{x\ge
y}\psi_{\rho\bx\s}^*(x)\psi_{\eta\by\tau}(y)-1_{x<y}\psi_{\eta\by\tau}(y)\psi_{\rho\bx\s}^*(x)$. 

The following characterization of $\cC$ is done in Appendix
\ref{app_covariance}. For any $(\rho,\bx,\s,x)$, $(\eta,\by,\tau,y)\in
I_{L,h}$,
\begin{equation}\label{eq_covariance_characterization}
\cC(\rho\bx\s x, \eta \by \tau y) =\frac{\delta_{\s,\tau}}{\beta
 L^2}\sum_{(\bk,\o)\in\G^*\times\cM_h}e^{-i\<\bx-\by,\bk\>}e^{i(x-y)\o}\cB_{\rho,\eta}^{\s}(\bk,\o),
\end{equation}
where for $\bk=(k_1,k_2)\in\G^*$, $\o\in\cM_h$, $\s\in\spin$,
\begin{align}
&\left(\cB_{\rho,\eta}^{\s}(\bk,\o)\right)_{1\le \rho,\eta\le 3}:=\notag\\
&\quad\left(\begin{array}{ccc} \frac{\cN_{1,1}^{\s}(\bk,\o)}{\cD^{\s}(\bk,\o)}
 & \frac{\cN_{1,2}^{\s}(\bk,\o)}{\cD^{\s}(\bk,\o)} &
  \frac{\cN_{1,3}^{\s}(\bk,\o)}{\cD^{\s}(\bk,\o)} \\
\frac{\cN_{2,1}^{\s}(\bk,\o)}{\cD^{\s}(\bk,\o)} &
 \frac{1}{h(1-e^{-i\o/h+\eo/h})}\left(1+
 \frac{\cN_{2,2}^{\s}(\bk,\o)}{\cD^{\s}(\bk,\o)}\right) &
 \frac{\cN_{2,3}^{\s}(\bk,\o)}{h(1-e^{-i\o/h+\eo/h})\cD^{\s}(\bk,\o)}\\
\frac{\cN_{3,1}^{\s}(\bk,\o)}{\cD^{\s}(\bk,\o)} &
 \frac{\cN_{3,2}^{\s}(\bk,\o)}{h(1-e^{-i\o/h+\eo/h})\cD^{\s}(\bk,\o)} &
\frac{1}{h(1-e^{-i\o/h+\eo/h})}\left(1+
 \frac{\cN_{3,3}^{\s}(\bk,\o)}{\cD^{\s}(\bk,\o)}\right)\end{array}
\right),\notag\\
&\cD^{\s}(\bk,\o):=h^2\left(1-e^{-\frac{i}{h}\o+\frac{1}{2h}(\ec+\eo)}\right)^2\notag\\
&\qquad-e^{-\frac{i}{h}\o+\frac{1}{2h}(\ec+\eo)}\left(\frac{(\ec-\eo)^2}{4}+2t^2\sum_{j=1}^2(1+\cos
 k_j)\right)+e^{-\frac{i}{h}\o+\frac{1}{2h}(\ec+\eo)}O_1^{\s}(\bk),\notag\\
&\cN_{1,1}^{\s}(\bk,\o):=h\left(1-e^{-\frac{i}{h}\o+\frac{1}{2h}(\ec+\eo)}\right)+\frac{\ec-\eo}{2}e^{-\frac{i}{h}\o+\frac{1}{2h}(\ec+\eo)}+e^{-\frac{i}{h}\o+\frac{1}{2h}(\ec+\eo)}O_2^{\s}(\bk),\notag\\
&\cN_{1,2}^{\s}(\bk,\o):=t(1+e^{ik_1})e^{-\frac{i}{h}\o+\frac{1}{2h}(\ec+\eo)}(1+O_3^{\s}(\bk)),\quad
 \cN_{1,3}^{\s}(\bk,\o):=\cN_{1,2}^{\s}((k_2,k_1),\o),\notag\\
&\cN_{2,1}^{\s}(\bk,\o):=\cN_{1,2}^{\s}(-\bk,\o),\notag\\
&\cN_{2,2}^{\s}(\bk,\o):=2t^2(1+\cos
 k_1)\Big(\frac{1}{2}e^{-\frac{i}{h}\o+\frac{1}{2h}(\ec+\eo)}+\frac{1}{2}e^{-\frac{2i}{h}\o+\frac{1}{2h}(\ec+3\eo)}\notag\\
&\qquad\qquad\qquad+\left(e^{-\frac{i}{h}\o+\frac{1}{2h}(\ec+\eo)}+e^{-\frac{2i}{h}\o+\frac{1}{2h}(\ec+3\eo)}\right)O_4^{\s}(\bk)\notag\\
&\qquad\qquad\qquad+\left(e^{-\frac{i}{h}\o+\frac{1}{2h}(\ec+\eo)}-e^{-\frac{2i}{h}\o+\frac{1}{2h}(\ec+3\eo)}\right)O_5^{\s}(\bk)\Big),\notag\\
&\cN_{2,3}^{\s}(\bk,\o):=t^2(1+e^{-ik_1})(1+e^{ik_2})\Big(\frac{1}{2}e^{-\frac{i}{h}\o+\frac{1}{2h}(\ec+\eo)}+\frac{1}{2}e^{-\frac{2i}{h}\o+\frac{1}{2h}(\ec+3\eo)}\notag\\
&\qquad\qquad\qquad+\left(e^{-\frac{i}{h}\o+\frac{1}{2h}(\ec+\eo)}+e^{-\frac{2i}{h}\o+\frac{1}{2h}(\ec+3\eo)}\right)O_4^{\s}(\bk)\notag\\
&\qquad\qquad\qquad+\left(e^{-\frac{i}{h}\o+\frac{1}{2h}(\ec+\eo)}-e^{-\frac{2i}{h}\o+\frac{1}{2h}(\ec+3\eo)}\right)O_5^{\s}(\bk)\Big),\notag\\
&\cN_{3,1}^{\s}(\bk,\o):=\cN_{1,2}^{\s}(-(k_2,k_1),\o),\quad
 \cN_{3,2}^{\s}(\bk,\o):=\cN_{2,3}^{\s}(-\bk,\o),\notag\\
&\cN_{3,3}^{\s}(\bk,\o):=\cN_{2,2}^{\s}((k_2,k_1),\o).\label{eq_covariance_inside_function}
\end{align}
The functions $O_j^{\s}(\cdot):\C^2\to \C$ $(j\in\{1,\cdots,5\},\s\in\spin)$ are entirely analytic and
satisfy that $O_j^{\s}(\bk+2\pi m\be_1+2\pi n\be_2)=O_j^{\s}(\bk)$
$(\forall \bk\in\C^2,m,n\in\Z)$. Moreover, for any compact set $K\subset
\C^2$,
\begin{equation}\label{eq_decay_extra_terms}
\sup_{\bk\in K,j\in\{1,\cdots,5\},\s\in\spin}|O_j^{\s}(\bk)|\le \frac{C_{K,E_{max}}}{h},
\end{equation}
where $C_{K,E_{max}}$ is a positive constant depending only on $K$ and
$E_{max}$. Though these information about $O_j^{\s}$ are sufficient
for our analysis to proceed, the functions $O_j^{\s}$
 are made explicit in \eqref{eq_vanishing_functions} in Appendix
 \ref{app_covariance}.
\begin{remark}
The functions $\cD^{\s}(\bk,\o)$, $\cN_{\rho,\eta}^{\s}(\bk,\o)$
 $(\rho,\eta\in\{1,2,3\},\s\in\spin)$ are analytic with respect to $\bk$. This
 property is one essential requirement of our method to prove 
 exponential decay of the correlation functions. As shown in Appendix
 \ref{app_covariance}, in the preliminary form before being expanded over $\cM_h$ the covariance $\cC(X,Y)$ contains a square root
 of $(\ec-\eo)^2/t^2+8\sum_{j=1}^2(1+\cos k_j)$, which is not
 analytic. In order to make the analyticity with $\bk$ apparent, 
we choose to transform the covariance into the sum over $\G^*\times \cM_h$.
\end{remark}
\begin{remark}
The dispersion relation for the free particle hopping to the nearest
 neighbor sites on the CuO lattice is given by \eqref{eq_eigen_values}
 in Appendix \ref{app_covariance}. As discussed in Remark
 \ref{rem_half_filled}, taking $\ec$, $\eo$ to be $-U_c/2$, $-U_o/2$
 respectively makes the system half-filled. If we shift the on-site
 quadratic term to the interacting part of the Hamiltonian, one of the
 dispersion relation denoted by $A_1^{\s}(t,\bk)$ in
 \eqref{eq_eigen_values} is changed into $0$. 
The formulation including the quadratic term in the interacting part is
 parallel to the formulation of the half-filled honeycomb lattice model
 in \cite{GM}, though in \cite{GM} the quadratic term is eventually
 erased by the non-corresponding property of the covariance at equal
 space-time. One remarkable fact is that the zero set of the free
 particle dispersion relation in the half-filled formulation of the CuO
 Hubbard model is, thus, the whole momentum space, while that is the contour 
of a square in the half-filled Hubbard model on the square lattice (see, e.g,
 \cite{R}) and that consists of 2 distinct points in the half-filled
 Hubbard model on the honeycomb lattice (see \cite{GM}). This suggests
 that trying to improve the temperature dependency of the convergence
 theory in the half-filled CuO Hubbard model would require a qualitatively
 different method from the infrared integration regimes for the
 half-filled 2D Hubbard model developed so far, in which the
 degeneracy of the zero set of the dispersion relation is crucial. 
 \end{remark}  

\subsection{The Grassmann integral formulation}\label{subsec_grassmann_formulation}
In order to relate the correlation function to the Grassmann
Gaussian integral, we introduce parameters $\la_{1}$,
$\la_{-1}\in\C$ and define
$U_{(\la_1,\la_{-1})}(\cdot,\cdot,\cdot,\cdot):(\{1,2,3\}\times\G\times\spin)^4\to\C$ by
\begin{align}
&U_{(\la_1,\la_{-1})}(\rho_1\bx_1\s_1,\rho_2\bx_2\s_2,\eta_1\by_1\tau_1,\eta_2\by_2\tau_2)\notag\\
&:=\frac{1}{4}(1_{(\s_1,\s_2)=(\ua,\da)}-1_{(\s_1,\s_2)=(\da,\ua)})
(1_{(\tau_1,\tau_2)=(\da,\ua)}-1_{(\tau_1,\tau_2)=(\ua,\da)})1_{\bx_1=\bx_2=\by_1=\by_2}\notag\\
&\qquad\cdot(U_c1_{\rho_1=\rho_2=\eta_1=\eta_2=1}+U_o1_{\rho_1=\rho_2=\eta_1=\eta_2=2\text{
 or }3})\notag\\
&\quad+\frac{1}{4}\la_1(1_{(\rho_1\bx_1\s_1,\rho_2\bx_2\s_2)=(\hcX_1,\hcX_2)}-1_{(\rho_1\bx_1\s_1,\rho_2\bx_2\s_2)=(\hcX_2,\hcX_1)})\notag\\
&\qquad\cdot(1_{(\eta_1\by_1\tau_1,\eta_2\by_2\tau_2)=(\hcY_2,\hcY_1)}-1_{(\eta_1\by_1\tau_1,\eta_2\by_2\tau_2)=(\hcY_1,\hcY_2)})\notag\\
&\quad+\frac{1}{4}\la_{-1}(1_{(\rho_1\bx_1\s_1,\rho_2\bx_2\s_2)=(\hcY_1,\hcY_2)}-1_{(\rho_1\bx_1\s_1,\rho_2\bx_2\s_2)=(\hcY_2,\hcY_1)})\notag\\
&\qquad\cdot(1_{(\eta_1\by_1\tau_1,\eta_2\by_2\tau_2)=(\hcX_2,\hcX_1)}-1_{(\eta_1\by_1\tau_1,\eta_2\by_2\tau_2)=(\hcX_1,\hcX_2)}).\label{eq_original_bi_anti_symmetric}
\end{align}
For another application in Section \ref{sec_multiscale_integration} we
purposely defined $U_{(\la_1,\la_{-1})}(\cdot,\cdot,\cdot,\cdot)$ to
satisfy
$U_{(\la_1,\la_{-1})}(X_2,X_1,Y_1,Y_2)=U_{(\la_1,\la_{-1})}(X_1,X_2,Y_2,Y_1)=-U_{(\la_1,\la_{-1})}(X_1,X_2,Y_1,Y_2)$.
Define the Grassmann polynomial
$V_{(\la_1,\la_{-1})}(\psi)\in\bigwedge \cV$ by
\begin{equation*}
V_{(\la_1,\la_{-1})}(\psi):=-\frac{1}{h}\sum_{x\in
 [0,\beta)_h}\sum_{X_1,X_2,Y_1,Y_2\atop\in\{1,2,3\}\times\G\times\spin}U_{(\la_1,\la_{-1})}(X_1,X_2,Y_1,Y_2)\opsi_{X_1x}\opsi_{X_2x}\psi_{Y_1x}\psi_{Y_2x}.
\end{equation*}

The Grassmann integral formulation of the correlation function is
summarized as follows.
\begin{lemma}\label{lem_grassmann_formulation}
\begin{enumerate}[(i)]
\item\label{item_real_part} For any $U>0$ there exists
     $N_U\in\N$ such that $\Re\int
     e^{V_{(\la,\la)}(\psi)}d\mu_{\cC}(\psi)$ $>0$ for any $h\in 2\N/\beta$
     with $h\ge 2N_U/\beta$, $\la,U_c,U_o\in\R$ with
     $|\la|,|U_c|,|U_o|\le U$.
\item\label{item_grassmann_formulation}
\begin{equation*}
\<\psi_{\hcX_1}^*\psi_{\hcX_2}^*\psi_{\hcY_2}\psi_{\hcY_1}+\text{h.c}\>_L=-\frac{1}{\beta}\lim_{h\to\infty\atop
 h\in 2\N/\beta}\frac{\partial}{\partial\la}\log\left(\int
						 e^{V_{(\la,\la)}(\psi)}d\mu_{\cC}(\psi)\right)\Big|_{\la=0},
\end{equation*}
where for $z\in\C$ with $\Re z>0$, $\log z:=\log |z|+i\Arg z$, $\Arg
     z\in (-\pi/2,\pi/2)$.
\end{enumerate}
\end{lemma}
\noindent
Lemma \ref{lem_grassmann_formulation} can be proved in a way similar to 
\cite[\mbox{Section 3}]{K2}. For the readers' convenience we
outline the proof in Appendix \ref{app_formulation_convergence}.

The analysis in the following sections treats the perturbed covariance
containing complex momentum variables inside. For $\bp\in\C^2$, 
$$
\cC(\rho\bx\s x,\eta\by\tau y)(\bp):=\frac{\delta_{\s,\tau}}{\beta L^2}\sum_{(\bk,\o)\in\G^*\times\cM_h}e^{-i\<\bx-\by,\bk\>}e^{i(x-y)\o}\cB_{\rho,\eta}^{\s}(\bk+\hat{s}(\s)\bp,\o),
$$
$\cC(\bp):=(\cC(X,Y)(\bp))_{X,Y\in
I_{L,h}}$. By admitting a few facts proved in Section
\ref{sec_preliminaries}, we can show the next lemma. The equality in
Lemma \ref{lem_contour_integral_formulation}
\eqref{item_contour_integral_transform} will be estimated in Section
\ref{sec_multiscale_integration} as the main objective.

\begin{lemma}\label{lem_contour_integral_formulation} For any $L\in\N$,
 $R\in (\cF_{t,\beta}(8/\pi^2),\infty)$, $\eps\in (8/\pi^2,1)$ and sufficiently large $h\in 2\N/\beta$ there exists $U_{small}>0$ such that the following statements hold
 true.
\begin{enumerate}[(i)]
\item\label{item_perturbed_real_part}
$\Re \int e^{V_{(\la_1,\la_{-1})}(\psi)}d\mu_{\cC(w\be_p)}(\psi)>0$ 
 for any $p\in \{1,2\}$ and all $(\la_1,\la_{-1},U_c,U_o,w)\in \C^5$
     with $|\la_1|$, $|\la_{-1}|$, $|U_c|$, $|U_o|\le U_{small}$, $|\Re
     w|\le R$, $|\Im w|\le \cF_{t,\beta}(\eps)$.
\item\label{item_perturbed_analyticity} For any $p\in\{1,2\}$ the
     function 
$$(\la_1,\la_{-1},U_c,U_o,w)\mapsto \log\left(\int
     e^{V_{(\la_1,\la_{-1})}(\psi)}d\mu_{\cC(w\be_p)}(\psi)\right)$$
is analytic in
$$
\left\{(\la_1,\la_{-1},U_c,U_o,w)\in\C^5\ \Big|\ \begin{array}{l}|\la_1|, |\la_{-1}|, |U_c|,
     |U_o|< U_{small},\\
|\Re w|<R,|\Im w|<\cF_{t,\beta}(\eps)\end{array}
\right\}.
$$ 
\item\label{item_contour_integral_transform}
For any $n\in \N$ with $2\pi n/L+\cF_{t,\beta}(8/\pi^2)<R$, $U_c,U_o\in \C$ with $|U_c|,|U_o|<U_{small}$ and $p\in\{1,2\}$,
\begin{align*}
&\left(\frac{L}{2\pi}\left(e^{i\frac{2\pi}{L}\<\sum_{j=1}^2(\hat{s}(\hs_j)\hbx_j-\hat{s}(\htau_j)\hby_j),\be_p\>}-1\right)\right)^n\frac{\partial}{\partial\la}\log\left(\int
     e^{V_{(\la,\la)}(\psi)}d\mu_{\cC}(\psi)\right)\Big|_{\la=0}\\
&=\sum_{a\in \{1,-1\}}\prod_{j=1}^n\left(\frac{L}{2\pi}\int_0^{2\pi
 a/L}d\theta_{a,j}\frac{1}{2\pi
 i}\oint_{|w_{a,j}-\theta_{a,j}|=\cF_{t,\beta}(8/\pi^2)/n}dw_{a,j}\frac{1}{(w_{a,j}-\theta_{a,j})^2}\right)\\
&\quad\cdot \frac{\partial}{\partial\la_a}\log\left(\int
     e^{V_{(\la_1,\la_{-1})}(\psi)}d\mu_{\cC\left(\sum_{j=1}^nw_{a,j}\be_p\right)}(\psi)\right)\Big|_{\la_1=\la_{-1}=0},
\end{align*}
where $\oint_{|w_{a,j}-\theta_{a,j}|=\cF_{t,\beta}(8/\pi^2)/n}dw_{a,j}$ represents the
     contour integral along the contour $\{w_{a,j}\in\C\ |\
     |w_{a,j}-\theta_{a,j}|=\cF_{t,\beta}(8/\pi^2)/n\}$ oriented counter clock-wise.
\end{enumerate}
\end{lemma}
\begin{proof}
\eqref{item_perturbed_real_part}: It follows from Lemma
 \ref{lem_covariance_properties_0} \eqref{item_covariance_analyticity_0},
 Lemma \ref{lem_covariance_properties_l}
 \eqref{item_covariance_analyticity_l} and \eqref{eq_covariance_sum_up}
that the function $w\mapsto
 \cC(X,Y)(w\be_p)$ is analytic in $\{w\in\C\ |\ |\Im w|<\cF_{t,\beta}(\eps')\}$ for
 any $\eps' \in (0,1)$, sufficiently large $h\in 2\N/\beta$
 and $X,Y\in I_{L,h}$. Thus, for any fixed large
 $h\in2\N/\beta$, $|\cC(X,Y)(w\be_p)|$ is uniformly bounded with respect to
 $X,Y\in I_{L,h}$ and $w\in\C$ with $|\Re w|\le R$, $|\Im w|\le
 \cF_{t,\beta}(\eps)$. Note that by definition $\int
 e^{V_{(\la_1,\la_{-1})}(\psi)}d\mu_{\cC(w\be_p)}(\psi)$ is a polynomial
 of $\la_1$, $\la_{-1}$, $U_c$, $U_o$, whose constant term is 1 and
 higher order terms have finite sums and products of
 $\cC(X,Y)(w\be_p)$ $(X,Y\in I_{L,h})$ in their coefficients. Thus, the
 uniform boundedness of $\cC(X,Y)(w\be_p)$ ensures that
$$
\lim_{U\searrow 0}\sup_{(\la_1,\la_{-1},U_c,U_o,w)\in \C^5 \atop
|\la_1|,|\la_{-1}|,|U_c|,|U_o|\le U, |\Re w|\le
 R,|\Im w|\le \cF_{t,\beta}(\eps)}\left|\int e^{V_{(\la_1,\la_{-1})}(\psi)}d\mu_{\cC(w\be_p)(\psi)}-1\right|=0,
$$
which implies the claim \eqref{item_perturbed_real_part}.

\eqref{item_perturbed_analyticity}: The claim
 \eqref{item_perturbed_real_part} and the analyticity of
 $\cC(X,Y)(w\be_p)$ with respect to $w$ verify the statement.

\eqref{item_contour_integral_transform}: Set
\begin{align*}
&S_1(\cC):=-\frac{1}{h}\sum_{x\in [0,\beta)_h}\int
 \opsi_{\hcX_1x}\opsi_{\hcX_2x}\psi_{\hcY_2x}\psi_{\hcY_1x}e^{V_{(0,0)}(\psi)}d\mu_{\cC}(\psi)\Big/\int
 e^{V_{(0,0)}(\psi)}d\mu_{\cC}(\psi),\\
&S_{-1}(\cC):=-\frac{1}{h}\sum_{x\in [0,\beta)_h}\int
 \opsi_{\hcY_1x}\opsi_{\hcY_2x}\psi_{\hcX_2x}\psi_{\hcX_1x}e^{V_{(0,0)}(\psi)}d\mu_{\cC}(\psi)\Big/\int
 e^{V_{(0,0)}(\psi)}d\mu_{\cC}(\psi).
\end{align*}
 Note the equality that 
\begin{align*}
&e^{ai\frac{2\pi}{L}\<\sum_{j=1}^n(\hat{s}(\s_j)\bx_j-\hat{s}(\tau_j)\by_j),\be_p\>}\det\left(\cC(\rho_j\bx_j\s_jx_j,\eta_k\by_k\tau_ky_k)\right)_{1\le
 j,k\le
 n}\\
&=\det\left(\cC(\rho_j\bx_j\s_jx_j,\eta_k\by_k\tau_ky_k)\left(a\frac{2\pi}{L}\be_p\right)\right)_{1\le
 j,k\le n}\quad (\forall a\in\{1,-1\})
\end{align*}
and the fact that $V_{(0,0)}(\psi)$ is invariant under the scaling
 $\opsi_{\rho\bx\s x}\to e^{ia
 \hat{s}(\s)\frac{2\pi}{L}\<\bx,\be_p\>}\opsi_{\rho\bx\s x}$, $\psi_{\rho\bx\s
 x}\to e^{-ia \hat{s}(\s)\frac{2\pi}{L}\<\bx,\be_p\>}\psi_{\rho\bx\s x}$
 $(a\in\{1,-1\},(\rho,\bx,\s,x)\in I_{L,h})$. Then, by remarking
  the definition of the Grassmann
 Gaussian integral and the claim \eqref{item_perturbed_analyticity} we
 can justify the following transformations.
\begin{align*}
&e^{i\frac{2\pi}{L}\<\sum_{j=1}^2(\hat{s}(\hs_j)\hbx_j-\hat{s}(\htau_j)\hby_j),\be_p\>}
\frac{\partial}{\partial\la}\log\left(\int
     e^{V_{(\la,\la)}(\psi)}d\mu_{\cC}(\psi)\right)\Big|_{\la=0}\\
&=\sum_{a\in\{1,-1\}}e^{i\frac{2\pi}{L}\<\sum_{j=1}^2(\hat{s}(\hs_j)\hbx_j-\hat{s}(\htau_j)\hby_j),\be_p\>}S_a(\cC)=\sum_{a\in\{1,-1\}}S_a\left(\cC\left(a\frac{2\pi}{L}\be_p\right)\right).\\
&\frac{L}{2\pi}\left(e^{i\frac{2\pi}{L}\<\sum_{j=1}^2(\hat{s}(\hs_j)\hbx_j-\hat{s}(\htau_j)\hby_j),\be_p\>}-1\right)\frac{\partial}{\partial\la}\log\left(\int
     e^{V_{(\la,\la)}(\psi)}d\mu_{\cC}(\psi)\right)\Big|_{\la=0}\\
&=
 \sum_{a\in\{1,-1\}}\frac{L}{2\pi}\left(S_a\left(\cC\left(a\frac{2\pi}{L}\be_p\right)\right)-S_a(\cC(\b0))\right)\\
&=
 \sum_{a\in\{1,-1\}}\frac{L}{2\pi}\int_{0}^{2\pi a/L}d\theta_a\frac{d}{d\theta_a}S_a\left(\cC\left(\theta_a\be_p\right)\right)\\
&=
 \sum_{a\in\{1,-1\}}\frac{L}{2\pi}\int_{0}^{2\pi a/L}d\theta_a\frac{1}{2\pi i}\oint_{|w_a-\theta_a|=\cF_{t,\beta}(8/\pi^2)/n}dw_a\frac{1}{(w_a-\theta_a)^2}S_a(\cC(w_a\be_p)).
\end{align*}
Repeating this procedure $n$ times results in the equality claimed in \eqref{item_contour_integral_transform}.
\end{proof}

\section{Preliminaries}\label{sec_preliminaries}
In this section we show some lemmas concerning the cut-off function
and the sliced covariance, which are the necessary tools for the
forthcoming multi-scale analysis. To begin with, let us fix a function
$\phi\in C_0^{\infty}(\R)$ with the following properties. (i)
$\phi(x)=1$ if $|x|\le 1$. (ii) $\phi(x)=0$ if $|x|\ge 2$. (iii) $\phi(x)\in (0,1)$ if $1<|x|<2$ and is
strictly increasing in $(-2,-1)$, strictly decreasing in
$(1,2)$. See, e.g, \cite[\mbox{Problem II.6}]{FKT} for a concrete
construction of such a function. From now let the notation
`$c$' stand for a generic positive constant which depends only on $\phi$ and is
independent of any other parameters.

\subsection{The cut-off function}\label{subsec_cut_off}
With a parameter $M\in \R_{>2}$ define the function $\chi\in C_0^{\infty}(\R)$
by $\chi(x):=\phi((x-M)/(M^2-M)+1)$. We can see that $\chi(x)=1$
$(\forall x\in [0,M])$, $\chi(x)=0$ $(\forall x\in[M^2,\infty))$,
$\chi(x)\in (0,1)$ $(\forall x\in (M,M^2))$, $\chi(\cdot)$ is strictly
decreasing in $(M,M^2)$ and 
\begin{equation}\label{eq_chi_bound}
\left|\left(\frac{d}{dx}\right)^m\chi(x)\right|\le c M^{-2m}\quad
(\forall x\in[0,\infty),\forall m\in \{0,\cdots,4\}).
\end{equation}
In the next subsection $\chi$ will be differentiated at most 4
times. Thus, it suffices to prepare the bound \eqref{eq_chi_bound} only up
to $m=4$. 

Set $N_h:=\lfloor \log(2h)/\log(M)\rfloor$ for $h\in 2\N/\beta$,
$N_{\beta}:=\max\{\lfloor \log(1/\beta)/\log(M)\rfloor + 1, 1\}$.
For large $h\in 2\N/\beta$ satisfying $N_h\ge
N_{\beta}+1$ we have that
\begin{align}
&\frac{1}{\beta}<M^{N_{\beta}}\le
 \max\left\{1,\frac{1}{\beta}\right\} M,\label{eq_beta_M}\\
&M^l\le 2h\quad (\forall l \in
\{N_{\beta},N_{\beta}+1,\cdots,N_h\}).\label{eq_h_M}\end{align}

Define the functions $\chi_l(\cdot):\R\to\R$ ($l\in \{N_{\beta},N_{\beta}+1,\cdots,N_h\}$) by
\begin{align*}
&\chi_{N_{\beta}}(\o):=\chi(M^{-N_{\beta}}h|1-e^{i\o/h}|),\\
&\chi_l(\o):=\chi(M^{-l}h|1-e^{i\o/h}|)-\chi(M^{-(l-1)}h|1-e^{i\o/h}|)\quad
 (\forall l\in\{N_{\beta}+1,\cdots,N_h\}).
\end{align*}
Since $h|1-e^{i\o/h}|\le 2h\le M^{N_h+1}$,
$\chi(M^{-N_h}h|1-e^{i\o/h}|)=1$ $(\forall \o\in\R)$. This
implies that
\begin{equation}\label{eq_decomposition_cutoff}
\sum_{l=N_{\beta}}^{N_h}\chi_l(\o)=1\quad (\forall \o\in\R).
\end{equation}
The support property of these functions is described as follows. For any
$\o\in\R$,
\begin{align*}
&\chi_{N_{\beta}}(\o)=\left\{\begin{array}{ll}1 &\text{ if
		      }h|1-e^{i\o/h}|\le M^{N_{\beta}+1},\\
\in (0,1)&\text{ if }M^{N_{\beta}+1}<h|1-e^{i\o/h}|<M^{N_{\beta}+2},\\
0&\text{ if }h|1-e^{i\o/h}|\ge M^{N_{\beta}+2},\end{array}
\right.\\
&\chi_{l}(\o)=\left\{\begin{array}{ll}0 &\text{ if
		      }h|1-e^{i\o/h}|\le M^{l},\\
\in (0,1]&\text{ if }M^{l}<h|1-e^{i\o/h}|<M^{l+2},\\
0&\text{ if }h|1-e^{i\o/h}|\ge M^{l+2},\end{array}
\right.(\forall l\in \{N_{\beta}+1,\cdots,N_h\}).
\end{align*}
The role of $\chi_l(\cdot)$ is a cut-off in the Matsubara frequency. The
support of $\chi_l(\cdot)$ can be estimated as follows.
\begin{lemma}\label{lem_matsubara_count}
For any $l\in \{N_{\beta},N_{\beta}+1,\cdots,N_h\}$,
$\frac{1}{\beta}\sum_{\o\in \cM_h}1_{\chi_l(\o)\neq 0}\le c M^{l+2}$.
\end{lemma}

\subsection{Properties of the sliced covariances}\label{subsec_sliced_covariance}
By using the cut-off function $\chi_l$ we define the covariance $\cC_l$
of $l$-th scale $(l\in\{N_{\beta},N_{\beta}+1,\cdots,N_h\})$ by
$$
\cC_l(\rho\bx\s x,\eta\by\tau y)(\bp):=\frac{\delta_{\s,\tau}}{\beta
L^2}\sum_{(\bk,\o)\in\G^*\times\cM_h}e^{-i\<\bx-\by,\bk\>}e^{i(x-y)\o}\chi_l(\o)\cB_{\rho,\eta}^{\s}(\bk+\hat{s}(\s)\bp,\o)
$$
for $(\rho,\bx,\s,x),(\eta,\by,\tau,y)\in I_{L,h}$, $\bp\in\C^2$.
Let $\cC_l(\bp):=(\cC_l(X,Y)(\bp))_{X,Y\in I_{L,h}}$. We will specify
a domain where $\cC_l(\cdot)$ is well-defined later in this
subsection. On such a domain the equality \eqref{eq_decomposition_cutoff} implies
that 
\begin{equation}\label{eq_covariance_sum_up}
\cC(\bp)=\sum_{l=N_{\beta}}^{N_h}\cC_l(\bp).
\end{equation}

In this
subsection we study various properties of $\cC_l$. For this purpose set
\begin{equation}\label{eq_first_def_E}
E(t,\bk):=2t^2\sum_{j=1}^2(1+\cos k_j):\C^2\to \C,
\end{equation}
and let us estimate $E(t,\bk)$, first of all.

\begin{lemma}\label{lem_properties_dispersion}
For any $\bk\in\R^2$, $j,p,q\in \{1,2\}$, $m\in\N\cup\{0\}$ and
 $w,z\in\C$ with $|\Im w|,|\Im z|\le r$,
\begin{align}
&\left|\left(\frac{\partial}{\partial
 k_j}\right)^mE(t,\bk+w\be_p+z\be_q)\right|\le 8t^2+8t^2\sinh(2r),\label{eq_derivative_dispersion}\\
&|\Im E(t,\bk+w\be_p+z\be_q)|\le
 4t^2\sinh(2r),\label{eq_imaginary_dispersion}\\
&\Re E(t,\bk+w\be_p+z\be_q)\ge -4t^2\sinh(2r).\label{eq_real_dispersion}
\end{align}
\end{lemma}
\begin{proof}
Note that
\begin{align*}
&E(t,\bk+w\be_p+z\be_q)\\
&=4t^2+2t^2\sum_{j=1}^2\cos(k_j+\Re w \delta_{p,j}+\Re z
 \delta_{q,j})\cosh(\Im w\delta_{p,j}+\Im z\delta_{q,j})\\
&\quad-i2t^2\sum_{j=1}^2\sin(k_j+\Re
 w\delta_{p,j}+\Re z \delta_{q,j})\sinh(\Im w\delta_{p,j}+\Im
 z\delta_{q,j}),
\end{align*}
which leads to $|E(t,\bk+w\be_p+z\be_q)|\le 8t^2+8t^2\sinh(2r)$. The
 upper bounds on $|(\partial/\partial k_j)^mE(\o,\bk+w\be_p+z\be_q)|$,
 $|\Im E(t,\bk+w\be_p+z\be_q)|$  can be obtained similarly. Moreover,
$\Re E(t,\bk+w\be_p+z\be_q)\ge 4t^2-4t^2\cosh(2r)\ge -4t^2\sinh(2r)$.
\end{proof}

The following lemma summarizes properties of $\cC_{N_{\beta}}$. The
$\beta$-dependency of Theorem \ref{thm_exponential_decay} in low
temperatures mainly stems from these upper bounds on $\cC_{N_{\beta}}$. 
From now we assume that
\begin{equation}\label{eq_first_condition_M}
M\ge 78E_{max}^2. 
\end{equation}

\begin{lemma}\label{lem_covariance_properties_0}
For any $\eps\in(0,1)$ there exists $N_{\eps}\in\N$ such that for any
 $h\in 2\N/\beta$ with $h\ge 2N_{\eps}/\beta$ the following statements
 hold true.
\begin{enumerate}[(i)]
\item\label{item_covariance_analyticity_0} The function $(w,z)\mapsto
     \cC_{N_{\beta}}(X,Y)(w\be_p+z\be_q)$ is analytic in
     $\{(w,z)\in\C^2\ |\ |\Im w|,$ $|\Im z|<\cF_{t,\beta}(\eps)\}$ for any $X,Y\in I_{L,h}$, $p,q\in\{1,2\}$.   
\item\label{item_covariance_L1_bound_0}
$$
\frac{1}{h}\sum_{(\bx,x)\in\G\times
     [-\beta,\beta)_h}|\cC_{N_{\beta}}(\rho\bx\s x, \eta \b0 \s
     0)(w\be_p)|\le \frac{c}{(1-\eps)\eps^2}M^{9-N_{\beta}}\max\{1,\beta\}^8
$$
for any $\rho,\eta\in \{1,2,3\}$, $\s\in\spin$, $p\in\{1,2\}$ and $w\in\C$ with $|\Im w|<
     \cF_{t,\beta}(\eps)$.
\item\label{item_covariance_determinant_bound_0}
$$
|\det(\<\bu_j,\bv_k\>_{\C^m}\cC_{N_{\beta}}(X_j,Y_k)(w\be_p))_{1\le
     j,k\le n}|\le \left(\frac{c}{1-\eps}M^6\max\{1,\beta\}^3\right)^n
$$
for any $m,n\in \N$, $\bu_j,\bv_j\in\C^m$ with
     $\|\bu_j\|_{\C^m},\|\bv_j\|_{\C^m}\le 1$, $X_j,Y_j\in I_{L,h}$
     $(j=1,\cdots,n)$, $p\in\{1,2\}$ and $w\in\C$ with $|\Im w|< \cF_{t,\beta}(\eps)$.
\end{enumerate}
\end{lemma}
\begin{proof}
First note that $\chi_{N_{\beta}}(\o)\neq 0$ implies $2|\o|/\pi\le
 M^{N_{\beta}+2}$, since $2|\theta|/\pi\le |1-e^{i\theta}|$ $(\forall
 \theta\in [-\pi,\pi])$. This inequality
 coupled with \eqref{eq_beta_M} proves
 that if $\chi_{N_{\beta}}(\o)\neq 0$,
\begin{equation}\label{eq_bound_omega}
|\o|\le c\max\left\{1,\frac{1}{\beta}\right\}M^3.
\end{equation}

\eqref{item_covariance_analyticity_0}: From the definition
 \eqref{eq_covariance_inside_function}, \eqref{eq_decay_extra_terms} and \eqref{eq_bound_omega} we
 observe that 
\begin{align*}
\cD^{\s}(\bk+w\be_p+z\be_q,\o)=&-\o^2+\ec\eo-\Re
 E(t,\bk+w\be_p+z\be_q)\\
&+i(-\o(\ec+\eo)-\Im
 E(t,\bk+w\be_p+z\be_q))+O(h^{-1}),
\end{align*}
where $O(h^{-1})$ represents terms of order $h^{-1}$. Moreover, if $|\Im w|$,
 $|\Im z|<r$, by \eqref{eq_imaginary_dispersion} and
 \eqref{eq_real_dispersion},
\begin{align*}
&|\cD^{\s}(\bk+w\be_p+z\be_q,\o)|\\
&\ge \max\{\o^2-\ec\eo+\Re
 E(t,\bk+w\be_p+z\be_q), |\o(\ec+\eo)|-|\Im
 E(t,\bk+w\be_p+z\be_q)|\}\\
&\quad+O(h^{-1})\\
&\ge
 \max\left\{\frac{\pi^2}{\beta^2}-\frac{1}{2}(\ec+\eo)^2-4t^2\sinh(2r), 
\frac{\pi}{\beta}|\ec+\eo|-4t^2\sinh(2r)\right\}+O(h^{-1})\\
&\ge 1_{|\ec+\eo|\le\frac{\pi}{\beta}}\left(\frac{\pi^2}{2\beta^2}-4t^2\sinh(2r)\right)+1_{|\ec+\eo|>
 \frac{\pi}{\beta}}\left(\frac{\pi^2}{\beta^2}-4t^2\sinh(2r)\right)+O(h^{-1})\\
&\ge
 \frac{\pi^2}{2\beta^2}-4t^2\sinh(2r)+O(h^{-1}).
\end{align*}
If $r=\cF_{t,\beta}(\eps)$ and $h$ is large enough, 
\begin{equation}\label{eq_denominator_lowerbound_0}
|\cD^{\s}(\bk+w\be_p+z\be_q,\o)|\ge
 \frac{(1-\eps)\pi^2}{4\beta^2}>0.
\end{equation}
Therefore, the denominator of
 $\chi_{N_{\beta}}(\o)\cB_{\rho,\eta}^{\s}(\bk+w\be_p+z\be_q,w)$ does
 not vanish for any $\rho,\eta\in\{1,2,3\}$, which ensures the
 analyticity of $\cC_{N_{\beta}}(X,Y)(w\be_p+z\be_q)$ in the claimed
 domain.

\eqref{item_covariance_L1_bound_0}: Fix $w,z\in\C$ with $|\Im w|$,
 $|\Im z|< \cF_{t,\beta}(\eps)$
 and $p,q\in\{1,2\}$. We will use the following bounds. For any $(k_1,k_2)\in\R^2$, 
\begin{align}
&|\sin(k_j+w\delta_{j,p}+z\delta_{j,q})|,|\cos(k_j+w\delta_{j,p}+z\delta_{j,q})|,\notag\\
&|\sin(k_1+w\delta_{1,p}+z\delta_{1,q}-k_2-w\delta_{2,p}-z\delta_{2,q})|,\notag\\
&|\cos(k_1+w\delta_{1,p}+z\delta_{1,q}-k_2-w\delta_{2,p}-z\delta_{2,q})|\le
 c\left(1+\frac{1}{\max\{\beta,\beta^2\}}\right)\label{eq_sin_cos_bound}
\end{align}
($\forall j\in\{1,2\}$). By keeping
 \eqref{eq_decay_extra_terms}, \eqref{eq_first_condition_M}, \eqref{eq_bound_omega}
  and \eqref{eq_sin_cos_bound} in mind, one
 can deduce the following. For any $\o\in \cM_h$ with
 $\chi_{N_{\beta}}(\o)\neq 0$ and large enough $h\in 2\N/\beta$,
\begin{align*}
&|\cN_{1,1}^{\s}(\bk+w\be_p+z\be_q,\o)|\le
 cM^3\max\left\{1,\frac{1}{\beta}\right\},\\
&|\cN_{\rho,\eta}^{\s}(\bk+w\be_p+z\be_q,\o)|\le c
 M\left(1+\frac{1}{\max\{\beta,\beta^2\}}\right)\quad (\forall
 (\rho,\eta)\in\{1,2,3\}^2\backslash\{(1,1)\}),\\
&|h(1-e^{iw/h+\eo/h})|\ge |\o|+O(h^{-1})\ge \frac{c}{\beta}.
\end{align*}
It follows from these inequalities and
 \eqref{eq_denominator_lowerbound_0} that 
\begin{align*}
&|\cB_{\rho,\eta}^{\s}(\bk+w\be_p+z\be_q,\o)|\\
&\le\left\{\begin{array}{ll}\frac{c}{1-\eps}M^3\beta\max\{1,\beta\}&\text{
						if }(\rho,\eta)=(1,1),\\
\frac{c}{1-\eps}M\beta\max\{1,\beta\}&\text{
						if
						}(\rho,\eta)\in\{(1,2),(1,3),(2,1),(3,1)\},\\
c\beta+\frac{c}{1-\eps}M\beta^2\max\{1,\beta\}&\text{ if
						      }(\rho,\eta)\in\{(2,2),(2,3),(3,2),(3,3)\},\end{array}\right.
\end{align*}
which results in
\begin{equation}\label{eq_uniform_bound_momentum_0}
|\cB_{\rho,\eta}^{\s}(\bk+w\be_p+z\be_q,\o)|\le
 \frac{c}{1-\eps}M^3\beta\max\{1,\beta\}^2\quad (\forall
 \rho,\eta\in\{1,2,3\}).
\end{equation}
Then, by using Lemma \ref{lem_matsubara_count} and \eqref{eq_beta_M} we
 have for any $X,Y\in I_{L,h}$ that
\begin{equation}\label{eq_uniform_bound_0}
|\cC_{N_{\beta}}(X,Y)(w\be_p+z\be_q)|\le
 \frac{c}{1-\eps}M^{N_{\beta}+5}\beta\max\{1,\beta\}^2\le
 \frac{c}{1-\eps}M^{6}\max\{1,\beta\}^3.
\end{equation}

The rest of the proof of \eqref{item_covariance_L1_bound_0} proceeds in
 the same way as in \cite[\mbox{Subsection 5.2}]{K2}. By noting the domain
 of analyticity proved in \eqref{item_covariance_analyticity_0} and the
 periodicity of $\cB_{\rho,\eta}^{\s}(\bk,\o)$ with respect to $\bk$ one
 can derive the following equality. For $n\in\N$,
\begin{align}
&\left(\frac{L}{2\pi}\left(e^{i\frac{2\pi}{L}\<\bx-\by,\be_q\>}-1\right)\right)^n\cC_{N_{\beta}}(\rho\bx\s
 x,\eta\by\tau y )(w\be_p)\notag\\
&=\prod_{j=1}^n\left(\frac{L}{2\pi}\int_0^{2\pi/L}d\theta_j\frac{1}{2\pi
 i}\oint_{|z_j-\theta_j|=\cF_{t,\beta}(\eps/2)/n}dz_j\frac{1}{(z_j-\theta_j)^2}\right)\notag\\
&\quad\cdot\cC_{N_{\beta}}(\rho\bx\s
 x,\eta\by\tau y)\left(w\be_p+\hat{s}(\s)\sum_{j=1}^nz_j\be_q\right).\label{eq_contour_integral_covariance}
\end{align}
By taking the absolute value of both sides of
 \eqref{eq_contour_integral_covariance} and using the inequality
 $n^n\le n! e^n$ and \eqref{eq_uniform_bound_0} we obtain
$$
\left|\frac{L}{2\pi}\left(e^{i\frac{2\pi}{L}\<\bx-\by,\be_q\>}-1\right)\right|^n|\cC_{N_{\beta}}(\rho\bx\s
 x,\eta\by\tau
 y)(w\be_p)|\le\frac{c}{1-\eps}M^6\max\{1,\beta\}^3\frac{n!e^n}{\cF_{t,\beta}(\eps/2)^{n}}
$$ 
for any $n\in\N\cup \{0\}$,
which leads to
\begin{align*}
&|\cC_{N_{\beta}}(\rho\bx\s x,\eta\by\tau y)(w\be_p)|\\
&\le
 \frac{c}{1-\eps}M^6\max\{1,\beta\}^3\left(\frac{\eps\pi^2}{16\max\{1,t^2\}\max\{\beta,\beta^2\}}+1\right)^{-\frac{1}{8e}\sum_{q=1}^2\left|\frac{e^{i2\pi\<\bx-\by,\be_q\>/L}-1}{2\pi/L}\right|}.
\end{align*}
Then, by using the inequality that $|(e^{i2\pi m/L}-1)/(2\pi/L)|\ge
 2|m|/\pi$ ($\forall m\in\Z$ with $|m|\le L/2$) and
 \eqref{eq_beta_M} we can deduce that
\begin{align*}
&\frac{1}{h}\sum_{(\bx,x)\in\G\times[-\beta,\beta)_h}|\cC_{N_{\beta}}(\rho\bx\s
 x,\eta\b0\tau 0)(w\be_p)|\\
&\le
 \frac{c}{1-\eps}M^6\beta\max\{1,\beta\}^3\left(\frac{\left(\frac{\eps\pi^2}{16\max\{1,t^2\}\max\{\beta,\beta^2\}}+1\right)^{1/(4\pi
 e)}+1}{\left(\frac{\eps\pi^2}{16\max\{1,t^2\}\max\{\beta,\beta^2\}}+1\right)^{1/(4\pi
 e)}-1}\right)^2\\
&\le
 \frac{c}{1-\eps}M^6\beta\max\{1,\beta\}^3\left(1+\frac{\max\{1,t^2\}\max\{\beta,\beta^2\}}{\eps}\right)^2\\
&\le \frac{c}{(1-\eps)\eps^2}M^8\beta\max\{1,\beta\}^7\le
 \frac{c}{(1-\eps)\eps^2}M^{9-N_{\beta}}\max\{1,\beta\}^8.
\end{align*}

\eqref{item_covariance_determinant_bound_0}: Define the complex Hilbert
 space $\cH$ by $\cH:=\C^m\otimes L^2(\{1,2,3\}\times \G^*\times
 \spin\times \cM_h)$ with the inner product
$$
\<\bu\otimes f,\bv\otimes g\>_{\cH}:=\<\bu,\bv\>_{\C^m}\frac{1}{\beta
 L^2}\sum_{(\rho,\bk,\s,\o)\atop\in
\{1,2,3\}\times\G^*\times\spin\times\cM_h}f(\rho,\bk,\s,\o)\overline{g(\rho,\bk,\s,\o)}.
$$
Moreover, define the vectors $f_{l,X}$, $g_{l,X}\in
 L^2(\{1,2,3\}\times\G^*\times\spin\times\cM_h)$
 $(X\in I_{L,h}$, $l\in\{N_{\beta},\cdots,N_h\})$ by
\begin{align}
&f_{l,\rho\bx\s
 x}(\eta,\bk,\tau,\o):=\delta_{\s,\tau}e^{-i\<\bx,\bk\>}e^{ix\o}\chi_l(\o)^{1/2}\cB_{\rho,\eta}^{\s}(\bk+\hat{s}(\s)w\be_p,\o),\label{eq_def_one_vector}\\
&g_{l,\rho\bx\s
 x}(\eta,\bk,\tau,\o):=\delta_{\rho,\eta}\delta_{\s,\tau}e^{-i\<\bx,\bk\>}e^{ix\o}\chi_l(\o)^{1/2}.\label{eq_def_the_other_vector}
\end{align}
The vectors $f_{l,X}$, $g_{l,X}$ for $l\ge N_{\beta}+1$ will be used in
 the proof of the next lemma. 
We see that $\<\bu,\bv\>_{\C^m}\cC_{N_{\beta}}(X,Y)(w\be_p)=\<\bu\otimes
 f_{N_{\beta},X},\bv\otimes
 g_{N_{\beta},Y}\>_{\cH}$. By Lemma \ref{lem_matsubara_count},
 \eqref{eq_beta_M} and \eqref{eq_uniform_bound_momentum_0}, 
\begin{align*}
&\|\bu\otimes f_{N_{\beta},X}\|_{\cH}\le
 (M^3\max\{1,\beta^{-1}\})^{1/2}\frac{c}{1-\eps}M^3\beta\max\{1,\beta\}^2,\\
& \|\bv\otimes g_{N_{\beta},X}\|_{\cH}\le
 c(M^3\max\{1,\beta^{-1}\})^{1/2},
\end{align*}
if $\|\bu\|_{\C^m}$, $\|\bv\|_{\C^m}\le 1$. Therefore, Gram's inequality
 guarantees that if $\|\bu_j\|_{\C^m}$, $\|\bv_j\|_{\C^m}$ $\le 1$
 $(\forall j\in \{1,\cdots,n\})$,
\begin{align*} 
|\det(\<\bu_j,\bv_k\>_{\C^m}\cC_{N_{\beta}}(X_j,Y_k)(w\be_p))_{1\le j,k\le
 n}|&\le \prod_{j=1}^n\|\bu_j\otimes
 f_{N_{\beta},X_j}\|_{\cH}\|\bv_j\otimes
 g_{N_{\beta},Y_j}\|_{\cH}\\
&\le\left(\frac{c}{1-\eps}M^6\max\{1,\beta\}^3\right)^n.
\end{align*}
\end{proof}

The following lemma gives upper bounds on $\cC_l$
$(l\in\{N_{\beta}+1,\cdots,N_h\})$, which are essentially independent of
$\beta$ in low temperatures.
\begin{lemma}\label{lem_covariance_properties_l}
For any $\eps\in (0,1)$ there exists $N_{\eps}\in\N$ such that for any
 $h\in2\N/\beta$ with $h\ge 2N_{\eps}/\beta$ and
 $l\in\{N_{\beta}+1,\cdots,N_h\}$ the following statements hold true.
\begin{enumerate}[(i)]
\item\label{item_covariance_analyticity_l}
The function $w\mapsto \cC_l(X,Y)(w\be_p)$ is analytic in $\{w\in\C\ |\
     |\Im w|<\cF_{t,\beta}(\eps)\}$ for any $X,Y\in I_{L,h}$,
      $p\in\{1,2\}$.
\item\label{item_covariance_L1_bound_l} 
$$
\frac{1}{h}\sum_{(\bx,x)\in\G\times
     [-\beta,\beta)_h}|\cC_{l}(\rho\bx\s x, \eta \b0 \s
     0)(w\be_p)|\le cM^{8-l}
$$
for any $\rho,\eta\in \{1,2,3\}$, $\s\in\spin$, $p\in\{1,2\}$ and $w\in\C$ with $|\Im w|< \cF_{t,\beta}(\eps)$.
\item\label{item_covariance_determinant_bound_l}
$$
|\det(\<\bu_j,\bv_k\>_{\C^m}\cC_{l}(X_j,Y_k)(w\be_p))_{1\le
     j,k\le n}|\le (cM^4)^n
$$
for any $m,n\in \N$, $\bu_j,\bv_j\in\C^m$ with
     $\|\bu_j\|_{\C^m},\|\bv_j\|_{\C^m}\le 1$, $X_j,Y_j\in I_{L,h}$
     $(j=1,\cdots,n)$, $p\in\{1,2\}$ and $w\in\C$ with 
$|\Im w|<\cF_{t,\beta}(\eps)$.
\item\label{item_covariance_position_decay}
$$
|\cC_l(\rho\hbx\s x,\eta \hby\tau y)(w\be_p)|\le cM^{3+N_{\beta}-l}
$$
for any $\hbx,\hby\in \Z^2$ with $1\le \|\hbx-\hby\|_{\R^2}\le L/2$,
     $(\rho,\s,x)$, $(\eta,\tau,y)\in\{1,2,3\}\times \spin \times
     [0,\beta)_h$, $p\in\{1,2\}$ and $w\in\C$ with 
$|\Im w|< \cF_{t,\beta}(\eps)$.
\item\label{item_covariance_tadpole_decay}
$$
|\cC_l(\rho\b0\s 0,\eta\b0\tau 0)(w\be_p)|\le cM^3(M^{l-N_h}+M^{N_{\beta}-l})
$$
for any $(\rho,\s),(\eta,\tau)\in \{1,2,3\}\times\spin$, $p\in\{1,2\}$ and $w\in\C$ with $|\Im w|<\cF_{t,\beta}(\eps)$.
\end{enumerate}
\end{lemma}
\begin{proof}
\eqref{item_covariance_analyticity_l}: For any $\o\in\R$ with
 $\chi_l(\o)\neq 0$ and sufficiently large $h$,
\begin{equation}\label{eq_upper_lower_l}
\frac{1}{2}M^l-E_{max}\le |h(1-e^{i\o/h+(\ec+\eo)/(2h)})|\le
 2M^{l+2}+E_{max}.
\end{equation}
The condition \eqref{eq_first_condition_M} implies that
$10+\frac{1}{2}E_{max}+9E_{max}^2\le \frac{39}{2}E_{max}^2\le\frac{1}{4}M$,
or
\begin{equation}\label{eq_result_condition_M}
9E_{max}^2+10M^{l-1}\le
 \frac{1}{2}\left(\frac{1}{2}M^l-E_{max}\right)\le\frac{1}{2}\left(\frac{1}{2}M^l-E_{max}\right)^2.
\end{equation}
Note that by \eqref{eq_beta_M} and \eqref{eq_derivative_dispersion},
\begin{equation}\label{eq_X_bound}
\left|\left(\frac{\partial}{\partial
       k_j}\right)^m\left(\frac{(\ec-\eo)^2}{4}+E(t,\bk+w\be_p)\right)\right|\le
9E_{max}^2+\frac{\pi^2}{\max\{\beta,\beta^2\}}\le
9E_{max}^2+\pi^2M^{l-1}
\end{equation}
for any $m\in\{0,\cdots,4\}$. Then, by using
 \eqref{eq_upper_lower_l}, \eqref{eq_result_condition_M} and
 \eqref{eq_X_bound} we have for any $\bk\in \R^2$ that
\begin{align}
&|\cD^{\s}(\bk+w\be_p,\o)|\ge
 \left(\frac{1}{2}M^l-E_{max}\right)^2-e^{E_{max}/h}(9E_{max}^2+\pi^2M^{l-1})+O(h^{-1})\notag\\
&\ge \left(\frac{1}{2}M^l-E_{max}\right)^2-9E_{max}^2-10M^{l-1}\ge \frac{1}{2}\left(\frac{1}{2}M^l-E_{max}\right)^2\ge
 \frac{1}{16}M^{2l}.\label{eq_denominator_lower_bound_l}
\end{align}
Thus, the denominator of $\chi_l(\o)\cB_{\rho,\eta}^{\s}(\bk+w\be_p,\o)$
is non-zero for any $\rho,\eta\in\{1,2,3\}$, which proves the claim
 \eqref{item_covariance_analyticity_l}.

\eqref{item_covariance_L1_bound_l},\eqref{item_covariance_position_decay}:
 Take $\o\in\R$ with $\chi_l(\o)\neq 0$, $p\in\{1,2\}$,
 $\bk\in\R^2$, $\s\in\spin$ and $w\in\C$ with $|\Im w|<\cF_{t,\beta}(\eps)$.
Estimating $|\chi_l(\o)(\partial/\partial
 k_j)^m\cB_{\rho,\eta}^{\s}(\bk+w\be_p,\o)|$, $|(\partial/\partial\o)^m
 (\chi_l(\o)\cB_{\rho,\eta}^{\s}(\bk+w\be_p,\o))|$ $(m=0,\cdots,4)$
 provides sufficient information to bound the sum of $\cC_l(w\be_p)$ over $\G\times[0,\beta)_h$.  
By using the inequalities \eqref{eq_beta_M}, \eqref{eq_h_M} and 
\eqref{eq_upper_lower_l} we obtain
\begin{align*}
&\left|\left(\frac{\partial}{\partial
       \o}\right)^m\cD^{\s}(\bk+w\be_p,\o)\right|\le c M^{4+(2-m)l}\quad
(\forall m\in\{0,\cdots,4\}),\\
&\left|\left(\frac{\partial}{\partial
       k_j}\right)^n\cD^{\s}(\bk+w\be_p,\o)\right|\le c M^{N_{\beta}}\quad
(\forall n\in\{1,\cdots,4\},j\in\{1,2\}),
\end{align*}
which, combined with \eqref{eq_denominator_lower_bound_l}, yields
\begin{equation}
\begin{split}
&\left|\left(\frac{\partial}{\partial
 \o}\right)^m\frac{1}{\cD^{\s}(\bk+w\be_p,\o)}\right|\le c
 M^{4m-(2+m)l}\quad (\forall m\in \{0,\cdots,4\}),\\
&\left|\left(\frac{\partial}{\partial
 k_j}\right)^n\frac{1}{\cD^{\s}(\bk+w\be_p,\o)}\right|\le c
 M^{N_{\beta}-4l}\quad (\forall n\in \{1,\cdots,4\},j\in\{1,2\}).\label{eq_denominator_derivative}
\end{split}
\end{equation}
One can similarly derive the following inequalities. For any $m\in
 \{0,\cdots,4\}$, $n\in \{1,\cdots,4\}$, $j\in\{1,2\}$ and
 $(\rho,\eta)\in \{1,2,3\}^2\backslash \{(1,1)\}$,
\begin{align}
&\left|\left(\frac{\partial}{\partial
 \o}\right)^m\cN_{1,1}^{\s}(\bk+w\be_p,\o)\right|\le c
 M^{2+(1-m)l},\quad \left|\left(\frac{\partial}{\partial
 k_j}\right)^n\cN_{1,1}^{\s}(\bk+w\be_p,\o)\right|\le c,\notag\\ 
&\left|\left(\frac{\partial}{\partial
 \o}\right)^m\cN_{\rho,\eta}^{\s}(\bk+w\be_p,\o)\right|\le c
 M^{N_{\beta}+1-ml},\quad \left|\left(\frac{\partial}{\partial
 k_j}\right)^n\cN_{\rho,\eta}^{\s}(\bk+w\be_p,\o)\right|\le c
 M^{N_{\beta}+1}.\label{eq_numerator_derivative_pre}
\end{align}
These imply that for any $m\in\{0,\dots,4\}$, $n\in\{1,\cdots,4\}$,
 $j\in\{1,2\}$, $\rho,\eta\in\{1,2,3\}$,
\begin{equation}\label{eq_numerator_derivative}
\begin{split}
\left|\left(\frac{\partial}{\partial
 \o}\right)^m\cN_{\rho,\eta}^{\s}(\bk+w\be_p,\o)\right|\le c
 M^{2+(1-m)l},\quad \left|\left(\frac{\partial}{\partial
 k_j}\right)^n\cN_{\rho,\eta}^{\s}(\bk+w\be_p,\o)\right|\le c
 M^{N_{\beta}+1}.
\end{split}
\end{equation}
As in \eqref{eq_upper_lower_l}, $|h(1-e^{i\o/h+\eo/h})|\ge
 \frac{1}{2}M^l-E_{max}\ge c M^l$. Thus, 
\begin{equation}\label{eq_lower_another}
\left|\left(\frac{\partial}{\partial \o}\right)^m\frac{1}{h(1-e^{i\o/h+\eo/h})
}\right|\le cM^{-(m+1)l}\quad (\forall m\in\{0,\cdots,4\}).
\end{equation}
One can also check that  
\begin{equation}\label{eq_bound_cutoff}
\left|\left(\frac{\partial}{\partial \o}\right)^m\chi_l(\o)\right|\le
 cM^{m(1-l)}\quad (\forall m\in\{0,\cdots,4\}).
\end{equation}
Then by using
 \eqref{eq_denominator_derivative}, \eqref{eq_numerator_derivative},
 \eqref{eq_lower_another}, \eqref{eq_bound_cutoff} and Leibniz' formula,
 we have for any $\rho,\eta\in\{1,2,3\}$, $j\in\{1,2\}$ that 
\begin{align}
&1_{\chi_l(\o)\neq 0}\left|\cB_{\rho,\eta}^{\s}(\bk+w\be_p,\o)\right|\le
 cM^{2-l},\label{eq_covariance_bound_momentum}\\
&\left|\left(\frac{\partial}{\partial
 \o}\right)^4(\chi_l(\o)\cB_{\rho,\eta}^{\s}(\bk+w\be_p,\o))\right|\le c
 M^{18-5l},\label{eq_covariance_bound_momentum_matsubara}\\
&\left|\chi_l(\o)\left(\frac{\partial}{\partial
 k_j}\right)^4\cB_{\rho,\eta}^{\s}(\bk+w\be_p,\o)\right|\le c
 M^{1+N_{\beta}-2l}.\label{eq_covariance_bound_momentum_derivative}
\end{align}

It follows from \eqref{eq_covariance_bound_momentum} and Lemma
 \ref{lem_matsubara_count} that
\begin{equation}\label{eq_covariance_bound_position}
|\cC_l(X,Y)(w\be_p)|\le cM^4\quad (\forall X,Y\in I_{L,h}).
\end{equation}

For a function $f:\C\to \C$, let
 $d_{\beta}f(\o):=\frac{\beta}{2\pi}(f(\o+2\pi/\beta)-f(\o))$. By
 remarking the periodicity that $\chi_l(\o+2\pi h
 m)\cB^{\s}_{\rho,\eta}(\bk+\hat{s}(\s)w\be_p,\o+2\pi
 hm)=\chi_l(\o)\cB_{\rho,\eta}^{\s}(\bk+\hat{s}(\s)w\be_p,\o)$ $(\forall m\in \Z)$, we observe that
\begin{align*}
&\left(\frac{\beta}{2\pi}\left(e^{-i\frac{2\pi}{\beta}(x-y)}-1\right)\right)^4\cC_l(\rho\bx\s
 x,\eta\by\tau y)(w\be_p)\\
&=\frac{\delta_{\s,\tau}}{\beta
 L^2}\sum_{(\bk,\o)\in\G^*\times
 \cM_h}e^{-i\<\bx-\by,\bk\>}e^{i(x-y)\o}d_{\beta}^4\left(\chi_l(\o)\cB_{\rho,\eta}^{\s}(\bk+\hat{s}(\s)w\be_p,\o)\right)\\
&=\frac{\delta_{\s,\tau}}{\beta
 L^2}\sum_{(\bk,\o)\in\G^*\times
 \cM_h}e^{-i\<\bx-\by,\bk\>}e^{i(x-y)\o}\prod_{m=1}^4\left(\frac{\beta}{2\pi}\int_{0}^{2\pi/\beta}dv_m\right)\\
&\quad\cdot\left(\frac{\partial}{\partial\o}\right)^4\left(\chi_l\left(\o+\sum_{m=1}^4v_m
\right)\cB_{\rho,\eta}^{\s}\left(\bk+\hat{s}(\s)w\be_p,\o+\sum_{m=1}^4v_m\right)\right).
\end{align*}
Then, the bound \eqref{eq_covariance_bound_momentum_matsubara} and Lemma
 \ref{lem_matsubara_count} lead to
\begin{equation}\label{eq_bound_position_temperature}
\left|\frac{\beta}{2\pi}\left(e^{-i\frac{2\pi}{\beta}(x-y)}-1\right)\right|^4|\cC_l(X,Y)(w\be_p)|\le
 c M^{20-4l}.
\end{equation}

Similarly by using the periodicity that
 $\cB_{\rho,\eta}^{\s}(\bk+\hat{s}(\s)w\be_p+2\pi
 n\be_j,\o)=\cB_{\rho,\eta}^{\s}(\bk+\hat{s}(\s)w\be_p,\o)$ $(\forall n\in
 \Z)$ we obtain
\begin{align*}
&\left(\frac{L}{2\pi}\left(e^{i\frac{2\pi}{L}\<\bx-\by,\be_j\>}-1\right)\right)^4\cC_l(\rho\bx\s
 x,\eta\by\tau y)(w\be_p)\\
&=\frac{\delta_{\s,\tau}}{\beta
 L^2}\sum_{(\bk,\o)\in\G^*\times
 \cM_h}e^{-i\<\bx-\by,\bk\>}e^{i(x-y)\o}\\
&\quad\cdot\prod_{n=1}^4\left(\frac{L}{2\pi}\int_{0}^{2\pi/L}du_n\right)\chi_l(\o)\left(\frac{\partial}{\partial
 k_j}\right)^4\cB_{\rho,\eta}^{\s}\left(\bk+\hat{s}(\s)w\be_p+\sum_{n=1}^4u_n\be_j,\o\right),
\end{align*}
which, combined with \eqref{eq_covariance_bound_momentum_derivative} and
 Lemma \ref{lem_matsubara_count}, yields
\begin{equation}\label{eq_bound_position_position}
\left|\frac{L}{2\pi}\left(e^{i\frac{2\pi}{L}\<\bx-\by,\be_j\>}-1\right)\right|^4|\cC_l(X,Y)(w\be_p)|\le
 c M^{3+N_{\beta}-l}\quad (\forall j\in\{1,2\}).
\end{equation}

The inequalities \eqref{eq_covariance_bound_position},
 \eqref{eq_bound_position_temperature} and
 \eqref{eq_bound_position_position} result in
\begin{align}
&|\cC_l(\rho\bx\s x,\eta \by \tau
 y)(w\be_p)|\notag\\
&\le\frac{cM^4}{1+M^{l-N_{\beta}+1}\sum_{j=1}^2\left|\frac{L}{2\pi}(e^{i2\pi\<\bx-\by,\be_j\>/L}-1)\right|^4+M^{4l-16}\left|\frac{\beta}{2\pi}(e^{i2\pi(x-y)/\beta}-1)\right|^4}\label{eq_covariance_decay_l}
\end{align}
for all $(\rho,\bx,\s,x)$, $(\eta,\by,\tau,y)\in I_{L,h}$. The decay
 bound \eqref{eq_covariance_decay_l} implies the claim
 \eqref{item_covariance_L1_bound_l} and the claim
 \eqref{item_covariance_position_decay}.

\eqref{item_covariance_determinant_bound_l}: The proof of
 \eqref{item_covariance_determinant_bound_l} is parallel to that of Lemma
 \ref{lem_covariance_properties_0}
 \eqref{item_covariance_determinant_bound_0}. Recall
 \eqref{eq_def_one_vector} and \eqref{eq_def_the_other_vector}. By using
 Lemma \ref{lem_matsubara_count} and
 \eqref{eq_covariance_bound_momentum}
one can show that for any $\bu$, $\bv\in\C^m$ with $\|\bu\|_{\C^m}$,
 $\|\bv\|_{\C^m}\le 1$, $\|\bu\otimes f_{l,X}\|_{\cH}\le
 c (M^{l+2})^{1/2}M^{2-l}$, $\|\bv\otimes g_{l,X}\|_{\cH}\le
 c(M^{l+2})^{1/2}$. Thus, we can apply Gram's inequality to conclude
 that 
\begin{align*}
&\left|\det(\<\bu_j,\bv_k\>_{\C^m}\cC_l(X_j,Y_k)(w\be_p))_{1\le j,k\le
 n}\right|\\
&\le \prod_{j=1}^n\|\bu_j\otimes f_{l,X_j}\|_{\cH}\|\bv_j\otimes
 g_{l,X_j}\|_{\cH}\le (cM^{l+2}\cdot M^{2-l})^n\le (cM^4)^n.
\end{align*}

\eqref{item_covariance_tadpole_decay}: Take $\o\in\cM_h$ with
 $\chi_l(\o)\neq 0$. Since
 $$|\cD^{\s}(\bk+w\be_p,\o)-h^2(1-e^{-i\o/h+(\ec+\eo)/(2h)})^2|\le
 9E_{max}^2+10M^{l-1}$$ 
by \eqref{eq_X_bound}, the inequalities
 \eqref{eq_upper_lower_l} and \eqref{eq_result_condition_M} justify that
\begin{align*}
&\frac{1}{\cD^{\s}(\bk+w\be_p,\o)}=\frac{1}{h^2\left(1-e^{-i\o/h+(\ec+\eo)/(2h)}\right)^2}+\frac{1}{h^3\left(1-e^{-i\o/h+(\ec+\eo)/(2h)}\right)^3}\\
&\qquad\qquad\qquad\qquad
 \cdot\sum_{m=1}^{\infty}\frac{\left(h^2\left(1-e^{-i\o/h+(\ec+\eo)/(2h)}\right)^2-\cD^{\s}(\bk+w\be_p,\o)\right)^m}{\left(h\left(1-e^{-i\o/h+(\ec+\eo)/(2h)}\right)\right)^{2m-1}},\\
&\sum_{m=1}^{\infty}\left|\frac{\left(h^2\left(1-e^{-i\o/h+(\ec+\eo)/(2h)}\right)^2-\cD^{\s}(\bk+w\be_p,\o)\right)^m}{\left(h\left(1-e^{-i\o/h+(\ec+\eo)/(2h)}\right)\right)^{2m-1}}\right|\le
 1.
\end{align*}
This particularly implies that
\begin{align}
&\left|\frac{1}{\beta
 L^2}\sum_{(\bk,\o)\in\G^*\times\cM_h}\chi_l(\o)\cB_{1,1}^{\s}(\bk+w\be_p,\o)\right|\notag\\
&\le \left|\frac{1}{\beta
 L^2}\sum_{(\bk,\o)\in\G^*\times\cM_h}\frac{\chi_l(\o)h\left(1-e^{-i\o/h+(\ec+\eo)/(2h)}\right)}{\cD^{\s}(\bk+w\be_p,\o)}\right|\notag\\
&\quad+\frac{1}{\beta
 L^2}\sum_{(\bk,\o)\in\G^*\times\cM_h}\frac{\chi_l(\o)\left|\cN_{1,1}^{\s}(\bk+w\be_p,\o)-h\left(1-e^{-i\o/h+(\ec+\eo)/(2h)}\right)\right|}{\left|\cD^{\s}(\bk+w\be_p,\o)\right|}\notag\\
&\le
 \left|\frac{1}{\beta}\sum_{\o\in\cM_h}\frac{\chi_l(\o)}{h\left(1-e^{-i\o/h+(\ec+\eo)/(2h)}\right)}\right|+c M^{3-l},\label{eq_tadpole_1_1_pre}
\end{align}
where Lemma \ref{lem_matsubara_count}, \eqref{eq_upper_lower_l} and
 \eqref{eq_denominator_lower_bound_l} were also used. Note that
\begin{align*}
&\frac{1}{\beta}\sum_{\o\in\cM_h}\frac{\chi_l(\o)}{h\left(1-e^{-i\o/h+(\ec+\eo)/(2h)}\right)}\\
&=\frac{1}{2\beta
 h}\sum_{\o\in
 \cM_h}\chi_l(\o)+\frac{1}{2\beta}\sum_{\o\in\cM_h}\chi_l(\o)\frac{h\left(1-e^{(\ec+\eo)/h}\right)}{h^2\left(1-e^{-i\o/h+(\ec+\eo)/(2h)}\right)
\left(1-e^{i\o/h+(\ec+\eo)/(2h)}\right)}.
\end{align*}
Then again by using Lemma \ref{lem_matsubara_count}, \eqref{eq_h_M} and
 \eqref{eq_upper_lower_l} we have
\begin{equation}\label{eq_tadpole_typical}
\left|\frac{1}{\beta}\sum_{\o\in\cM_h}\frac{\chi_l(\o)}{h\left(1-e^{-i\o/h+(\ec+\eo)/(2h)}\right)}\right|\le
 c M^{l-N_h+2}+cM^{3-l}.
\end{equation}
Substituting \eqref{eq_tadpole_typical} into \eqref{eq_tadpole_1_1_pre}
 gives
\begin{equation}\label{eq_tadpole_1_1}
\left|\frac{1}{\beta L^2}\sum_{(\bk,\o)\in\G^*\times
 \cM_h}\chi_l(\o)\cB_{1,1}^{\s}(\bk+w\be_p,\o)\right|\le c
M^{l-N_h+2}+cM^{3-l}.
\end{equation}

It follows from \eqref{eq_numerator_derivative_pre} that 
\begin{equation}\label{eq_tadpole_others_pre}
\frac{1}{\beta L^2}\sum_{(\bk,\o)\in\G^*\times
 \cM_h}\frac{\chi_l(\o)|\cN_{\rho,\eta}^{\s}(\bk+w\be_p,\o)|}{|\cD^{\s}(\bk+w\be_p,\o)|}\le
 c M^{3+N_{\beta}-l}\quad(\forall (\rho,\eta)\in
 \{1,2,3\}^2\backslash\{(1,1)\}).
\end{equation}
The procedure to derive \eqref{eq_tadpole_typical} similarly shows that
\begin{equation}\label{eq_tadpole_typical_another}
\left|\frac{1}{\beta}\sum_{\o\in\cM_h}\frac{\chi_l(\o)}{h\left(1-e^{-i\o/h+\eo/h)}\right)}\right|\le
c M^{l-N_h+2}+cM^{3-l}.
\end{equation}
The bounds \eqref{eq_tadpole_others_pre} and
 \eqref{eq_tadpole_typical_another} yield that for any $(\rho,\eta)\in \{1,2,3\}^2\backslash\{(1,1)\}$, 
\begin{equation}\label{eq_tadpole_others}
\left|\frac{1}{\beta L^2}\sum_{(\bk,\o)\in\G^*\times
 \cM_h}\chi_l(\o)\cB_{\rho,\eta}^{\s}(\bk+w\be_p,\o)\right|\le c
M^{l-N_h+2}+cM^{3+N_{\beta}-l}.
\end{equation}
By \eqref{eq_tadpole_1_1} and \eqref{eq_tadpole_others} we can confirm
 the inequality claimed in \eqref{item_covariance_tadpole_decay}.
\end{proof}

\section{Multi-scale integration}\label{sec_multiscale_integration}
In this section we will find an $h$-,$L$-independent upper bound on 
\begin{equation}\label{eq_schwinger_functional}
\left(\frac{L}{2\pi}\left(e^{i\frac{2\pi}{L}\<\sum_{j=1}^2(\hat{s}(\hs_j)\hbx_j-\hat{s}(\htau_j)\hby_j),\be_p\>}-1\right)\right)^n\frac{\partial}{\partial\la}\log\left(\int
     e^{V_{(\la,\la)}(\psi)}d\mu_{\cC}(\psi)\right)\Big|_{\la=0}
\end{equation}
($n\in\N\cup\{0\}$, $p\in\{1,2\}$) by estimating the right-hand
side of Lemma \ref{lem_contour_integral_formulation}
\eqref{item_contour_integral_transform} by means of a multi-scale
integration over the Matsubara frequency $\cM_h$. By using the upper 
bound on \eqref{eq_schwinger_functional} we will complete the proof of
Theorem \ref{thm_exponential_decay} in the end of this section.

\subsection{Notations for the multi-scale expansion}\label{subsec_multi_notations}
Let us decide some notational rules to systematically handle Grassmann
polynomials during the multi-scale expansion, in addition to those
already introduced in Subsection \ref{subsec_grassmann_integral}.

For $\bX^m=(X_1^m,X_2^m,\cdots,X^m_m)\in I_{L,h}^m$ $(m\in\N)$ let
$(\opsi)_{\bX^m}:=\opsi_{X^m_1}\opsi_{X^m_2}\cdots\opsi_{X^m_m}$, 
$(\psi)_{\bX^m}:=\psi_{X^m_1}\psi_{X^m_2}\cdots\psi_{X^m_m}\in\bigwedge^m\cV$. 
Define the extended index set $\tilde{I}_{L,h}$ by
$\tilde{I}_{L,h}:=I_{L,h}\times\{1,-1\}$. The index set
$\tilde{I}_{L,h}$ is used in the following way. For $(X,a)\in
\tilde{I}_{L,h}$, $\psi_{(X,a)}:=\opsi_X$ if $a=1$,
$\psi_{(X,a)}:=\psi_X$ if $a=-1$. For
$\tilde{\bX}^m=(\tilde{X}_1^m,\tilde{X}_2^m,\cdots,\tilde{X}_m^m)$ $\in
\tilde{I}_{L,h}^m$ let
$(\psi)_{\tilde{\bX}^m}:=\psi_{\tilde{X}_1^m}\psi_{\tilde{X}_2^m}\cdots\psi_{\tilde{X}_m^m}\in \bigwedge^m\cV$. 

For $\bX^m\in I_{L,h}^m$, $\bX^n=(X_1^n,X_2^n,\cdots,X_n^n)\in
I_{L,h}^n$ with $m\le n$, we write $\bX^m\subset \bX^n$ if there exist
$j_1,j_2,\cdots,j_m\in\{1,2,\cdots,n\}$ such that $j_1<j_2<\cdots<j_m$
and $\bX^m=(X_{j_1}^n,X_{j_2}^n,\cdots,X_{j_m}^n)$. Moreover in this
case we define $\bX^n\backslash \bX^m\in I_{L,h}^{n-m}$ by
$\bX^n\backslash \bX^m:=(X_{k_1}^n,X_{k_2}^n,\cdots,X_{k_{n-m}}^n)$,
where $1\le k_1<k_2<\cdots<k_{n-m}\le n$ and $k_q\notin
\{j_1,j_2,\cdots,j_m\}$ $(\forall q\in\{1,2,\cdots,n-m\})$.

For $\tilde{\bX}^m\in \tilde{I}_{L,h}^m$, $\tilde{\bX}^n\in
\tilde{I}_{L,h}^n$ with $m\le n$ the notations $\tilde{\bX}^m\subset
\tilde{\bX}^n$ and $\tilde{\bX}^n\backslash\tilde{\bX}^m$ are defined in
the same way as above. For $\bX^m=(X_1^m,X_2^m,\cdots,X_m^m)\in I_{L,h}^m$ and $a\in
\{1,-1\}$ let $\tilde{\bX}(a)^m:=((X_1^m,a),(X_2^m,a),\cdots,(X_m^m,a))\in \tilde{I}_{
L,h}^m$.

For a function $f_m:I_{L,h}^m\times I_{L,h}^m\to \C$ $(
m\in\N)$ let 
\begin{align*}
&\|f_m\|_1 :=\left(\frac{1}{h}\right)^{2m}\sum_{\bX^m,\bY^m\in I_{L,h}^m}|f_m(\bX^m,\bY^m)|,\\
&\|f_m\|_{1,\infty}:=\max\Bigg\{\\
&\max_{j\in \{0,\cdots,m-1\},\atop X\in
 I_{L,h}}\Bigg\{\left(\frac{1}{h}\right)^{2m-1}\sum_{\bX^j\in I_{L,h}^j}\sum_{\bX^{m-1-j}\in I_{L,h}^{m-1-j}}\sum_{\bY^m\in I_{L,h}^m}|f_m((\bX^j,X,\bX^{m-1-j}),\bY^m)|\Bigg\},\\
&\max_{j\in \{0,\cdots,m-1\},\atop Y\in
 I_{L,h}}\Bigg\{\left(\frac{1}{h}\right)^{2m-1}\sum_{\bY^j\in I_{L,h}^j}\sum_{\bY^{m-1-j}\in I_{L,h}^{m-1-j}}\sum_{\bX^m\in I_{L,h}^m}|f_m(\bX^m,(\bY^j,Y,\bY^{m-1-j}))|\Bigg\}\Bigg\}.
\end{align*}
We see that $\|\cdot\|_{1}$, $\|\cdot\|_{1,\infty}$ are norms in the
complex vector space of all functions on $I_{L,h}^m\times
I_{L,h}^m$. For notational consistency we also set $\|f_0\|_{1}$,
$\|f_0\|_{1,\infty}:=|f_0|$ for any complex number $f_0$. 

Let us call a function $f_m:I_{L,h}^m\times
I_{L,h}^m\to \C$ bi-anti-symmetric if 
\begin{align*}
&f_m((X_{\nu(1)},X_{\nu(2)},\cdots,X_{\nu(m)}),(Y_{\xi(1)},Y_{\xi(2)},\cdots,Y_{\xi(m)}))\\
&=\sgn(\nu)\sgn(\xi)f_m((X_1,X_2,\cdots,X_m),(Y_1,Y_2,\cdots,Y_m))
\end{align*}
for any $(X_1,X_2,\cdots,X_m)$, $(Y_1,Y_2,\cdots,Y_m)\in I_{L,h}^m$ and
$\nu$, $\xi\in\S_m$. Recalling the numbering that $I_{L,h}=\{X_{o,j}\}_{j=1}^{N_{L,h}}$, let 
$$(I_{L,h})_o^m:=\{(X_{o,j_1},X_{o,j_2},\cdots,X_{o,j_m})\in I_{L,h}^m\ |\
j_1<j_2<\cdots<j_m\}\quad (\forall m\in \N).$$
It holds for any bi-anti-symmetric function
$f_m(\cdot,\cdot):I_{L,h}^m\times I_{L,h}^m\to \C$ that
\begin{equation}\label{eq_norm_ordered_base}
\|f_m\|_{1}=\left(\frac{1}{h}\right)^{2m}(m!)^2\sum_{\bX^m,\bY^m\in
 (I_{L,h})_o^m}|f_m(\bX^m,\bY^m)|.
\end{equation}
Bi-anti-symmetric functions appear as kernels of Grassmann polynomials.
Remark that $f(\psi)\in \bigoplus_{n=0}^{N_{L,h}}\cP_n\bigwedge \cV$
can be uniquely written as
$$f(\psi)=\sum_{m=0}^{N_{L,h}}\left(\frac{1}{h}\right)^{2m}\sum_{\bX^m,\bY^m\in
I_{L,h}^m}f_m(\bX^m,\bY^m)(\opsi)_{\bX^m}(\psi)_{\bY^m}$$
with bi-anti-symmetric kernels $f_m(\cdot,\cdot):I_{L,h}^m\times
I_{L,h}^m\to\C$ $(m\in\{0,\cdots,N_{L,h}\})$. 
Moreover, if 
\begin{align*}
&\left(\frac{1}{h}\right)^{2m}\sum_{\bX^m,\bY^m\in
I_{L,h}^m}f_m(\bX^m,\bY^m)(\opsi)_{\bX^m}(\psi)_{\bY^m}\\
&=\left(\frac{1}{h}\right)^{2m}\sum_{\bX^m,\bY^m\in
I_{L,h}^m}g_m(\bX^m,\bY^m)(\opsi)_{\bX^m}(\psi)_{\bY^m}
\end{align*}
and $f_m(\cdot,\cdot)$ is bi-anti-symmetric, then the inequalities 
\begin{align}\label{eq_anti_symmetric_estimate}
\|f_m\|_1\le \|g_m\|_1\text{ and }\|f_m\|_{1,\infty}\le \|g_m\|_{1,\infty}
\end{align}
hold.

Assume that $f_{l,m}(\cdot,\cdot):I_{L,h}^m\times I_{L,h}^m\to\C$ is
bi-anti-symmetric $(\forall l,m\in\N\cup\{0\})$ and
$\lim_{l\to\infty}f_{l,m}(\bX^m,\bY^m)$ exists in $\C$ $(\forall
m\in\N\cup\{0\}, \bX^m,\bY^m\in I_{L,h}^m)$. Set
$$
f_l(\psi):=\sum_{m=0}^{N_{L,h}}\left(\frac{1}{h}\right)^{2m}\sum_{ \bX^m,\bY^m\in I_{L,h}^m}f_{l,m}(\bX^m,\bY^m)(\opsi)_{\bX^m}(\psi)_{\bY^m}.
$$
In this case we define $\lim_{l\to\infty}f_l(\psi) \in
\bigoplus_{n=0}^{N_{L,h}}\cP_n\bigwedge \cV$ by 
$$
\lim_{l\to\infty}f_l(\psi):=\sum_{m=0}^{N_{L,h}}\left(\frac{1}{h}\right)^{2m}\sum_{
\bX^m,\bY^m\in I_{L,h}^m}\lim_{l\to \infty}f_{l,m}(\bX^m,\bY^m)(\opsi)_{\bX^m}(\psi)_{\bY^m}.
$$

We call $f_z(\psi)\in
\bigoplus_{n=0}^{N_{L,h}}\cP_n\bigwedge \cV$ analytic with respect to
$z$ in a domain $\cO(\subset\C)$ if so is every bi-anti-symmetric kernel
of $f_z(\psi)$. Under this condition we define $(d/dz) f_z(\psi)\in
\bigoplus_{n=0}^{N_{L,h}}\cP_n\bigwedge \cV$ by replacing each
bi-anti-symmetric kernel of $f_z(\psi)$ by its derivative. Moreover,
the following Taylor expansion holds true. For
any $\hat{z}\in\cO$,
$$
f_z(\psi)=\sum_{n=0}^{\infty}\frac{1}{n!}\left(\frac{d}{dz}\right)^nf_z(\psi)\Big|_{z=\hat{z}}(z-\hat{z})^n
$$
in a neighbor of $\hat{z}$. 

\subsection{A multi-scale integration over the Matsubara frequency}
\label{subsec_sketch_multiscale}
Here let us describe the multi-scale integration process. From now until
the proof of Theorem \ref{thm_exponential_decay} in Subsection
\ref{subsec_final_integration} we fix arbitrary $R\in
(\cF_{t,\beta}(8/\pi^2),\infty)$, $\eps\in (9/\pi^2,1)$, $p\in\{1,2\}$,
 $L\in\N$ satisfying
$\max_{j,k\in\{1,2\}}\|\hbx_j-\hby_k\|_{\R^2}\le L/2$ and sufficiently
large $h\in 2\N/\beta$. There exists $U_{small}>0$ such that all the statements of Lemma
\ref{lem_contour_integral_formulation}, Lemma
\ref{lem_covariance_properties_0} and Lemma
\ref{lem_covariance_properties_l} hold true for these fixed parameters.
Set
\begin{align*}
&D_{small}:=\{(z_1,z_2,z_3,z_4)\in\C^4\ |\ |z_j|\le U_{small}\quad
 (\forall j\in \{1,2,3,4\})\},\\
&D_R:=\{z\in\C\ |\ |\Re z|\le R,\ |\Im z|\le
 \cF_{t,\beta}(9/\pi^2)\}.
\end{align*}
By taking $U_{small}$ smaller if necessary we may assume that \\
$\Re \int
e^{V_{(\la_1,\la_{-1})}(\psi)}d\mu_{\sum_{j=l}^{N_h}\cC_j(w\be_p)}(\psi)>0$
for all $(\la_1,\la_{-1},U_c,U_o)\in D_{small}$, $w\in D_R$ and
$l\in\{N_{\beta}, \cdots,N_h\}$. This property allows us to
define $G^{\ge l}(\psi)\in \bigwedge \cV$
$(l\in\{N_{\beta},\cdots,N_{h}+1\})$ by
\begin{align*}
&G^{\ge l}(\psi):=\log\left(\int
e^{V_{(\la_1,\la_{-1})}(\psi+\psi^0)}d\mu_{\sum_{j=l}^{N_h}\cC_j(w\be_p)}(\psi^0)\right)\quad
 (l\in\{N_{\beta},\cdots,N_{h}\}),\\
&G^{\ge N_h+1}(\psi):=V_{(\la_1,\la_{-1})}(\psi)
\end{align*}
for any $(\la_1,\la_{-1},U_c,U_o)\in D_{small}$, $w\in D_R$. The
definition of logarithm of Grassmann polynomials is provided in
Definition \ref{def_log_grassmann} in Appendix \ref{app_log_grassmann}. 

By noting the equality that
$$
G^{\ge l}(\psi)=\log\left(\int\left(\int
e^{V_{(\la_1,\la_{-1})}(\psi+\psi^1+\psi^0)}d\mu_{\sum_{j=l+1}^{N_h}\cC_j(w\be_p)}(\psi^0)\right)d\mu_{\cC_l(w\be_p)}(\psi^1)\right) 
$$
(see, e.g, \cite[\mbox{Proposition I.21}]{FKT} to justify this equality),
Lemma \ref{lem_exponential_log_equality} proved in Appendix
\ref{app_log_grassmann} ensures that for any $(\la_1,\la_{-1},U_c,U_o)\in D_{small}$, $w\in D_R$, $l\in
\{N_{\beta},\cdots,N_h\}$,
\begin{equation}\label{eq_recursive_relation_grassmann}
G^{\ge l}(\psi)=\log \left(\int
e^{G^{\ge l+1}(\psi+\psi^1)}d\mu_{\cC_l(w\be_p)}(\psi^1)\right).
\end{equation}

Since 
$$
\lim_{U_{small}\searrow 0}\sup_{(\la_1,\la_{-1},U_c,U_o)\in
D_{small},\atop w\in D_R, z\in \C\text{ with }|z|\le 2}
\left|\int 
e^{z G^{\ge l+1}(\psi^0)}d\mu_{\cC_l(w\be_p)}(\psi^0)-1\right|=0,
$$
one can see from Definition \ref{def_log_grassmann} that
$z\mapsto \log\left(\int 
e^{z G^{\ge l+1}(\psi+\psi^0)}d\mu_{\cC_l(w\be_p)}(\psi^0)\right)$
is analytic in $\{z\in\C\ |\ |z|<2\}$ for any
$(\la_1,\la_{-1},U_c,U_o)\in D_{small}$, $w\in D_R$ if $U_{small}$ is
small enough. Thus, the Taylor expansion around $z=0$ reads
\begin{equation}\label{eq_taylor_grassmann}
G^{\ge
 l}(\psi)=\sum_{n=1}^{\infty}\frac{1}{n!}\left(\frac{d}{dz}\right)^n\log
\left(\int 
e^{z G^{\ge
l+1}(\psi+\psi^0)}d\mu_{\cC_l(w\be_p)}(\psi^0)\right)\Big|_{z=0}
\end{equation}
for any $(\la_1,\la_{-1},U_c,U_o)\in D_{small}$, $w\in D_R$, $l\in
\{N_{\beta},\cdots,N_h\}$.

Each term of \eqref{eq_taylor_grassmann} can be characterized
further. It follows from Definition \ref{def_log_grassmann} and
\eqref{eq_def_exponential_grassmann} that
$$
\frac{d}{dz}\log\left(\int 
e^{z G^{\ge
l+1}(\psi+\psi^0)}d\mu_{\cC_l(w\be_p)}(\psi^0)\right)\Big|_{z=0}=\int
G^{\ge l+1}(\psi+\psi^0)d\mu_{\cC_l(w\be_p)}(\psi^0).
$$
The higher order derivatives can be expanded by means of the tree
formula. We especially apply the version clearly proved in 
\cite[\mbox{Theorem 3}]{SW}. For $n\in\N_{\ge 2}$,
\begin{equation}\label{eq_expanding_trees}
 \frac{1}{n!}\left(\frac{d}{dz}\right)^n\log
\left(\int 
e^{z G^{\ge
l+1}(\psi+\psi^0)}d\mu_{\cC_l(w\be_p)}(\psi^0)\right)\Big|_{z=0}=T_{ree}(n,\cC_l(w\be_p),G^{\ge
l+1}),
\end{equation}
where for $n\in \N_{\ge 2}$, a matrix $Q=(Q(X,Y))_{X,Y\in I_{L,h}}$ and
$f(\psi)\in\bigwedge \cV$,
\begin{align}
T_{ree}(n,Q,f):=&\frac{1}{n!}\sum_{T\in\T_n}\prod_{\{q,r\}\in
 T}(\D_{q,r}(Q)+\D_{r,q}(Q))\int_{[0,1]^{n-1}}d\bs
\sum_{\xi\in \S_n(T)}\varphi(T,\xi,\bs)\notag\\
&\cdot e^{\sum_{u,v=1}^nM_{at}(T,\xi,\bs)_{u,v}\D_{u,v}(Q)}\prod_{j=1}^nf(\psi^j+\psi)\Big|_{\psi^j=\b0\atop
 \forall j\in \{1,\cdots,n\}}.\label{eq_tree_expansion}
\end{align}
The new notations in \eqref{eq_tree_expansion} are defined as follows. $\T_n$ is the
set of all trees over the vertices $\{1,2,\cdots,n\}$, for 
$q,r\in\{1,\cdots,n\}$,
$$
\D_{q,r}(Q):=-\sum_{X,Y\in I_{L,h}}Q(X,Y)\frac{\partial}{\partial \opsi_X^q}\frac{\partial}{\partial \psi_Y^r}:\bigwedge\left(\bigoplus_{j=1}^n\cV_j\right)\to\bigwedge\left(\bigoplus_{j=1}^n\cV_j\right),
$$
$\S_n(T)$ is a $T$-dependent subset of $\S_n$,
the function $\varphi(T,\xi,\cdot):[0,1]^{n-1}\to\R_{\ge 0}$ depends on
$T\in\T_n$, $\xi\in \S_n(T)$ and satisfies
\begin{equation}\label{eq_nice_property_phi}
\int_{[0,1]^{n-1}}d\bs\sum_{\xi\in
 \S_n(T)}\varphi(T,\xi,\bs)=1\quad(\forall T\in \T_n),
\end{equation}
and $(M_{at}(T,\xi,\bs)_{u,v})_{1\le u,v\le n}$ is a
$(T,\xi,\bs)$-dependent real symmetric
non-negative matrix satisfying $M_{at}(T,\xi,\bs)_{u,u}=1$ $(\forall u\in
\{1,\cdots,n\})$. 

Our strategy is to introduce a counterpart of $G^{\ge l}$ via the tree
formula inductively without assuming that $(\la_1,\la_{-1},U_c,U_o)\in
D_{small}$ and prove that the counterpart is well-defined for larger
$(\la_1,\la_{-1},U_c,U_o)$. Consequently by the identity theorem
for analytic functions we will be able to find an upper bound on
\eqref{eq_schwinger_functional} with the enlarged coupling constants
$U_c$, $U_o$ in the end of this section.

\subsection{Estimation by induction}
Let us start the concrete analysis. In the following we fix arbitrary $w\in D_R$ unless
otherwise stated. Define $J^{\ge l}(\psi)$, $F^{\ge l}(\psi)$, $T^{\ge
l}(\psi)\in\bigoplus_{n=0}^{N_{L,h}}\cP_n \bigwedge\cV$ $(l\in
\{N_{\beta},\cdots,N_h+1\})$ inductively as follows.
\begin{equation*}
F^{\ge N_h+1}(\psi):=V_{(\la_1,\la_{-1})}(\psi),\
 T^{\ge N_h+1}(\psi):=0,\ J^{\ge N_h+1}(\psi):=F^{\ge N_h+1}(\psi)+
 T^{\ge N_h+1}(\psi).
\end{equation*}
For $l\in \{N_{\beta},\cdots,N_h\}$,
\begin{align*}
&F^{\ge l}(\psi):=\int J^{\ge
 l+1}(\psi+\psi^0)d\mu_{\cC_l(w\be_p)}(\psi^0),\\
&T^{\ge l}_n(\psi):=T_{ree}(n,\cC_l(w\be_p),J^{\ge l+1})\ (\forall
 n\in \N_{\ge 2}),\ T^{\ge l}(\psi):=\sum_{n=2}^{\infty}T_n^{\ge
 l}(\psi),\\
&J^{\ge l}(\psi):=F^{\ge l}(\psi)+T^{\ge l}(\psi).
\end{align*}
We will later make sure that $\sum_{n=2}^{\infty}T_n^{\ge l}(\psi)$ is
well-defined in $\bigoplus_{n=0}^{N_{L,h}}\cP_n \bigwedge\cV$ if the
input $J^{\ge l+1}$ satisfies a certain smallness condition. 
For $m\in \N\cup \{0\}$, $l\in \{N_{\beta},\cdots,N_h+1\}$ let 
\begin{align*}
&F_m^{\ge l}(\psi):=\cP_m F^{\ge l}(\psi)=\left(\frac{1}{h}\right)^{2m}\sum_{\bX^m,\bY^m\in
 I_{L,h}^m}F_m^{\ge l}(\bX^m,\bY^m)(\opsi)_{\bX^m}(\psi)_{\bY^m},\\
&T_{n,m}^{\ge l}(\psi):=\cP_mT_{n}^{\ge l}(\psi)=\left(\frac{1}{h}\right)^{2m}\sum_{\bX^m,\bY^m\in
 I_{L,h}^m}T_{n,m}^{\ge
 l}(\bX^m,\bY^m)(\opsi)_{\bX^m}(\psi)_{\bY^m}\quad (\forall n\in \N_{\ge
 2}),\\
&T_m^{\ge l}(\psi):=\cP_m T^{\ge l}(\psi)=\left(\frac{1}{h}\right)^{2m}\sum_{\bX^m,\bY^m\in
 I_{L,h}^m}T_m^{\ge l}(\bX^m,\bY^m)(\opsi)_{\bX^m}(\psi)_{\bY^m},
\end{align*}
where $F_m^{\ge l}(\cdot,\cdot)$, $T_{n,m}^{\ge l}(\cdot,\cdot)$,
$T_m^{\ge l}(\cdot,\cdot):I_{L,h}^m\times I_{L,h}^m\to \C$ are
bi-anti-symmetric. 

It will be convenient to set $J^{\ge l}_{free,m}(\psi):=F^{\ge l}_m(\psi)$,
$J^{\ge l}_{tree,m}(\psi):=T^{\ge l}_m(\psi)$ and write 
$$
J_{b,m}^{\ge l}(\psi)=\left(\frac{1}{h}\right)^{2m}\sum_{\bX^m,\bY^m\in I_{L,h}^m}J_{b,m}^{\ge
l}(\bX^m,\bY^m)(\opsi)_{\bX^m}(\psi)_{\bY^m}
$$   
with the bi-anti-symmetric kernel $J_{b,m}^{\ge l}(\cdot,\cdot)$ for 
$b\in\{free,tree\}$. Moreover, set
\begin{equation}\label{eq_final_coefficient}
c_0:=\frac{\max\{c,1\}}{(1-\eps)\eps^2}M^9\max\{1,\beta\}^8,
\end{equation}
where the constant $c$ is taken to be the largest one among those
appearing in the upper bounds of Lemma \ref{lem_covariance_properties_0}
and Lemma \ref{lem_covariance_properties_l}. We observe that $c_0\ge 1$ and
\begin{align}
&\|\cC_l(w\be_p)\|_{1,\infty}\le c_0M^{-l}\quad (\forall l\in
 \{N_{\beta},\cdots,N_h\}),\label{eq_simple_decay_bound}\\
&|\det(\<\bu_j,\bv_k\>_{\C^m}\cC_l(X_j,Y_k)(w\be_p))_{1\le j,k\le
 n}|\le c_0^n\label{eq_simple_determinant_bound}\\
&\qquad(\forall l\in \{N_{\beta},\cdots,N_h\},
 m,n\in\N,\bu_j,\bv_j\in\C^m\text{ with
 }\|\bu_j\|_{\C^m},\|\bv_j\|_{\C^m}\le 1,\notag\\
&\qquad X_j,Y_j\in I_{L,h}
 (j=1,\cdots,n)),\notag\\
&|\cC_l(\rho\hbx\s 0,\eta\hby\tau 0)(w\be_p)|\le
 c_0(M^{l-N_h}+M^{N_{\beta}-l})\label{eq_simple_tadpole_bound}\\
&\qquad(\forall l\in \{N_{\beta}+1,\cdots,N_h\},\hbx,\hby\in\Z^2\text{ with
 }0\le \|\hbx-\hby\|_{\R^2}\le L/2,\rho,\eta\in
 \{1,2,3\},\notag\\
&\qquad\s,\tau\in\spin).\notag
\end{align}

Let us introduce
a parameter $\alpha\in \R_{>0}$. As the main objective in this
subsection we will prove the following.
\begin{proposition}\label{prop_inductive_bound}
Assume that
\begin{equation}\label{eq_smallness_assumption}
\begin{split}
&M\ge \max\{78 E_{max}^2,2^8\},\quad\alpha\ge 2^{10}M^2,\\ 
&|\la_1|,|\la_{-1}|,|U_c|,|U_o|<2^{-4}\alpha^{-2}c_0^{-2}M^{N_{\beta}}.
\end{split}
\end{equation}
Then for any $l\in \{N_{\beta}+1,\cdots,N_h+1\}$ the following
 inequalities hold.
\begin{equation}\label{eq_inductive_bound_each_scale}
\begin{split}
&M^{-N_{\beta}}\alpha c_0\sum_{b\in\{free,tree\}}\|J_{b,1}^{\ge
 l}\|_{1,\infty}<1,\\
&M^{-N_{\beta}}\sum_{m=1}^{N_{L,h}}\alpha^mc_0^mM^{(l-N_{\beta})(m-2)}\sum_{b\in\{free,tree\}}\|J_{b,m}^{\ge
 l}\|_{1,\infty}<1.
\end{split}
\end{equation}
\end{proposition}

The core part of the proof of Proposition \ref{prop_inductive_bound} is
the estimation of $\|T_{n,m}^{\ge l}\|_{1,\infty}$, which needs the
next lemma.
\begin{lemma}\label{lem_determinant_bound_application}
For any $\tilde{\bX}^{m_j}\in \tilde{I}_{L,h}^{m_j}$
 $(j=1,\cdots,n)$, $T\in \T_n$, $\xi\in \S_n(T)$, $\bs\in [0,1]^{n-1}$
and $l\in \{N_{\beta},\cdots,N_h\}$,
\begin{equation*}
\left|e^{\sum_{q,r=1}^nM_{at}(T,\xi,\bs)_{q,r}\D_{q,r}(\cC_l(w\be_p))}\prod_{j=1}^n(\psi^j)_{\tilde{\bX}^{m_j}}\Big|_{\psi^j=\b0\atop
 \forall j\in\{1,\cdots,n\}}\right|\le
c_0^{\frac{1}{2}\sum_{j=1}^nm_j}.
\end{equation*}
\end{lemma}
\begin{proof}
This can be proved by using \eqref{eq_simple_determinant_bound} and the
 properties of $M_{at}(T,\xi,\bs)$ and by repeating the same argument as in 
\cite[\mbox{Lemma 4.5}]{K1}.
\end{proof}

\begin{lemma}\label{lem_bound_tree_term_partial}
For any $m\in \{0,\cdots,N_{L,h}\}$, $n\in \N_{\ge 2}$ and
 $l\in\{N_{\beta},\cdots,N_h\}$,
\begin{align*}
\|T_{n,m}^{\ge l}\|_{1,\infty}&\le (1_{m=0}N_{L,h}/h+1_{m\ge
 1})2^{-3m}c_0^{-m}\frac{1}{n(n-1)}M^{-l(n-1)}\\
&\quad \cdot
 \prod_{j=1}^n\left(\sum_{m_j=1}^{N_{L,h}}2^{5m_j}c_0^{m_j}\|J_{m_j}^{\ge
 l+1}\|_{1,\infty}\right)1_{\sum_{j=1}^nm_j-n+1\ge m}.
\end{align*}
\end{lemma}
\begin{proof} For $T\in \T_n$ and a matrix
 $Q=(Q(X,Y))_{X,Y\in I_{L,h}}$ define the operator $O_{pe}(T,Q)$ on
 $\bigwedge \left(\left(\bigoplus_{j=1}^n\cV_j\right)\bigoplus
 \cV\right)$ by 
\begin{align}
O_{pe}(T,Q):=&\int_{[0,1]^{n-1}}d\bs\sum_{\xi\in
 \S_n(T)}\varphi(T,\xi,\bs)e^{\sum_{q,r=1}^nM_{at}(T,\xi,\bs)_{q,r}\D_{q,r}(Q)}\notag\\
&\cdot\prod_{\{q,r\}\in T}\left(\D_{q,r}(Q)+\D_{r,q}(Q)\right).\label{eq_tree_formula_operator}
\end{align}
It follows from the definition that 
\begin{align*}
T^{\ge l}_{n,m}(\psi)&=\prod_{j=1}^n\left(\sum_{m_j=1}^{N_{L,h}}\left(\frac{1}{h}\right)^{2m_j}\sum_{\bX^{m_j},\bY^{m_j}\in
 I_{L,h}^{m_j}}J^{\ge l+1}_{m_j}(\bX^{m_j},\bY^{m_j})\right)1_{\sum_{j=1}^nm_j-n+1\ge m}\\
&\quad
 \cdot\frac{1}{n!}\sum_{T\in\T_n}\cP_m\left(O_{pe}(T,\cC_l(w\be_p))\prod_{q=1}^n(\opsi^q+\opsi)_{\bX^{m_q}}(\psi^q+\psi)_{\bY^{m_q}}\Big|_{\psi^j=\b0\atop\forall
 j\in\{1,\cdots,n\}}\right).
\end{align*}
The constraint $1_{\sum_{j=1}^nm_j-n+1\ge m}$ is due to the fact that
 the operator \\
$\prod_{\{q,r\}\in
 T}\left(\D_{q,r}(\cC_l(w\be_p))+\D_{r,q}(\cC_l(w\be_p))\right)$ erases
 $n-1$ fields from $\prod_{j=1}^n(\opsi^j)_{\bX^{m_j}}$ and from
 $\prod_{j=1}^n(\psi^j)_{\bY^{m_j}}$, respectively. By using anti-symmetry,
\begin{align*}
&T_{n,m}^{\ge
 l}(\psi)\\
&=\prod_{j=1}^n\Bigg(\sum_{m_j=1}^{N_{L,h}}\left(\frac{1}{h}\right)^{2m_j}\sum_{\bX^{m_j},\bY^{m_j}\in
 I_{L,h}^{m_j}}\sum_{k_j,l_j=0}^{m_j}\sum_{\bW^{k_j}\subset\bX^{m_j}}\sum_{\bZ^{l_j}\subset\bY^{m_j}}\\
&\quad\cdot J_{m_j}^{\ge
 l+1}((\bW^{k_j},\bX^{m_j}\backslash
 \bW^{k_j}),(\bZ^{l_j},\bY^{m_j}\backslash \bZ^{l_j}))\Bigg)1_{\sum_{j=1}^nm_j-n+1\ge
 m}1_{\sum_{j=1}^nk_j=\sum_{j=1}^nl_j=m}\\
&\quad\cdot\frac{1}{n!}\sum_{T\in\T_n}\left(O_{pe}(T,\cC_l(w\be_p))\prod_{q=1}^n(\opsi)_{\bW^{k_q}}(\opsi^q)_{\bX^{m_q}\backslash
 \bW^{k_q}}(\psi)_{\bZ^{l_q}}(\psi^q)_{\bY^{m_q}\backslash
 \bZ^{l_q}}\Big|_{\psi^j=\b0\atop\forall
 j\in\{1,\cdots,n\}}\right)\\
&=\prod_{j=1}^n\left(\sum_{m_j=1}^{N_{L,h}}\sum_{k_j,l_j=0}^{m_j}\left(\begin{array}{c}m_j\\
  k_j\end{array}\right) \left(\begin{array}{c}m_j\\
			      l_j\end{array}\right)\right)1_{\sum_{j=1}^nm_j-n+1\ge m}1_{\sum_{j=1}^nk_j=\sum_{j=1}^nl_j=m}\\
&\quad\cdot\frac{1}{n!}\sum_{T\in\T_n}\prod_{q=1}^n\left(\left(\frac{1}{h}\right)^{k_q+l_q}\sum_{\bW^{k_q}\in
 I_{L,h}^{k_q}}\sum_{\bZ^{l_q}\in I_{L,h}^{l_q}}\right)\\
&\quad\cdot J_T^{(m_1,\cdots,m_n),(k_1,\cdots,k_n),(l_1,\cdots,l_n)}((\bW^{k_1},\cdots,\bW^{k_n}),(\bZ^{l_1},\cdots, \bZ^{l_n}))\prod_{r=1}^n(\opsi)_{\bW^{k_r}}\prod_{s=1}^n(\psi)_{\bZ^{l_s}},
\end{align*}
where
\begin{align}
&J_T^{(m_1,\cdots,m_n),(k_1,\cdots,k_n),(l_1,\cdots,l_n)}((\bW^{k_1},\cdots,\bW^{k_n}),(\bZ^{l_1},\cdots,
 \bZ^{l_n})):=\notag\\
&\prod_{j=1}^n\left(\left(\frac{1}{h}\right)^{2m_j-k_j-l_j}\sum_{\bX^{m_j-k_j}\in
 I_{L,h}^{m_j-k_j}}\sum_{\bY^{m_j-l_j}\in I_{L,h}^{m_j-l_j}}J_{m_j}^{\ge
 l+1}((\bW^{k_j},\bX^{m_j-k_j}),(\bZ^{l_j},\bY^{m_j-l_j}))\right)\notag\\
&\quad \cdot \eps_{(m_1,\cdots,m_n),(k_1,\cdots,k_n),(l_1,\cdots,l_n)}O_{pe}(T,\cC_l(w\be_p))\prod_{q=1}^n(\opsi^q)_{\bX^{m_q-k_q}}(\psi^q)_{\bY^{m_q-l_q}}\Big|_{\psi^j=\b0\atop\forall
 j\in\{1,\cdots,n\}},\label{eq_definition_sub_function}
\end{align}
with the factor
 $\eps_{(m_1,\cdots,m_n),(k_1,\cdots,k_n),(l_1,\cdots,l_n)}\in\{1,-1\}$
 depending only on $(m_1,\cdots,m_n)$, $(k_1,\cdots,k_n),(l_1,\cdots,l_n)$.
By \eqref{eq_anti_symmetric_estimate} and the triangle inequality of the norm $\|\cdot\|_{1,\infty}$,
\begin{align}
\|T^{\ge l}_{n,m}\|_{1,\infty}&\le
\prod_{j=1}^n\left(\sum_{m_j=1}^{N_{L,h}}\sum_{k_j,l_j=0}^{m_j}\left(\begin{array}{c}m_j\\
  k_j\end{array}\right) \left(\begin{array}{c}m_j\\
			      l_j\end{array}\right)\right)1_{\sum_{j=1}^nm_j-n+1\ge m}1_{\sum_{j=1}^nk_j=\sum_{j=1}^nl_j=m}\notag\\
&\quad\cdot\frac{1}{n!}\sum_{T\in\T_n}\|J_T^{(m_1,\cdots,m_n),(k_1,\cdots,k_n),(l_1,\cdots,l_n)}\|_{1,\infty}.\label{eq_application_triangle_inequality}
\end{align}

Let us find an upper bound on
 $\|J_T^{(m_1,\cdots,m_n),(k_1,\cdots,k_n),(l_1,\cdots,l_n)}\|_{1,\infty}$.
 Let $d_j$ denote the incidence number of the vertex $j$ in $T$. If $d_j> 2m_j-k_j-l_j$ for some
 $j\in\{1,\cdots,n\}$,
 $\|J_T^{(m_1,\cdots,m_n),(k_1,\cdots,k_n),(l_1,\cdots,l_n)}\|_{1,\infty}=0$,
 since in this case 
$$\prod_{\{q,r\}\in
 T}\left(\D_{q,r}(\cC_l(w\be_p))+\D_{r,q}(\cC_l(w\be_p))\right)
\prod_{s=1}^n(\opsi^s)_{\bX^{m_s-k_s}}(\psi^s)_{\bY^{m_s-l_s}}=0
$$
for any $\bX^{m_s-k_s}\in I_{L,h}^{m_s-k_s}$, $\bY^{m_s-l_s}\in
 I_{L,h}^{m_s-l_s}$ ($s=1,\cdots,n$). 

Assume that $d_j\le 2m_j-k_j-l_j$ $(\forall j\in \{1,\cdots,n\})$. First
 consider the case that $m\neq 0$. Let $q_0\in\{1,\cdots,n\}$ be a
 vertex with $k_{q_0}\neq 0$. For $q,r\in\{1,\cdots,n\}$ let
 $\dis_T(q,r)(\in \N\cup\{0\})$ denote the distance between the vertex
 $q$ and the vertex $r$ along the unique path connecting $q$ with $r$ in
 $T$. Define $L_r^q(T)\subset T$ by 
\begin{equation}\label{eq_tree_subset}
L_r^q(T):=\{\{r,s\}\in T\ |\ \dis_T(q,s)=\dis_T(q,r)+1\}.
\end{equation}
Note that if $d_r=1$ and $r\neq q_0$, then $L_r^{q_0}(T)=
 \emptyset$. If $d_r\neq 1$ or $r=q_0$, we can number each line of
 $L_r^{q_0}(T)$ so that 
\begin{align*}
&L_{q_0}^{q_0}(T)=\{\{q_0,s_1^{q_0}\},\{q_0,s_2^{q_0}\},\cdots,\{q_0,s_{d_{q_0}}^{q_0}\}\},\\
&L_{r}^{q_0}(T)=\{\{r,s_1^r\},\{r,s_2^r\},\cdots,\{r,s_{d_r-1}^r\}\}\quad(\forall
 r\in \{1,\cdots,n\}\backslash\{q_0\}\text{ with }d_r\neq 1).
\end{align*}
For any $\{q_0,s\}\in L^{q_0}_{q_0}(T)$ there uniquely exists
 $j\in \{1,2,\cdots,d_{q_0}\}$ such that
 $\{q_0,s\}=\{q_0,s_j^{q_0}\}$. For $\nu\in \S_{d_{q_0}}$ set
 $\nu(\{q_0,s\}):=\nu(j)$. Similarly for
 $r\in\{1,\cdots,n\}\backslash\{q_0\}$ with $d_r\neq 1$, $\{r,s\}\in L^{q_0}_r(T)$
 and $\nu\in \S_{d_r-1}$ let $\nu(\{r,s\})\in
 \{1,2,\cdots,d_r-1\}$ be defined by $\nu(\{r,s\}):=\nu(j)$, where
 $s=s_j^r$.

Moreover, define $\tilde{\cC}:\tilde{I}_{L,h}\times\tilde{I}_{L,h}\to
 \C$ by 
$$
\tilde{\cC}((X,u),(Y,v)):=\left\{\begin{array}{ll}0&\text{ if }u=v,\\
                                -\cC_l(X,Y)(w\be_p)&\text{ if
				 }u=1,v=-1,\\
 \cC_l(Y,X)(w\be_p)&\text{ if
				 }u=-1,v=1.\end{array}\right.
$$
By considering $q_0$ as the root of $T$ we see that
\begin{align}
&\prod_{\{q,r\}\in
 T}\left(\D_{q,r}(\cC_l(w\be_p))+\D_{r,q}(\cC_l(w\be_p))\right)\prod_{j=1}^n(\opsi^j)_{\bX^{m_j-k_j}}(\psi^j)_{\bY^{m_j-l_j}}\notag\\
&=\prod_{\{q,r\}\in T}\left(\sum_{\tilde{X},\tilde{Y}\in
 \tilde{I}_{L,h}}\tilde{\cC}(\tilde{X},\tilde{Y})\frac{\partial}{\partial\psi^q_{\tilde{X}}}\frac{\partial}{\partial\psi^r_{\tilde{Y}}}\right)\prod_{j=1}^n(\psi^j)_{\tilde{\bX}(1)^{m_j-k_j}}(\psi^j)_{\tilde{\bY}(-1)^{m_j-l_j}}\notag\\
&=\prod_{j=1\atop j\neq
 q_0}^n\Bigg(\sum_{\tilde{X}_j\in\tilde{I}_{L,h}\text{ with }\atop
 \tilde{X}_j\subset
 (\tilde{\bX}(1)^{m_j-k_j},\tilde{\bY}(-1)^{m_j-l_j})}\Bigg)\notag\\
&\quad\cdot\sum_{\tilde{\bX}^{d_{q_0}}\in\tilde{I}_{L,h}^{d_{q_0}}\text{ with }\atop
 \tilde{\bX}^{d_{q_0}}\subset
 (\tilde{\bX}(1)^{m_{q_0}-k_{q_0}},\tilde{\bY}(-1)^{m_{q_0}-l_{q_0}})}\sum_{\nu_{q_0}\in
 \S_{d_{q_0}}}\prod_{\{q_0,r\}\in
 L^{q_0}_{q_0}(T)}\tilde{\cC}(\tilde{X}_{\nu_{q_0}(\{q_0,r\})}^{d_{q_0}},\tilde{X}_r)\notag\\
&\quad\cdot\prod_{q=1\atop q\neq q_0\text{ and }d_q\neq 1}^n\Bigg(\sum_{\tilde{\bX}^{d_{q}-1}\in\tilde{I}_{L,h}^{d_{q}-1}\text{ with }\atop
 \tilde{\bX}^{d_{q}-1}\subset
 (\tilde{\bX}(1)^{m_{q}-k_{q}},\tilde{\bY}(-1)^{m_{q}-l_{q}})\backslash
 \tilde{X}_q}\sum_{\nu_{q}\in
 \S_{d_{q}-1}}\prod_{\{q,r\}\in
 L^{q_0}_{q}(T)}\tilde{\cC}(\tilde{X}_{\nu_{q}(\{q,r\})}^{d_{q}-1},\tilde{X}_r)
\Bigg)\notag\\
&\quad\cdot\eps_{\pm}\cdot(\psi^{q_0})_{(\tilde{\bX}(1)^{m_{q_0}-k_{q_0}},\tilde{\bY}(-1)^{m_{q_0}-l_{q_0}})\backslash\tilde{\bX}^{d_{q_0}}}\notag\\
&\quad\cdot\prod_{s=1\atop s\neq
 q_0}^n\big(1_{d_s\neq 1}(\psi^{s})_{((\tilde{\bX}(1)^{m_{s}-k_{s}},\tilde{\bY}(-1)^{m_{s}-l_{s}})\backslash
 \tilde{X}_s)\backslash \tilde{\bX}^{d_{s}-1}}
+1_{d_s=1}(\psi^{s})_{(\tilde{\bX}(1)^{m_{s}-k_{s}},\tilde{\bY}(-1)^{m_{s}-l_{s}})\backslash
 \tilde{X}_s}\big),\label{eq_expansion_tree_operator}
\end{align}
where $\eps_{\pm}=1$ or $-1$.
In the following let $\prod_{u=1\atop\text{ordered}}^vg_u$ denote
 $g_1g_2\cdots g_v$ for $v\in \N$. One finds this notation useful when
 each term $g_u$ depends on $g_1,g_2,\cdots,g_{u-1}$. Moreover, set 
$d(T,q_0):=\max_{1\le j\le n}\dis_T(q_0,j)$.  By
 substituting \eqref{eq_expansion_tree_operator} into
 \eqref{eq_definition_sub_function} and using
 \eqref{eq_nice_property_phi}, \eqref{eq_simple_decay_bound} and Lemma
 \ref{lem_determinant_bound_application} we have for any $W_{fixed}\in
 I_{L,h}$ that 
\begin{align}
&\left(\frac{1}{h}\right)^{2m-1}\sum_{\bW^{k_{q_0}-1}\in
 I_{L,h}^{k_{q_0}-1}}\sum_{\bZ^{l_{q_0}}\in
 I_{L,h}^{l_{q_0}}}\prod_{j=1\atop j\neq
 q_0}^n\left(\sum_{\bW^{k_j}\in I_{L,h}^{k_j}}\sum_{\bZ^{l_j}\in
 I_{L,h}^{l_j}}\right)\notag\\
&\ \cdot|J_T^{(m_1,\cdots,m_n),(k_1,\cdots,k_n),(l_1,\cdots,l_n)}((\bW^{k_1},\cdots,\bW^{k_{q_0-1}},(W_{fixed},\bW^{k_{q_0}-1}),\bW^{k_{q_0+1}},\cdots,\bW^{k_n}),\notag\\
&\quad(\bZ^{l_1},\cdots,\bZ^{l_n}))|\notag\\
&\le \left(\frac{1}{h}\right)^{2m_{q_0}-1}\sum_{\bW^{k_{q_0}-1}\in
 I_{L,h}^{k_{q_0}-1}}\sum_{\bZ^{l_{q_0}}\in
 I_{L,h}^{l_{q_0}}}\sum_{\bX^{m_{q_0}-k_{q_0}}\in
 I_{L,h}^{m_{q_0}-k_{q_0}}} \sum_{\bY^{m_{q_0}-l_{q_0}}\in
 I_{L,h}^{m_{q_0}-l_{q_0}}}\notag\\
&\quad\cdot |J_{m_{q_0}}^{\ge
 l+1}((W_{fixed},\bW^{k_{q_0}-1},\bX^{m_{q_0}-k_{q_0}}),(\bZ^{l_{q_0}},\bY^{m_{q_0}-l_{q_0}}))|\notag\\
&\quad\cdot \sum_{\tilde{\bX}^{d_{q_0}}\in
 \tilde{I}_{L,h}^{d_{q_0}}\text{ with }\atop
 \tilde{\bX}^{d_{q_0}}\subset
 (\tilde{\bX}(1)^{m_{q_0}-k_{q_0}},\tilde{\bY}(-1)^{m_{q_0}-l_{q_0}})}\sum_{\nu_{q_0}\in
 \S_{d_{q_0}}}\prod_{\{q_0,r\}\in L^{q_0}_{q_0}(T)}\notag\\
&\quad\cdot\Bigg(
\left(\frac{1}{h}\right)^{2m_r}\sum_{\bW^{k_r}\in I_{L,h}^{k_r}}\sum_{\bZ^{l_r}\in
 I_{L,h}^{l_r}}\sum_{\bX^{m_r-k_r}\in
 I_{L,h}^{m_r-k_r}}\sum_{\bY^{m_r-l_r}\in
 I_{L,h}^{m_r-l_r}}\sum_{\tilde{X}_r\in\tilde{I}_{L,h}\text{ with }\atop
\tilde{X}_r\subset
 (\tilde{\bX}(1)^{m_r-k_r},\tilde{\bY}(-1)^{m_r-l_r})}\notag\\
&\quad\cdot |J_{m_r}^{\ge
 l+1}((\bW^{k_r},\bX^{m_r-k_r}),(\bZ^{l_r},\bY^{m_r-l_r}))||\tilde{\cC}(\tilde{X}^{d_{q_0}}_{\nu_{q_0}(\{q_0,r\})},\tilde{X}_r)|\Bigg)\notag\\
&\quad \cdot \prod_{u=1\atop \text{ordered}}^{d(T,q_0)-1}\Bigg(\prod_{j\in\{1,\cdots,n\}\text{ with
 }\atop\dis_T(q_0,j)=u\text{ and }d_j\neq 1}\Bigg(\sum_{\tilde{\bX}^{d_j-1}\in
 \tilde{I}_{L,h}^{d_j-1}\text{ with }\atop \tilde{\bX}^{d_j-1}\subset
 (\tilde{\bX}(1)^{m_j-k_j},\tilde{\bY}(-1)^{m_j-l_j})\backslash
 \tilde{X}_j}\sum_{\nu_j\in \S_{d_j-1}}\prod_{\{j,r\}\in L_j^{q_0}(T)}\notag\\
&\quad\cdot\Bigg(\left(\frac{1}{h}\right)^{2m_r}\sum_{\bW^{k_r}\in I_{L,h}^{k_r}}\sum_{\bZ^{l_r}\in
 I_{L,h}^{l_r}}\sum_{\bX^{m_r-k_r}\in
 I_{L,h}^{m_r-k_r}}\sum_{\bY^{m_r-l_r}\in
 I_{L,h}^{m_r-l_r}}\sum_{\tilde{X}_r\in\tilde{I}_{L,h}\text{ with }\atop
 \tilde{X}_r\subset
 (\tilde{\bX}(1)^{m_r-k_r},\tilde{\bY}(-1)^{m_r-l_r})}\notag\\
&\quad\cdot |J_{m_r}^{\ge
 l+1}((\bW^{k_r},\bX^{m_r-k_r}),(\bZ^{l_r},\bY^{m_r-l_r}))||\tilde{\cC}(\tilde{X}_{\nu_j(\{j,r\})}^{d_j-1},\tilde{X}_r)|\Bigg)\Bigg)\Bigg)\notag\\
&\quad\cdot \Bigg|\int_{[0,1]^{n-1}}d\bs\sum_{\xi\in
 \S_n(T)}\varphi(T,\xi,\bs)e^{\sum_{q,r=1}^nM_{at}(T,\xi,\bs)_{q,r}\D_{q,r}(\cC_l(w\be_p))}\notag\\
&\quad\cdot (\psi^{q_0})_{(\tilde{\bX}(1)^{m_{q_0}-k_{q_0}},\tilde{\bY}(-1)^{m_{q_0}-l_{q_0}})\backslash
 \tilde{\bX}^{d_{q_0}}}\prod_{s=1\atop s\neq q_0}^n\big(1_{d_s\neq 1}(\psi^{s})_{((\tilde{\bX}(1)^{m_{s}-k_{s}},\tilde{\bY}(-1)^{m_{s}-l_{s}})\backslash
 \tilde{X}_s)\backslash \tilde{\bX}^{d_{s}-1}}\notag\\
&\quad+1_{d_s=1}(\psi^{s})_{(\tilde{\bX}(1)^{m_{s}-k_{s}},\tilde{\bY}(-1)^{m_{s}-l_{s}})\backslash
 \tilde{X}_s}\big)\Big|_{\psi^j=\b0\atop\forall
 j\in\{1,\cdots,n\}}\Bigg|\notag\\
&\le \|J_{m_{q_0}}^{\ge
 l+1}\|_{1,\infty}\left(\begin{array}{c}2m_{q_0}-k_{q_0}-l_{q_0}\\
			d_{q_0}\end{array}\right)d_{q_0}!\notag\\
&\quad\cdot\prod_{\{q_0,r\}\in L^{q_0}_{q_0}(T)}\Big((2m_r-k_r-l_r)\|J_{m_r}^{\ge l+1}\|_{1,\infty}c_0M^{-l}\Big)\notag\\
&\quad\cdot \prod_{u=1}^{d(T,q_0)-1}\Bigg(\prod_{j\in\{1,\cdots,n\}\text{ with }\atop
 \dis_T(q_0,j)=u\text{ and }d_j\neq 1}\Bigg(\left(\begin{array}{c}2m_j-k_j-l_j-1\\
			       d_j-1\end{array}\right)(d_j-1)!\notag\\
&\quad\cdot\prod_{\{j,r\}\in L_j^{q_0}(T)}((2m_r-k_r-l_r)\|J^{\ge l+1}_{m_r}\|_{1,\infty}c_0M^{-l})\Bigg)\Bigg)c_0^{\frac{1}{2}\sum_{q=1}^n(2m_q-k_q-l_q-d_q)}\notag\\
&=(c_0M^{-l})^{n-1}\notag\\
&\quad\cdot\prod_{j=1}^n\left(c_0^{\frac{1}{2}(2m_j-k_j-l_j-d_j)}\|J_{m_j}^{\ge
 l+1}\|_{1,\infty}(2m_j-k_j-l_j)\left(\begin{array}{c}2m_j-k_j-l_j-1\\
				      d_j-1\end{array}\right)(d_j-1)!\right).
\label{eq_bound_sub_function}
\end{align}
By arbitrariness of $q_0$ and the fixed variable $W_{fixed}$,
 $\|J_T^{(m_1,\cdots,m_n),(k_1,\cdots,k_n),(l_1,\cdots,l_n)}\|_{1,\infty}$
 can be bounded by the right-hand side of
 \eqref{eq_bound_sub_function}.

In the case that $m=0$ we fix any $q_0\in\{1,\cdots,n\}$ and repeat the
 same calculation as above by setting $k_j$, $l_j$ to be 0 for all $j\in
 \{1,\cdots,n\}$. The only difference in the consequence is
 that $\|J_{m_{q_0}}^{\ge l+1}\|_1$ comes in place of
 $\|J_{m_{q_0}}^{\ge l+1}\|_{1,\infty}$. Since $\|J_{m_{q_0}}^{\ge
 l+1}\|_1\le (N_{L,h}/h)\|J_{m_{q_0}}^{\ge l+1}\|_{1,\infty}$, we only
 need to multiply the right-hand side of \eqref{eq_bound_sub_function}
by the extra factor $N_{L,h}/h$ in this case. 

By substituting these results into
 \eqref{eq_application_triangle_inequality}, replacing the sum over
 trees by the sum over possible incidence numbers and using Cayley's
 theorem on the number of trees with fixed incidence numbers, we can
 deduce that
\begin{align}
&\|T_{n,m}^{\ge l}\|_{1,\infty}\le \prod_{j=1}^n\left(\sum_{m_j=1}^{N_{L,h}}
\sum_{k_j,l_j=0}^{m_j}\left(\begin{array}{c}m_j\\ k_j\end{array}
\right)\left(\begin{array}{c}m_j\\
	     l_j\end{array}\right)\right)1_{\sum_{j=1}^nm_j-n+1\ge m}
1_{\sum_{j=1}^nk_j=\sum_{j=1}^nl_j=m}\notag\\
&\cdot\frac{1}{n!}\prod_{q=1}^n\left(\sum_{d_q=1}^{2m_q-k_q-l_q}\right)
1_{\sum_{q=1}^nd_q=2(n-1)}\frac{(n-2)!}{\prod_{s=1}^n(d_s-1)!}(1_{m=0}N_{L,h}/h+1_{m\ge
 1 })(c_0M^{-l})^{n-1}\notag\\
&\cdot\prod_{r=1}^n\left(c_0^{\frac{1}{2}(2m_r-k_r-l_r-d_r)}\|J_{m_r}^{\ge l+1}\|_{1,\infty}(2m_r-k_r-l_r)\left(\begin{array}{c}
									   2m_r-k_r-l_r-1\\ d_r-1\end{array}\right)(d_r-1)!\right)\notag\\
&=(1_{m=0}N_{L,h}/h+1_{m\ge
 1 })c_0^{-m}\frac{1}{n(n-1)}M^{-l(n-1)}\notag\\
&\cdot\prod_{j=1}^n\left(\sum_{m_j=1}^{N_{L,h}}
\sum_{k_j,l_j=0}^{m_j}\left(\begin{array}{c}m_j\\ k_j\end{array}
\right)\left(\begin{array}{c}m_j\\
	     l_j\end{array}\right)c_0^{m_j}\|J_{m_j}^{\ge
 l+1}\|_{1,\infty}\right)
1_{\sum_{j=1}^nm_j-n+1\ge m}1_{\sum_{j=1}^nk_j=\sum_{j=1}^nl_j=m}\notag
\\
&\cdot\prod_{q=1}^n\left((2m_q-k_q-l_q)\sum_{d_q=1}^{2m_q-k_q-l_q}\left(\begin{array}{c}2m_q-k_q-l_q-1\\
								  d_q-1\end{array}\right)\right)1_{\sum_{q=1}^nd_q=2(n-1)}.\label{eq_application_cayley}
\end{align}
By using the inequality that $2m_q-k_q-l_q\le 2^{m_q-(k_q+l_q)/2+1}$
 one has 
\begin{equation}\label{eq_sum_incidence}
\prod_{q=1}^n\left((2m_q-k_q-l_q)\sum_{d_q=1}^{2m_q-k_q-l_q}\left(\begin{array}{c}2m_q-k_q-l_q-1\\
								  d_q-1\end{array}\right)\right)\le 2^{3\sum_{q=1}^nm_q-3m}.
\end{equation}
By combining \eqref{eq_sum_incidence} with
 \eqref{eq_application_cayley}, dropping the constraints
 $1_{\sum_{j=1}^nk_j=\sum_{j=1}^nl_j=m}$,\\
$1_{\sum_{q=1}^nd_q=2(n-1)}$
 and summing over $k_j, l_j$ $(j=1,\cdots,n)$ we obtain
 the claimed upper bound.
\end{proof}

The following lemma will not be used until Subsection
\ref{subsec_final_integration}. Since its proof is close to the proof of
Lemma \ref{lem_bound_tree_term_partial}, let us show at this
point.
\begin{lemma}\label{lem_bound_tree_term_partial_derivative}
For any $m,m'\in\{0,\cdots,N_{L,h}\}$, $n\in\N_{\ge 2}$, $b\in
 \{free,tree\}$, $\bX^{m'},\bY^{m'}$ $\in I_{L,h}^{m'}$ and $l\in\{N_{\beta},\cdots,N_h\}$,
\begin{align*}
\left\|\frac{\partial T_{n,m}^{\ge l}}{\partial J_{b,m'}^{\ge
 l+1}(\bX^{m'},\bY^{m'})}\right\|_{1}&\le 1_{m'\ge
 1}(m'!)^2h^{-2m'}2^{5m'-3m}c_0^{m'-m}M^{-l(n-1)}\\
&\quad\cdot\prod_{j=1}^{n-1}\left(\sum_{m_j=1}^{N_{L,h}}2^{5m_j}c_0^{m_j}\|J_{m_j}^{\ge
 l+1}\|_{1,\infty}\right)1_{\sum_{j=1}^{n-1}m_j+m'-n+1\ge m}.
\end{align*}
\end{lemma}
\begin{proof}
By using anti-symmetry,
\begin{align}
&\frac{\partial T_{n,m}^{\ge l}(\psi)}{\partial J_{b,m'}^{\ge
 l+1}(\bX^{m'},\bY^{m'})}\notag\\
&=h^{-2m'}(m'!)^2\sum_{j_0=1}^n
\prod_{j=1\atop j\neq j_0}^n\left(\sum_{m_j=1}^{N_{L,h}}\left(\frac{1}{h}\right)^{2m_j}\sum_{\bX^{m_j},\bY^{m_j}\in
 I_{L,h}^{m_j}}J^{\ge
 l+1}_{m_j}(\bX^{m_j},\bY^{m_j})\right)\notag\\
&\quad\cdot 1_{\sum_{j=1,j\neq j_0}^nm_j+m'-n+1\ge
 m}\frac{1}{n!}\sum_{T\in\T_n}\cP_m\Bigg(O_{pe}(T,\cC_l(w\be_p))\notag\\
&\quad\cdot(\opsi^{j_0}+\opsi)_{\bX^{m'}}(\psi^{j_0}+\psi)_{\bY^{m'}}\prod_{q=1\atop q\neq j_0}^n(\opsi^q+\opsi)_{\bX^{m_q}}(\psi^q+\psi)_{\bY^{m_q}}\Big|_{\psi^j=\b0\atop\forall
 j\in\{1,\cdots,n\}}\Bigg).\label{eq_summation_derivative_pre}
\end{align}
We can see from \eqref{eq_summation_derivative_pre} that $\partial T_{n,m}^{\ge l}(\psi)/\partial J_{b,m'}^{\ge
 l+1}(\bX^{m'},\bY^{m'})=0$ if $m'=0$, since \\
$\prod_{\{q,r\}\in T}(\D_{q,r}(\cC_l(w\be_p))+\D_{r,q}(\cC_l(w\be_p)))\prod_{j=1,j\neq j_0}^n(\opsi^j+\opsi)_{\bX^{m_j}}(\psi^j+\psi)_{\bY^{m_j}}=0$.
The equality \eqref{eq_summation_derivative_pre} leads to 
\begin{align}
&\frac{\partial T_{n,m}^{\ge l}(\psi)}{\partial J_{b,m'}^{\ge
 l+1}(\bX^{m'},\bY^{m'})}\notag\\
&=h^{-2m'}(m'!)^2\sum_{j_0=1}^n\sum_{m_{j_0}=1}^{N_{L,h}}\sum_{k_{j_0},l_{j_0}=0}^{m_{j_0}}1_{m_{j_0}=m'}\prod_{j=1\atop j\neq j_0}^n\left(\sum_{m_j=1}^{N_{L,h}}\sum_{k_j,l_j=0}^{m_j}\left(\begin{array}{c}m_j\\
  k_j\end{array}\right) \left(\begin{array}{c}m_j\\
			      l_j\end{array}\right)\right)\notag\\
&\quad\cdot 1_{\sum_{j=1}^nm_j-n+1\ge m}1_{\sum_{j=1}^nk_j=\sum_{j=1}^nl_j=m}\frac{1}{n!}\sum_{T\in\T_n}\prod_{q=1}^n\left(\left(\frac{1}{h}\right)^{k_q+l_q}\sum_{\bW^{k_q}\in
 I_{L,h}^{k_q}}\sum_{\bZ^{l_q}\in I_{L,h}^{l_q}}\right)\notag\\
&\quad\cdot J_{T,j_0}^{(m_1,\cdots,m_n),(k_1,\cdots,k_n),(l_1,\cdots,l_n)}((\bW^{k_1},\cdots,\bW^{k_n}),(\bZ^{l_1},\cdots, \bZ^{l_n}))\prod_{r=1}^n(\opsi)_{\bW^{k_r}}\prod_{s=1}^n(\psi)_{\bZ^{l_s}},\notag
\end{align}
where
\begin{align*}
&J_{T,j_0}^{(m_1,\cdots,m_n),(k_1,\cdots,k_n),(l_1,\cdots,l_n)}((\bW^{k_1},\cdots,\bW^{k_n}),(\bZ^{l_1},\cdots,
 \bZ^{l_n}))\\
&:=h^{k_{j_0}+l_{j_0}}1_{\bW^{k_{j_0}}\subset
 \bX^{m'}}1_{\bZ^{l_{j_0}}\subset \bY^{m'}}\prod_{j=1\atop j\neq
 j_0}^n\Bigg(\left(\frac{1}{h}\right)^{2m_j-k_j-l_j}\sum_{\bX^{m_j-k_j}\in
 I_{L,h}^{m_j-k_j}}\sum_{\bY^{m_j-l_j}\in I_{L,h}^{m_j-l_j}}\\
&\quad\cdot J_{m_j}^{\ge
 l+1}((\bW^{k_j},\bX^{m_j-k_j}),(\bZ^{l_j},\bY^{m_j-l_j}))\Bigg)\eps_{\bW^{k_{j_0}},\bX^{m'},\bZ^{l_{j_0}},\bY^{m'},j_0,(m_1,\cdots,m_n),(k_1,\cdots,k_n),(l_1,\cdots,l_n)}\\
&\quad\cdot
 O_{pe}(T,\cC_l(w\be_p))(\opsi^{j_0})_{\bX^{m'}\backslash \bW^{k_{j_0}}}(\psi^{j_0})_{\bY^{m'}\backslash\bZ^{l_{j_0}}}\prod_{q=1\atop q\neq j_0}^n(\opsi^q)_{\bX^{m_q-k_q}}(\psi^q)_{\bY^{m_q-l_q}}\Big|_{\psi^j=\b0\atop\forall
 j\in\{1,\cdots,n\}},
\end{align*}
with the factor
 $\eps_{\bW^{k_{j_0}},\bX^{m'},\bZ^{l_{j_0}},\bY^{m'},j_0,(m_1,\cdots,m_n),(k_1,\cdots,k_n),(l_1,\cdots,l_n)}\in\{1,-1\}$
 depending only on $\bW^{k_{j_0}},\bX^{m'},\bZ^{l_{j_0}},\bY^{m'}$, $j_0,(m_1,\cdots,m_n),(k_1,\cdots,k_n),(l_1,\cdots,l_n)$.
It follows from \eqref{eq_anti_symmetric_estimate} and the triangle
 inequality of the norm $\|\cdot\|_1$ that
\begin{align}
&\left\|\frac{\partial T_{n,m}^{\ge l}}{\partial J_{b,m'}^{\ge
 l+1}(\bX^{m'},\bY^{m'})}\right\|_{1}\notag\\
&\le
h^{-2m'}(m'!)^2\sum_{j_0=1}^n\sum_{m_{j_0}=1}^{N_{L,h}}\sum_{k_{j_0},l_{j_0}=0}^{m_{j_0}}1_{m_{j_0}=m'}\prod_{j=1\atop j\neq j_0}^n\left(\sum_{m_j=1}^{N_{L,h}}\sum_{k_j,l_j=0}^{m_j}\left(\begin{array}{c}m_j\\
  k_j\end{array}\right) \left(\begin{array}{c}m_j\\
			      l_j\end{array}\right)\right)\notag\\
&\quad\cdot 1_{\sum_{j=1}^nm_j-n+1\ge
 m}1_{\sum_{j=1}^nk_j=\sum_{j=1}^nl_j=m}\frac{1}{n!}\sum_{T\in\T_n}\|J_{T,j_0}^{(m_1,\cdots,m_n),(k_1,\cdots,k_n),(l_1,\cdots,l_n)}\|_{1}.\label{eq_application_triangle_inequality_derivative}
\end{align}
The estimation of
 $\|J_{T,j_0}^{(m_1,\cdots,m_n),(k_1,\cdots,k_n),(l_1,\cdots,l_n)}\|_{1}$
 is parallel to that of\\
 $\|J_{T}^{(m_1,\cdots,m_n),(k_1,\cdots,k_n),(l_1,\cdots,l_n)}\|_{1,\infty}$
 in the proof of Lemma \ref{lem_bound_tree_term_partial}. Here we
 consider $j_0$ as the root of $T$, while this role was played by the
 vertex $q_0$ in the previous lemma.
By noting that 
$$
\left(\frac{1}{h}\right)^{k_{j_0}+l_{j_0}}\sum_{\bW^{k_{j_0}}\in
 I_{L,h}^{k_{j_0}}}\sum_{\bZ^{l_{j_0}}\in
 I_{L,h}^{l_{j_0}}}(h^{k_{j_0}+l_{j_0}}1_{\bW^{k_{j_0}}\subset
 \bX^{m'}}1_{\bZ^{l_{j_0}}\subset
		  \bY^{m'}})=\left(\begin{array}{c}m_{j_0}\\ k_{j_0}\end{array}
\right)\left(\begin{array}{c}m_{j_0}\\ l_{j_0}\end{array}
\right)
$$
and letting $d_1,\cdots,d_n$ be the incidence numbers of $T$ we have
\begin{align}
&\|J_{T,j_0}^{(m_1,\cdots,m_n),(k_1,\cdots,k_n),(l_1,\cdots,l_n)}\|_{1}\notag\\
&\le 1_{d_j\le 2m_j-k_j-l_j\ (\forall j\in \{1,\cdots,n\})}
\left(\begin{array}{c}m_{j_0}\\ k_{j_0}\end{array}
\right)\left(\begin{array}{c}m_{j_0}\\ l_{j_0}\end{array}
\right)(c_0M^{-l})^{n-1}\prod_{q=1\atop q\neq j_0}^n\|J_{m_q}^{\ge
 l+1}\|_{1,\infty}\notag\\
&\quad\cdot\prod_{j=1}^n\left(c_0^{\frac{1}{2}(2m_j-k_j-l_j-d_j)}(2m_j-k_j-l_j)\left(\begin{array}{c}2m_j-k_j-l_j-1\\
				      d_j-1\end{array}\right)(d_j-1)!\right)
.\label{eq_bound_sub_function_derivative}
\end{align}
By returning the right-hand side of
 \eqref{eq_bound_sub_function_derivative} to \eqref{eq_application_triangle_inequality_derivative} and replacing the
 sum over trees by the sum over possible incidence numbers we obtain
\begin{align*}
&\left\|\frac{\partial T_{n,m}^{\ge l}}{\partial J_{b,m'}^{\ge
 l+1}(\bX^{m'},\bY^{m'})}\right\|_{1}
\le 1_{m'\ge 1}h^{-2m'}(m'!)^2\sum_{j_0=1}^n
c_0^{-m}\frac{1}{n(n-1)}M^{-l(n-1)}\\
&\quad\cdot\prod_{j=1}^n\left(\sum_{m_j=1}^{N_{L,h}}
\sum_{k_j,l_j=0}^{m_j}\left(\begin{array}{c}m_j\\ k_j\end{array}
\right)\left(\begin{array}{c}m_j\\
	     l_j\end{array}\right)c_0^{m_j}\right)\\
&\quad\cdot1_{m_{j_0}=m'}1_{\sum_{j=1}^nm_j-n+1\ge m}1_{\sum_{j=1}^nk_j=\sum_{j=1}^nl_j=m}\prod_{r=1\atop r\neq j_0}^n\|J_{m_r}^{\ge
 l+1}\|_{1,\infty}\\
&\quad\cdot\prod_{q=1}^n\left((2m_q-k_q-l_q)\sum_{d_q=1}^{2m_q-k_q-l_q}\left(\begin{array}{c}2m_q-k_q-l_q-1\\
								  d_q-1\end{array}\right)\right)1_{\sum_{q=1}^nd_q=2(n-1)}.
\end{align*}
Then, the same calculation as in the last part of the proof of Lemma
 \ref{lem_bound_tree_term_partial} yields the claimed upper bound.
\end{proof}

For compactness of the argument we assume the condition
\eqref{eq_smallness_assumption} throughout this section. The following
lemma itself, however, can be proved under a weaker condition.
\begin{lemma}\label{lem_inductive_bound_each_tree_integration}
Fix any $l\in \{N_{\beta}, \cdots,N_h\}$. Assume that
 \eqref{eq_smallness_assumption} and
 \eqref{eq_inductive_bound_each_scale} for $l+1$ hold. Then, $T^{\ge
 l}(\psi)$ is well-defined in $\bigoplus_{n=0}^{N_{L,h}}\cP_n \bigwedge\cV$. Moreover, the following
 inequalities hold true.
\begin{align}
&\|T_m^{\ge l}\|_{1,\infty}\le (1_{m=0}N_{L,h}/h+1_{m\ge
     1})2^{-3m}c_0^{-m}M^{N_{\beta}}(2^6\alpha^{-1})^2M^{-(l-N_{\beta})}\label{eq_item_tree_term_norm}\\
&\quad (\forall m\in \{0,\cdots,N_{L,h}\}).\notag\\
&M^{-N_{\beta}}\sum_{m=1}^{N_{L,h}}\alpha^mc_0^mM^{(l-N_{\beta})(m-2)}\|T_m^{\ge l}\|_{1,\infty}\le 2^8\alpha^{-1}M^2.\label{eq_item_tree_term_sum}\\
&
 M^{-N_{\beta}}\sum_{m=3}^{N_{L,h}}\alpha^mc_0^mM^{(l-N_{\beta})(m-2)}\|F_m^{\ge l}\|_{1,\infty}\le
     2^7M^{-1}.\label{eq_item_free_term_norm_sum}\\
&M^{-N_{\beta}}\sum_{m=3}^{N_{L,h}}\left(\begin{array}{c}m\\j\end{array}\right)^2c_0^m\|J_m^{\ge
 l+1}\|_{1,\infty}\le
     (2^2\alpha^{-1})^3M^{-(l+1-N_{\beta})}\quad (\forall j\in\{1,2,3\}).\label{eq_item_full_term_sum}
\end{align}
\end{lemma}
\begin{proof}
Proof of \eqref{eq_item_tree_term_norm}: The assumptions ensure that
 $$M^{-N_{\beta}}2^5c_0\|J_1^{\ge l+1}\|_{1,\infty}\le 2^5 \alpha^{-1},\quad
 M^{-N_{\beta}}\sum_{m=2}^{N_{L,h}}2^{5m}c_0^m\|J_m^{\ge
 l+1}\|_{1,\infty}\le (2^5\alpha^{-1})^2,$$ 
which result in
\begin{equation}\label{eq_useful_inside}
M^{-N_{\beta}}\sum_{m=1}^{N_{L,h}}2^{5m}c_0^m\|J_m^{\ge
 l+1}\|_{1,\infty}\le 2^6\alpha^{-1}.
\end{equation}
By substituting \eqref{eq_useful_inside} into the upper bound obtained
 in Lemma \ref{lem_bound_tree_term_partial} we have
\begin{align*}
&\sum_{n=2}^{\infty}
\|T_{n,m}^{\ge l}\|_{1,\infty}\\
&\le (1_{m=0}N_{L,h}/h+1_{m\ge
 1})2^{-3m}c_0^{-m}M^{N_{\beta}}\sum_{n=2}^{\infty}\frac{1}{n(n-1)}M^{-(l-N_{\beta})(n-1)}(2^6\alpha^{-1})^n\\
&\le (1_{m=0}N_{L,h}/h+1_{m\ge
 1})2^{-3m}c_0^{-m}M^{N_{\beta}}M^{-(l-N_{\beta})}(2^6\alpha^{-1})^2,
\end{align*}
where we used that $\sum_{n=2}^{\infty}\frac{1}{n(n-1)}=1$. This implies
 the well-definedness of $T^{\ge l}(\psi)$ and \eqref{eq_item_tree_term_norm}.

Proof of \eqref{eq_item_tree_term_sum}: It follows from Lemma
 \ref{lem_bound_tree_term_partial} and the inequalities \\
$2^{-3}\alpha M^{l-N_{\beta}}(2^{-3}\alpha M^{l-N_{\beta}}-1)^{-1}\le 2$ and
 $2^{2m}M^{-m}\le 2^2M^{-1}$ $(\forall m\in \N)$ that
\begin{align*}
&M^{-N_{\beta}}\sum_{m=1}^{N_{L,h}}\alpha^mc_0^mM^{(l-N_{\beta})(m-2)}\|T_m^{\ge
 l}\|_{1,\infty}\\
&\le
 M^{-N_{\beta}}M^{-2(l-N_{\beta})}\sum_{n=2}^{\infty}\frac{1}{n(n-1)}M^{-l(n-1)}\\
&\quad\cdot\prod_{j=1}^n\left(\sum_{m_j=1}^{N_{L,h}}2^{5m_j}c_0^{m_j}\|J_{m_j}^{\ge
 l+1}\|_{1,\infty}\right)\sum_{m=1}^{\sum_{q=1}^nm_q-n+1}2^{-3m}\alpha^mM^{(l-N_{\beta})m}\\
&\le 2\sum_{n=2}^{\infty}\frac{1}{n(n-1)}M^{-(l-N_{\beta})(n+1)}\\
&\quad\cdot\prod_{j=1}^n\left(M^{-N_{\beta}}\sum_{m_j=1}^{N_{L,h}}2^{5m_j}c_0^{m_j}\|J_{m_j}^{\ge
 l+1}\|_{1,\infty}\right)(2^{-3}\alpha
 M^{l-N_{\beta}})^{\sum_{q=1}^nm_q-n+1}\\
&=
 2^{-2}\alpha\sum_{n=2}^{\infty}\frac{1}{n(n-1)}\left(2^{3}\alpha^{-1}M^{-N_{\beta}}\sum_{m=1}^{N_{L,h}}2^{2m}\alpha^mc_0^mM^{(l-N_{\beta})(m-2)}\|J_m^{\ge
 l+1}\|_{1,\infty}\right)^n\\
&\le 2^{-2}\alpha\sum_{n=2}^{\infty}\frac{1}{n(n-1)}\left(2^{5}\alpha^{-1}M^{1-N_{\beta}}\sum_{m=1}^{N_{L,h}}\alpha^mc_0^mM^{(l+1-N_{\beta})(m-2)}\|J_m^{\ge
 l+1}\|_{1,\infty}\right)^n\\
&\le 2^{-2}\alpha
 \sum_{n=2}^{\infty}\frac{1}{n(n-1)}(2^5\alpha^{-1}M)^n\le
2^{8}\alpha^{-1}M^2.\end{align*}

Proof of \eqref{eq_item_free_term_norm_sum}: One can see from the definition of
 $F_m^{\ge l}(\psi)$ and \eqref{eq_simple_determinant_bound} that 
$$
\|F_m^{\ge l}\|_{1,\infty}\le
			   \sum_{j=m}^{N_{L,h}}\left(\begin{array}{c}j\\m\end{array}\right)^2c_0^{j-m}\|J_j^{\ge l+1}\|_{1,\infty}\le \sum_{j=m}^{N_{L,h}}2^{2j}c_0^{j-m}\|J_j^{\ge l+1}\|_{1,\infty}.
$$
Substituting this inequality and using the inequalities $\alpha
 M^{l-N_{\beta}}(\alpha M^{l-N_{\beta}}-1)^{-1}\le 2$ and
 $2^{2j}M^{-j}\le 2^6M^{-3}$ $(\forall j\in \N_{\ge 3})$ yield that
\begin{align*}
&M^{-N_{\beta}}\sum_{m=3}^{N_{L,h}}\alpha^mc_0^mM^{(l-N_{\beta})(m-2)}\|F_m^{\ge
 l}\|_{1,\infty}\\
&\le M^{-N_{\beta}}\sum_{j=3}^{N_{L,h}}\sum_{m=3}^{j}2^{2j}\alpha^mc_0^jM^{(l-N_{\beta})(m-2)}\|J_j^{\ge
 l+1}\|_{1,\infty}\\
&\le 2 M^{-N_{\beta}}\sum_{j=3}^{N_{L,h}}2^{2j}\alpha^jc_0^jM^{(l-N_{\beta})(j-2)}\|J_j^{\ge l+1}\|_{1,\infty}\\
&\le
 2^7M^{-1-N_{\beta}}\sum_{j=3}^{N_{L,h}}\alpha^jc_0^jM^{(l+1-N_{\beta})(j-2)}\|J_j^{\ge
 l+1}\|_{1,\infty}\le 2^7 M^{-1}.
\end{align*}

Proof of \eqref{eq_item_full_term_sum}: By using the inequalities that
 $2^{2m}\alpha^{-m}\le (2^2\alpha^{-1})^3$ and\\
 $M^{l+1-N_{\beta}} \le M^{(l+1-N_{\beta})(m-2)}$ $(\forall m\in \N_{\ge
 3})$,
\begin{align*}
&M^{-N_{\beta}}\sum_{m=3}^{N_{L,h}}\left(\begin{array}{c}m\\j\end{array}\right)^2c_0^m\|J_m^{\ge
 l+1}\|_{1,\infty}\le
 M^{-N_{\beta}}\sum_{m=3}^{N_{L,h}}2^{2m}c_0^m\|J_m^{\ge
 l+1}\|_{1,\infty}\\
&\le (2^2\alpha^{-1})^3M^{-(l+1-N_{\beta})}M^{-N_{\beta}}\sum_{m=3}^{N_{L,h}}\alpha^mc_0^mM^{(l+1-N_{\beta})(m-2)}\|J_m^{\ge
 l+1}\|_{1,\infty}\\
&\le (2^2\alpha^{-1})^3M^{-(l+1-N_{\beta})}.
\end{align*}
\end{proof}

Proposition \ref{prop_inductive_bound} can be proved by repeatedly
using the inequalities of Lemma \ref{lem_inductive_bound_each_tree_integration}.
\begin{proof}[Proof of Proposition \ref{prop_inductive_bound}]
The proof is made by induction on $l\in
 \{N_{\beta}+1,\cdots,N_h+1\}$. Set
 $U_{max}:=\max\{|\la_1|,|\la_{-1}|,|U_c|,|U_o|\}$. 
The bi-anti-symmetric kernel $F_2^{\ge
 N_h+1}(\cdot,\cdot)$ can be written as follows.
\begin{equation*}
\begin{split}
&F_2^{\ge
 N_h+1}((\rho_1,\bx_1,\s_1,x_1),(\rho_2,\bx_2,\s_2,x_2),(\eta_1,\by_1,\tau_1,y_1),(\eta_2,\by_2,\tau_2,y_2))\\
&=-h^31_{x_1=x_2=y_1=y_2}U_{(\la_1,\la_{-1})}((\rho_1,\bx_1,\s_1),(\rho_2,\bx_2,\s_2),(\eta_1,\by_1,\tau_1),(\eta_2,\by_2,\tau_2)),
\end{split}
\end{equation*}
where $U_{(\la_1,\la_{-1})}(\cdot,\cdot,\cdot,\cdot)$ is defined in
 \eqref{eq_original_bi_anti_symmetric}. This implies that 
\begin{equation}\label{eq_initial_free_bound}
\|F_2^{\ge
 N_h+1}\|_{1,\infty}\le \frac{3}{2}U_{max},
\end{equation}
and thus by \eqref{eq_smallness_assumption},
\begin{equation*}
M^{-N_{\beta}}\sum_{m=1}^{N_{L,h}}\alpha^mc_0^mM^{(N_h+1-N_{\beta})(m-2)}\sum_{b\in\{free,tree\}}\|J_{b,m}^{\ge
 N_h+1}\|_{1,\infty}\le M^{-N_{\beta}}\alpha^2
 c_0^2\frac{3}{2}U_{max}<1.
\end{equation*}
Hence, \eqref{eq_inductive_bound_each_scale} holds for $l=N_h+1$.

Take any $l\in \{N_{\beta}+1,\cdots,N_h\}$ and assume that
\eqref{eq_inductive_bound_each_scale} holds true for all
$\hat{l}\in\{l+1,\cdots,N_h+1\}$. Remark the following equalities.
\begin{align}
&F_1^{\ge l}(\psi)=F_1^{\ge l+1}(\psi)+\cP_1\int J_2^{\ge
 N_h+1}(\psi+\psi^0)d\mu_{\cC_l(w\be_p)}(\psi^0)\notag\\
&\quad+\cP_1\int (F_2^{\ge
 l+1}(\psi+\psi^0)-J_2^{\ge N_h+1}(\psi+\psi^0))d\mu_{\cC_l(w\be_p)}(\psi^0)+T^{\ge l+1}_1(\psi)\notag\\
&\quad+\cP_1\int T_2^{\ge
 l+1}(\psi+\psi^0)d\mu_{\cC_l(w\be_p)}(\psi^0)+\cP_1\int\sum_{j=3}^{N_{L,h}}J_j^{\ge
 l+1}(\psi+\psi^0)d\mu_{\cC_l(w\be_p)}(\psi^0),\label{eq_flow_equation_1}\\
&F_2^{\ge l}(\psi)=F_2^{\ge l+1}(\psi)+T_2^{\ge
 l+1}(\psi)+\cP_2\int\sum_{j=3}^{N_{L,h}}J_j^{\ge
 l+1}(\psi+\psi^0)d\mu_{\cC_l(w\be_p)}(\psi^0).\label{eq_flow_equation_2}
\end{align}
Since 
\begin{align*}
&\cP_1\int J_2^{\ge
 N_h+1}(\psi+\psi^0)d\mu_{\cC_l(w\be_p)}(\psi^0)\\
&=-\frac{1}{h}\sum_{(\rho,\bx,\s,x)\in
 I_{L,h}}(1_{\rho =1}U_c+1_{\rho\in \{2,3\}}U_o)\cC_l(\rho \b0 \s 0,\rho
 \b0 \s 0)(w\be_p)\opsi_{\rho\bx\s x}\psi_{\rho \bx \s x}\\
&\quad-\frac{\la_1}{h}\sum_{x\in
 [0,\beta)_h}\big(\cC_l(\hcX_10,\hcY_10)(w\be_p)\opsi_{\hcX_2x}\psi_{\hcY_2x}
-\cC_l(\hcX_10,\hcY_20)(w\be_p)\opsi_{\hcX_2x}\psi_{\hcY_1x}\\
&\quad-\cC_l(\hcX_20,\hcY_10)(w\be_p)\opsi_{\hcX_1x}\psi_{\hcY_2x}
+\cC_l(\hcX_20,\hcY_20)(w\be_p)\opsi_{\hcX_1x}\psi_{\hcY_1x}\big)\\
&\quad-\frac{\la_{-1}}{h}\sum_{x\in
 [0,\beta)_h}\big(\cC_l(\hcY_10,\hcX_10)(w\be_p)\opsi_{\hcY_2x}\psi_{\hcX_2x}-\cC_l(\hcY_10,\hcX_20)(w\be_p)\opsi_{\hcY_2x}\psi_{\hcX_1x}\\
&\quad -\cC_l(\hcY_20,\hcX_10)(w\be_p)\opsi_{\hcY_1x}\psi_{\hcX_2x}
+\cC_l(\hcY_20,\hcX_20)(w\be_p)\opsi_{\hcY_1x}\psi_{\hcX_1x}\big),
\end{align*}
the equation \eqref{eq_flow_equation_1}, coupled with
 \eqref{eq_simple_tadpole_bound}, leads to 
\begin{align}
\|F_1^{\ge l}\|_{1,\infty}\le& \|F_1^{\ge
 l+1}\|_{1,\infty}+9U_{max}c_0(M^{l-N_h}+M^{N_{\beta}-l})+2^2c_0\|F_2^{\ge l+1}-J_2^{\ge N_h+1}\|_{1,\infty}\notag\\
&+\|T_1^{\ge
 l+1}\|_{1,\infty}+2^2c_0\|T_2^{\ge
 l+1}\|_{1,\infty}+\sum_{j=3}^{N_{L,h}}j^2c_0^{j-1}\|J_j^{\ge
 l+1}\|_{1,\infty}.\label{eq_flow_inequality_1}
\end{align}
By the induction hypothesis we can apply \eqref{eq_item_tree_term_norm},
 \eqref{eq_item_full_term_sum} to derive the following from \eqref{eq_flow_equation_2}.
\begin{align}
&\|F_2^{\ge l}-J_2^{\ge N_h+1}\|_{1,\infty}\notag\\
&\le \|F_2^{\ge l+1}-J_2^{\ge
 N_h+1}\|_{1,\infty}+\|T_2^{\ge
 l+1}\|_{1,\infty}+\sum_{j=3}^{N_{L,h}}\left(\begin{array}{c} j \\
					     2\end{array}\right)^2c_0^{j-2}\|J_j^{\ge l+1}\|_{1,\infty}\notag\\
&\le \|F_2^{\ge l+1}-J_2^{\ge
 N_h+1}\|_{1,\infty}+
 2^{-6}c_0^{-2}M^{N_{\beta}}(2^6\alpha^{-1})^2M^{-(l+1-N_{\beta})}\notag\\
&\quad+c_0^{-2}M^{N_{\beta}}(2^2\alpha^{-1})^3M^{-(l+1-N_{\beta})}\notag\\
&\le
 (2^{-6}c_0^{-2}(2^6\alpha^{-1})^2+c_0^{-2}(2^2\alpha^{-1})^3)M^{N_{\beta}}\sum_{j=l}^{N_h}M^{-(j+1-N_{\beta})}\notag\\
&\le 2^{-4}c_0^{-2}(2^6\alpha^{-1})^2M^{N_{\beta}}M^{-(l+1-N_{\beta})},\label{eq_flow_bound_2}
\end{align}
where we have also used that $M(M-1)^{-1}\le 2$. Similarly we have
\begin{equation}\label{eq_flow_bound_2_previous}
\|F_2^{\ge l+1}-J_2^{\ge N_h+1}\|_{1,\infty}\le
 2^{-4}c_0^{-2}(2^6\alpha^{-1})^2M^{N_{\beta}}M^{-(l+2-N_{\beta})}.
\end{equation}
Then, by inserting \eqref{eq_item_tree_term_norm},
 \eqref{eq_item_full_term_sum} and \eqref{eq_flow_bound_2_previous} into
 \eqref{eq_flow_inequality_1},
\begin{align}
&\|F_1^{\ge l}\|_{1,\infty}\notag\\
&\le \|F_1^{\ge
 l+1}\|_{1,\infty}+9U_{max}c_0(M^{l-N_h}+M^{N_{\beta}-l})
+2^{-2}c_0^{-1}(2^6\alpha^{-1})^2M^{N_{\beta}}M^{-(l+2-N_{\beta})}\notag\\
&\quad+2^{-3}c_0^{-1}(2^6\alpha^{-1})^2M^{N_{\beta}}M^{-(l+1-N_{\beta})}
+2^{-4}c_0^{-1}(2^6\alpha^{-1})^2M^{N_{\beta}}M^{-(l+1-N_{\beta})}
\notag\\
&\quad+c_0^{-1}(2^2\alpha^{-1})^3M^{N_{\beta}}M^{-(l+1-N_{\beta})}\notag\\
&\le \|F_1^{\ge
 l+1}\|_{1,\infty}+9U_{max}c_0(M^{l-N_h}+M^{N_{\beta}-l})+2^{-1}c_0^{-1}(2^6\alpha^{-1})^2M^{N_{\beta}}M^{-(l+1-N_{\beta})}\notag\\
&\le
 9U_{max}c_0\sum_{j=l}^{N_h}(M^{j-N_h}+M^{N_{\beta}-j})+2^{-1}c_0^{-1}(2^6\alpha^{-1})^2M^{N_{\beta}}\sum_{j=l}^{N_h}M^{-(j+1-N_{\beta})}\notag\\
&\le 36
 U_{max}c_0+c_0^{-1}(2^6\alpha^{-1})^2M^{N_{\beta}}M^{-(l+1-N_{\beta})}.\label{eq_bound_free_degree_1}
\end{align}
It follows from \eqref{eq_item_tree_term_norm} and
 \eqref{eq_bound_free_degree_1} that 
\begin{align}
&M^{-N_{\beta}}\alpha c_0\sum_{b\in\{free,tree\}}\|J_{b,1}^{\ge
 l}\|_{1,\infty}\notag\\
&\le 36U_{max}\alpha
 c_0^2M^{-N_{\beta}}+2^{12}\alpha^{-1}M^{-(l+1-N_{\beta})}+2^9\alpha^{-1}M^{-(l-N_{\beta})}\notag\\
&\le  36U_{max}\alpha
 c_0^2M^{-N_{\beta}}+2^{12}\alpha^{-1}M^{-2}+2^9\alpha^{-1}M^{-1}.\label{eq_bound_degree_1}
\end{align}
Moreover, by using \eqref{eq_item_tree_term_sum}, 
 \eqref{eq_item_free_term_norm_sum}, \eqref{eq_initial_free_bound},
\eqref{eq_flow_bound_2} and 
 \eqref{eq_bound_free_degree_1}, 
\begin{align}
&M^{-N_{\beta}}\sum_{m=1}^{N_{L,h}}\alpha^mc_0^mM^{(l-N_{\beta})(m-2)}\sum_{b\in\{free,tree\}}\|J_{b,m}^{\ge
 l}\|_{1,\infty}\notag\\
&\le M^{-N_{\beta}}\alpha c_0 M^{-(l-N_{\beta})}\|F_{1}^{\ge
 l}\|_{1,\infty}+ M^{-N_{\beta}}\alpha^2 c_0^2 \|F_{2}^{\ge
 l}\|_{1,\infty}\notag\\
&\quad+M^{-N_{\beta}}\sum_{m=3}^{N_{L,h}}\alpha^mc_0^mM^{(l-N_{\beta})(m-2)}\|F_{m}^{\ge
 l}\|_{1,\infty}
+M^{-N_{\beta}}\sum_{m=1}^{N_{L,h}}\alpha^mc_0^mM^{(l-N_{\beta})(m-2)}\|T_{m}^{\ge
 l}\|_{1,\infty}\notag\\
&\le 36 U_{max}\alpha
 c_0^2M^{-l}+2^{12}\alpha^{-1}M^{2N_{\beta}-2l-1}+
2^8M^{N_{\beta}-l-1}+\frac{3}{2}U_{max}\alpha^2c_0^2M^{-N_{\beta}}
+2^7M^{-1}\notag\\
&\quad +2^8\alpha^{-1}M^2\notag\\
&\le 2U_{max}\alpha^2c_0^2M^{-N_{\beta}}+2^{12}\alpha^{-1}M^{-3}+2^8M^{-2}+2^7M^{-1}+2^8\alpha^{-1}M^2.\label{eq_bound_degree_all}
\end{align}
One can check that the right-hand sides of \eqref{eq_bound_degree_1}
 and \eqref{eq_bound_degree_all} are less than 1 under the assumption
 \eqref{eq_smallness_assumption} and conclude the proof.
\end{proof}

\subsection{An upper bound on the final integration}
\label{subsec_final_integration}
 Later in this subsection we will see that \eqref{eq_schwinger_functional}
 is equal to the multi-contour integral of $\partial J_0^{\ge
 N_{\beta}}/\partial \la_a|_{(\la_1,\la_{-1})=(0,0)}$ $(a=1,-1)$ if the coupling
 constants $U_c$, $U_o$ obey the sufficient condition for Proposition
 \ref{prop_inductive_bound} to hold.  Keeping this fact in mind, let us
 try to find an $h$-,$L$-independent upper bound on $|\partial J_0^{\ge
 N_{\beta}}/\partial \la_a|_{(\la_1,\la_{-1})=(0,0)}|$ by using the
 results obtained in the previous subsection. This will enable us to
 bound \eqref{eq_schwinger_functional}, too. We need the following
 lemma.
\begin{lemma}\label{lem_bound_derivative}
Assume \eqref{eq_smallness_assumption}. For any $b\in\{free,tree\}$,
 $m'\in \{0,\cdots,N_{L,h}\}$, $\bX^{m'},\bY^{m'}$ $\in I_{L,h}^{m'}$ and
 $l\in \{N_{\beta}+1,\cdots,N_h\}$ the following inequalities hold.
\begin{align}
&\sum_{m=0}^{N_{L,h}}\left\|\frac{\partial T_m^{\ge l}}{\partial
 J_{b,m'}^{\ge l+1}(\bX^{m'},\bY^{m'})}\right\|_{1}(Ac_0)^{m}
\le 1_{m'\ge 1}h^{-2m'}(m'!)^2(2^2Ac_0)^{m'}2^8
\alpha^{-1}M^{-(l-N_{\beta})}A\label{eq_item_derivative_T_J}\\
&\quad(\forall A\in [2^4,M^{l+1-N_{\beta}}]).\notag\\
&\sum_{d\in\{free,tree\}}\left|\frac{\partial J_{d,0}^{\ge
			 N_{\beta}}}{\partial J_{b,m'}^{\ge
			 N_{\beta}+1}(\bX^{m'},\bY^{m'})}\right|\le
h^{-2m'}(m'!)^2(2^5c_0)^{m'}.\label{eq_item_derivative_J_J_0}\\
&\sum_{m=0}^{N_{L,h}}\sum_{d\in\{free,tree\}}\left\|\frac{\partial
					     J_{d,m}^{\ge l}}{\partial 
 J_{b,m'}^{\ge l+1}(\bX^{m'},\bY^{m'})}\right\|_{1}(Ac_0)^{m}
\le h^{-2m'}(m'!)^2(2^2Ac_0)^{m'}\label{eq_item_derivative_J_J}\\
&\quad(\forall A\in [2^4,M^{l+1-N_{\beta}}]).\notag\\
&\sum_{m=0}^{N_{L,h}}\left\|\frac{\partial F_{m}^{\ge l}}{\partial 
 F_{\hm}^{\ge l+1}(\hbX^{\hm},\hbY^{\hm})}\right\|_{1}m!(Ac_0)^{m}
\le
 h^{-2\hm}(\hm!)^3((M^{l-N_h}+M^{N_{\beta}-l})c_0+Ac_0)^{\hm}\label{eq_item_derivative_F_F}\\
&\quad(\forall \hbX^{\hm},\hbY^{\hm}\in
     I_{L,h}^{\hm}\text{ with }\hbX^{\hm}\subset
 ((\hcX_1,s),(\hcX_2,s)),\hbY^{\hm}\subset
 ((\hcY_1,s),(\hcY_2,s))\notag\\
&\qquad\text{ for some }s\in[0,\beta)_h, \ \forall A\ge 0).\notag\\
&\frac{1}{\beta}\sum_{m=0}^{N_{L,h}}\left\|\frac{\partial F_{m}^{\ge
				    N_{h}+1}}{\partial
				    \la_{a}}\right\|_{1}m!(Ac_0)^m\le
2(Ac_0)^2\quad (\forall a\in
 \{1,-1\}).\label{eq_item_derivative_F_lambda}
\end{align}
\end{lemma}
\begin{proof}
Proof of \eqref{eq_item_derivative_T_J}: By Lemma
 \ref{lem_bound_tree_term_partial_derivative} and the inequality that
 $2^{-3}A(2^{-3}A-1)^{-1}\le 2$,
\begin{align}
&\sum_{m=0}^{N_{L,h}}\left\|\frac{\partial T_m^{\ge l}}{\partial
 J_{b,m'}^{\ge l+1}(\bX^{m'},\bY^{m'})}\right\|_{1}(Ac_0)^{m}\notag\\
&\le 1_{m'\ge
 1}h^{-2m'}(m'!)^2(2^5c_0)^{m'}\sum_{n=2}^{\infty}M^{-l(n-1)}\prod_{j=1}^{n-1}\left(\sum_{m_j=1}^{N_{L,h}}2^{5m_j}c_0^{m_j}\|J_{m_j}^{\ge
 l+1}\|_{1,\infty}\right)\notag\\
&\quad\cdot\sum_{m=0}^{\sum_{j=1}^{n-1}m_j+m'-n+1}(2^{-3}A)^m\notag\\
&\le 1_{m'\ge
 1}h^{-2m'}(m'!)^22(2^2Ac_0)^{m'}\sum_{n=2}^{\infty}
\left(2^{3}A^{-1}M^{-l}\sum_{m=1}^{N_{L,h}}2^{2m}A^m c_0^m\|J_m^{\ge
 l+1}\|_{1,\infty}\right)^{n-1}.\label{eq_tree_sum_general}
\end{align}
By Proposition \ref{prop_inductive_bound} and the assumption that $A\le M^{l+1-N_{\beta}}$ we have
\begin{align}
&2^3A^{-1}M^{-l}\sum_{m=1}^{N_{L,h}}2^{2m}A^mc_0^m\|J_m^{\ge
 l+1}\|_{1,\infty}\notag\\
&=2^3A^{-1}M^{-(l-N_{\beta})}\left(2^2M^{-N_{\beta}}Ac_0\|J_1^{\ge
 l+1}\|_{1,\infty}+A^2M^{-N_{\beta}}\sum_{m=2}^{N_{L,h}}2^{2m}c_0^mA^{m-2}\|J_m^{\ge
 l+1}\|_{1,\infty}\right)\notag\\
&\le 2^3A^{-1}M^{-(l-N_{\beta})}(2^2A\alpha^{-1}+2^4A^2\alpha^{-2})=M^{-(l-N_{\beta})}(1+2^2A\alpha^{-1})2^5\alpha^{-1}.\label{eq_sum_inside_input}
\end{align}
By giving \eqref{eq_sum_inside_input} back to
 \eqref{eq_tree_sum_general} and remarking that
 $M^{-(l-N_{\beta})}(1+2^2A\alpha^{-1})\le 1$ and $1+2^2A\alpha^{-1}\le
 2A$, we obtain 
\begin{align*}
&\sum_{m=0}^{N_{L,h}}\left\|\frac{\partial T_m^{\ge l}}{\partial
 J_{b,m'}^{\ge l+1}(\bX^{m'},\bY^{m'})}\right\|_{1}(Ac_0)^{m}\notag\\
&\le 1_{m'\ge 1}h^{-2m'}(m'!)^2(2^2Ac_0)^{m'}M^{-(l-N_{\beta})}2^2A\sum_{n=2}^{\infty}(2^{5}\alpha^{-1})^{n-1},
\end{align*}
which gives the bound \eqref{eq_item_derivative_T_J}.

Proof of \eqref{eq_item_derivative_J_J_0}: By Lemma
 \ref{lem_bound_tree_term_partial_derivative} and
 \eqref{eq_useful_inside},
\begin{align}
\left|\frac{\partial T_{0}^{\ge
			 N_{\beta}}}{\partial J_{b,m'}^{\ge
			 N_{\beta}+1}(\bX^{m'},\bY^{m'})}\right|&\le
1_{m'\ge 1}h^{-2m'}(m'!)^2(2^5c_0)^{m'}\sum_{n=2}^{\infty}(2^6\alpha^{-1})^{n-1}\notag\\
&\le 1_{m'\ge 1}h^{-2m'}(m'!)^2(2^5c_0)^{m'}2^7\alpha^{-1}.\label{eq_tree_bound_initial}
\end{align}

Let us characterize the derivative of $F_m^{\ge l}(\psi)$ with respect to
 $J_{b,m'}^{\ge l+1}(\bX^{m'},\bY^{m'})$, as it will be useful in the rest of the proof of \eqref{eq_item_derivative_J_J_0} as well as
 in the proofs of \eqref{eq_item_derivative_J_J},
 \eqref{eq_item_derivative_F_F}. If $m\le m'$,
\begin{align*}
&\frac{\partial F_{m}^{\ge
			 l}(\psi)}{\partial J_{b,m'}^{\ge
			l+1}(\bX^{m'},\bY^{m'})}=h^{-2m}\sum_{\bX^m,\bY^m\in I_{L,h}^m}\Big(h^{2m-2m'}(m'!)^2
1_{\bX^m\subset \bX^{m'}}1_{\bY^m\subset \bY^{m'}}\\
&\quad\cdot\eps_{\bX^m,\bX^{m'},\bY^m,\bY^{m'}}\int
 (\opsi^0)_{\bX^{m'}\backslash \bX^m}(\psi^0)_{\bY^{m'}\backslash
 \bY^m}d\mu_{\cC_l(w\be_p)}(\psi^0)\Big)(\opsi)_{\bX^m}(\psi)_{\bY^m},
\end{align*}
where the factor $\eps_{\bX^m,\bX^{m'},\bY^m,\bY^{m'}}\in\{1,-1\}$
 depends only on $\bX^m,\bX^{m'},\bY^m,\bY^{m'}$. This equality and
 \eqref{eq_anti_symmetric_estimate} imply that
\begin{align}
&\left\|\frac{\partial F_{m}^{\ge
			 l}}{\partial J_{b,m'}^{\ge
			l+1}(\bX^{m'},\bY^{m'})}\right\|_{1}
\notag\\
&\le h^{-2m'}(m'!)^2\sum_{\bX^m\in I_{L,h}^m\atop \bX^m\subset
 \bX^{m'}}\sum_{\bY^m\in I_{L,h}^m\atop \bY^m\subset
 \bY^{m'}}\left|\int
 (\opsi^0)_{\bX^{m'}\backslash \bX^m}(\psi^0)_{\bY^{m'}\backslash
 \bY^m}d\mu_{\cC_l(w\be_p)}(\psi^0)\right|.\label{eq_free_derivative_preparation}
\end{align}

By using \eqref{eq_simple_determinant_bound} one can derive from \eqref{eq_free_derivative_preparation} that
$$
\left|\frac{\partial F_{0}^{\ge
			 N_{\beta}}}{\partial J_{b,m'}^{\ge
			 N_{\beta}+1}(\bX^{m'},\bY^{m'})}\right|\le h^{-2m'}(m'!)^2c_0^{m'},
$$
which, coupled with \eqref{eq_tree_bound_initial}, yields the bound
 \eqref{eq_item_derivative_J_J_0}.

Proof of \eqref{eq_item_derivative_J_J}: By using
 \eqref{eq_free_derivative_preparation} and the inequality that
$$\sum_{m=0}^{m'}\left(\begin{array}{c}m'\\m\end{array}\right)^2\le
 (1_{m'=0}+1_{m'\ge 1}2^{-1})2^{2m'},$$
\begin{align}
&\sum_{m=0}^{N_{L,h}}\left\|\frac{\partial F_{m}^{\ge
			 l}}{\partial J_{b,m'}^{\ge
			 l+1}(\bX^{m'},\bY^{m'})}\right\|_{1}(Ac_0)^m\le
				     h^{-2m'}(m'!)^2\sum_{m=0}^{m'}\left(\begin{array}{c}m'\\m\end{array}\right)^2c_0^{m'-m}(Ac_0)^m\notag\\
&\le
						h^{-2m'}(m'!)^2(Ac_0)^{m'}\sum_{m=0}^{m'}\left(\begin{array}{c}m'\\m\end{array}\right)^2\le (1_{m'=0}+1_{m'\ge 1}2^{-1})h^{-2m'}(m'!)^2(2^2Ac_0)^{m'}.\label{eq_free_sum_general}
\end{align}
By combining \eqref{eq_free_sum_general} with \eqref{eq_item_derivative_T_J} one can obtain \eqref{eq_item_derivative_J_J}.

Proof of \eqref{eq_item_derivative_F_F}: By applying \eqref{eq_simple_tadpole_bound} to \eqref{eq_free_derivative_preparation} we have 
\begin{align*}
&\sum_{m=0}^{N_{L,h}}\left\|\frac{\partial F_{m}^{\ge
			 l}}{\partial F_{\hm}^{\ge
			 l+1}(\hbX^{\hm},\hbY^{\hm})}\right\|_{1}m!(Ac_0)^m\\
&\le
				     h^{-2\hm}(\hm
			   !)^2\sum_{m=0}^{\hm}\left(\begin{array}{c}\hm\\m\end{array}\right)^2(\hm-m)!m!((M^{l-N_h}+M^{N_{\beta}-l})c_0)^{\hm-m}(Ac_0)^m\\
&= h^{-2 \hm}(\hm!)^3 ((M^{l-N_h}+M^{N_{\beta}-l})c_0+Ac_0)^{\hm}.
\end{align*}

Proof of \eqref{eq_item_derivative_F_lambda}: The claimed inequality
 follows from \eqref{eq_anti_symmetric_estimate} and the equality 
\begin{equation*}
\frac{\partial F_m^{\ge N_h+1}(\psi)}{\partial
 \la_a}=-1_{m=2}\frac{1}{h}\sum_{x\in[0,\beta)_h}(1_{a=1}\opsi_{\hcX_1x}\opsi_{\hcX_2x}\psi_{\hcY_2x}\psi_{\hcY_1x}+1_{a=-1}\opsi_{\hcY_1x}\opsi_{\hcY_2x}\psi_{\hcX_2x}\psi_{\hcX_1x}).
\end{equation*}
\end{proof}

\begin{corollary}\label{cor_analyticity_variables}
Assume \eqref{eq_smallness_assumption}. Take any
 $l\in\{N_{\beta},\cdots,N_h\}$, $d\in \{free,tree\}$, $m\in
 \{0,1,\cdots,N_{L,h}\}$ and $\bX^m,\bY^m\in I_{L,h}^m$. Moreover, assume that
$m=0$ if $l=N_{\beta}$. The following statements hold true.
\begin{enumerate}[(i)]
\item\label{item_independent_variables} $J_{d,m}^{\ge l}(\bX^m,\bY^m)$
     is analytic with respect to the independent variables\\
     $J_{b,m'}^{\ge l+1}(\bX^{m'},\bY^{m'})$ $(b\in\{free,tree\},m'\in
     \{0,\cdots,N_{L,h}\},\bX^{m'},\bY^{m'}\in (I_{L,h})_o^{m'})$ in
     the domain characterized by \eqref{eq_inductive_bound_each_scale}
     for $l+1$.
\item\label{item_complex_variables} The function $(\la_1,\la_{-1},U_c,U_o,w)\mapsto J_{d,m}^{\ge
     l}(\bX^m,\bY^m)$
is analytic in the domain
\begin{equation}\label{eq_enlarged_domain}
\{(\la_1,\la_{-1},U_c,U_o,w)\in\C^5\ |\
 |\la_1|,|\la_{-1}|,|U_c|,|U_o|<2^{-4}\alpha^{-2}c_0^{-2}M^{N_{\beta}},w\in D_R^i\}.
\end{equation}
\end{enumerate}
\end{corollary}
\begin{remark}
Since the inequality \eqref{eq_inductive_bound_each_scale} for $l+1$ is
 independent of $J_{b,0}^{\ge l+1}$ $(b\in\{free,tree\})$, the claim
 \eqref{item_independent_variables} implies that $J_{d,m}^{\ge
 l}(\bX^m,\bY^m)$ is entirely analytic with respect to the variables
 $J_{b,0}^{\ge l+1}$ $(b\in\{free,tree\})$.
\end{remark}

\begin{proof}[Proof of Corollary \ref{cor_analyticity_variables}]
The inequalities \eqref{eq_item_derivative_J_J_0},
 \eqref{eq_item_derivative_J_J} imply the claim
 \eqref{item_independent_variables}. It is trivial that
 $(\la_1,\la_{-1},U_c,U_o,w)\mapsto J_{b,m'}^{\ge
 N_h+1}(\bX^{m'},\bY^{m'})$ is analytic in \eqref{eq_enlarged_domain}
 for all $b\in \{free,tree\}$, $m'\in \{0,\cdots,N_{L,h}\}$,
 $\bX^{m'},\bY^{m'}\in I_{L,h}^{m'}$. Then the analyticity of\\
 $J_{d,m}^{\ge l}(\bX^m,\bY^m)$ with respect to
 $(\la_1,\la_{-1},U_c,U_o)$ follows from the claim
 \eqref{item_independent_variables}, Proposition
 \ref{prop_inductive_bound} and the analyticity of composition of
 analytic functions. Assume that for some $l'\in\{l,\cdots,N_h\}$,
 $J^{\ge l'+1}_{b,m'}(\bX^{m'},\bY^{m'})$ $(b\in\{free,tree\},m'\in
 \{0,\cdots,N_{L,h}\},\bX^{m'},\bY^{m'}\in I_{L,h}^{m'})$ are analytic
 with respect to $w\in D_R^i$. For any $q\in\{0,\cdots,N_{L,h}\}$,
 $\bX^q,\bY^q\in I_{L,h}^q$ the analyticity of $F_q^{\ge
 l'}(\bX^{q},\bY^{q})$, $T_{n,q}^{\ge
 l'}(\bX^{q},\bY^{q})$ $(n\in\N_{\ge 2})$ is clear since these consist
 of finite sums and products of $J_{m'}^{\ge l'+1}(\bX^{m'},\bY^{m'})$
 and $\cC_{l'}(w\be_p)$, which are analytic in $D_R^i$. Moreover, the
 proof of the inequality \eqref{eq_item_tree_term_norm} shows that $\sum_{n=2}^jT_{n,q}^{\ge
 l'}(\bX^{q},\bY^{q})$ converges to $T_{q}^{\ge
 l'}(\bX^{q},\bY^{q})$ uniformly with respect to $w\in D_R$ as $j\to
 \infty$. This implies that $T_q^{\ge l'}(\bX^q,\bY^q)$ is analytic in
 $D_R^i$. Thus, the induction concludes that
 $w\mapsto J_{d,m}^{\ge l}(\bX^m,\bY^m)$ is analytic in $D_R^i$.
\end{proof}
 
\begin{proposition}\label{prop_bound_schwinger}
Assume \eqref{eq_smallness_assumption}.
The following inequality holds for any $a\in\{1,-1\}$ and
 $(\la_1,\la_{-1},U_c,U_o,w)$ contained in the
 domain \eqref{eq_enlarged_domain}.
$$
\frac{1}{\beta}\left|\frac{\partial J_0^{\ge N_{\beta}}}{\partial
 \la_a}\right|\le 2^{12}c_0^2.
$$
\end{proposition}
\begin{proof}
Let us assume that $a=1$. The proof for $a=-1$ is essentially the
 same. By Corollary \ref{cor_analyticity_variables} we can apply the
 chain rule to derive the following.
\begin{align}
\frac{1}{\beta}\left|\frac{\partial J_0^{\ge N_{\beta}}}{\partial
 \la_1}\right|&\le
 \frac{1}{\beta}\prod_{l=N_{\beta}}^{N_{h}+1}\left(\sum_{m_l=0}^{N_{L,h}}\sum_{b_l\in\{free,tree\}}\sum_{\bX^{m_l},\bY^{m_l}\in
 (I_{L,h})_o^{m_l}}\right)1_{m_{N_{\beta}}=0}\notag\\
&\quad\cdot\left|\frac{\partial J^{\ge
 N_h+1}_{b_{N_h+1},m_{N_h+1}}(\bX^{m_{N_h+1}},\bY^{m_{N_h+1}})}{\partial
 \la_1}\right|\prod_{j=N_{\beta}}^{N_h}\left|\frac{\partial
 J_{b_j,m_j}^{\ge j}(\bX^{m_j},\bY^{m_j})}{\partial J^{\ge
 j+1}_{b_{j+1},m_{j+1}}(\bX^{m_{j+1}},\bY^{m_{j+1}})}\right|\notag\\
&=\sum_{\hat{l}=N_{\beta}}^{N_h+1}C_R(\hat{l}),\label{eq_chain_rule}
\end{align}
where
\begin{align*}
&C_R(\hat{l}):=\\
&\frac{1}{\beta}\prod_{l=N_{\beta}}^{N_h+1}\left(\sum_{m_l=0}^{N_{L,h}}
\sum_{b_l\in\{free,tree\}}\sum_{\bX^{m_l},\bY^{m_l}\in
 (I_{L,h})_o^{m_l}}\right)1_{m_{N_{\beta}}=0}1_{b_l=free(\forall
 l\in\{\hat{l},\hat{l}+1,\cdots,N_h+1\}),b_{\hat{l}-1}=tree}\\
&\cdot \left|\frac{\partial J^{\ge
 N_h+1}_{b_{N_h+1},m_{N_h+1}}(\bX^{m_{N_h+1}},\bY^{m_{N_h+1}})}{\partial
 \la_1}\right|\prod_{j=N_{\beta}}^{N_h}\left|\frac{\partial
 J_{b_j,m_j}^{\ge j}(\bX^{m_j},\bY^{m_j})}{\partial J^{\ge
 j+1}_{b_{j+1},m_{j+1}}(\bX^{m_{j+1}},\bY^{m_{j+1}})}\right|.
\end{align*}

Let us decompose $\sum_{\hat{l}=N_{\beta}}^{N_{h}+1}C_R(\hat{l})$ into
 $\sum_{\hat{l}=N_{\beta}}^{N_{\beta}+1}C_R(\hat{l})$ and
 $\sum_{\hat{l}=N_{\beta}+2}^{N_{h}+1}C_R(\hat{l})$ and estimate each part
 separately. In the following calculation we use the equality
 \eqref{eq_norm_ordered_base} repeatedly. By using
 \eqref{eq_item_derivative_J_J_0}, \eqref{eq_item_derivative_F_F},
 \eqref{eq_item_derivative_F_lambda} in this order,
\begin{align}
&\sum_{\hat{l}=N_{\beta}}^{N_{\beta}+1}C_R(\hat{l})=\frac{1}{\beta}\prod_{l=N_{\beta}+1}^{N_h+1}\left(\sum_{m_l=0}^{N_{L,h}}\sum_{\bX^{m_l},\bY^{m_l}\in
 (I_{L,h})^{m_l}_o}\right)\left|\frac{\partial F^{\ge
 N_h+1}_{m_{N_h+1}}(\bX^{m_{N_h+1}},\bY^{m_{N_h+1}})}{\partial
 \la_1}\right|\notag\\
&\quad\cdot 1_{\exists x\in[0,\beta)_h,\exists \nu,\xi\in \S_2\text{
 s.t.
 }\bX^{m_{N_h+1}}=((\hcX_{\nu(1)},x),(\hcX_{\nu(2)},x)),\bY^{m_{N_h+1}}=((\hcY_{\xi(1)},x),(\hcY_{\xi(2)},x))}\notag\\
&\quad\cdot\prod_{j=N_{\beta}+1}^{N_h}\left(\left|\frac{\partial
 F_{m_j}^{\ge j}(\bX^{m_j},\bY^{m_j})}{\partial F^{\ge
 j+1}_{m_{j+1}}(\bX^{m_{j+1}},\bY^{m_{j+1}})}\right|1_{\bX^{m_j}\subset
 \bX^{m_{j+1}},\bY^{m_j}\subset
 \bY^{m_{j+1}}}\right)\notag\\
&\quad\cdot\sum_{d\in \{free,tree\}}\left|\frac{\partial
 J^{\ge N_{\beta}}_{d,0}}{\partial F^{\ge
 N_{\beta}+1}_{m_{N_{\beta}+1}}(\bX^{m_{N_{\beta}+1}},\bY^{m_{N_{\beta}+1}})}\right|\notag\\
&\le \frac{1}{\beta}\prod_{l=N_{\beta}+1}^{N_h+1}\left(\sum_{m_l=0}^{N_{L,h}}\sum_{\bX^{m_l},\bY^{m_l}\in
 (I_{L,h})^{m_l}_o}\right)\left|\frac{\partial F^{\ge
 N_h+1}_{m_{N_h+1}}(\bX^{m_{N_h+1}},\bY^{m_{N_h+1}})}{\partial
 \la_1}\right|\notag\\
&\quad\cdot 1_{\exists x\in[0,\beta)_h,\exists \nu,\xi\in \S_2\text{
 s.t.
 }\bX^{m_{N_h+1}}=((\hcX_{\nu(1)},x),(\hcX_{\nu(2)},x)),\bY^{m_{N_h+1}}=((\hcY_{\xi(1)},x),(\hcY_{\xi(2)},x))}\notag\\
&\quad\cdot\prod_{j=N_{\beta}+1}^{N_h}\left(\left|\frac{\partial
 F_{m_j}^{\ge j}(\bX^{m_j},\bY^{m_j})}{\partial F^{\ge
 j+1}_{m_{j+1}}(\bX^{m_{j+1}},\bY^{m_{j+1}})}\right|1_{\bX^{m_j}\subset
 \bX^{m_{j+1}},\bY^{m_j}\subset
 \bY^{m_{j+1}}}\right)\notag\\
&\quad\cdot h^{-2m_{N_{\beta}+1}}(m_{N_{\beta}+1}!)^2
 (2^5c_0)^{m_{N_{\beta}+1}}\notag\\
&\le
 \frac{1}{\beta}\sum_{m_{N_h+1}=0}^{N_{L,h}}\sum_{\bX^{m_{N_h+1}},\bY^{m_{N_h+1}}\in
 (I_{L,h})_0^{m_{N_h+1}}}\left|\frac{\partial F^{\ge
 N_h+1}_{m_{N_h+1}}(\bX^{m_{N_h+1}},\bY^{m_{N_h+1}})}{\partial
 \la_1}\right|\notag\\
&\quad\cdot
 h^{-2m_{N_{h}+1}}(m_{N_{h}+1}!)^3\left(\sum_{l=N_{\beta}+1}^{N_h}(M^{l-N_h}+M^{N_{\beta}-l})c_0+2^5c_0\right)^{m_{N_h+1}}\notag\\
&\le 2
 \left(\sum_{l=N_{\beta}+1}^{N_h}(M^{l-N_h}+M^{N_{\beta}-l})c_0+2^5c_0\right)^2.\label{eq_free_chain_rule}
\end{align}

By applying \eqref{eq_item_derivative_J_J_0},
 \eqref{eq_item_derivative_J_J}, \eqref{eq_item_derivative_T_J},
 \eqref{eq_item_derivative_F_F}, \eqref{eq_item_derivative_F_lambda} in
 this order and recalling the condition
 \eqref{eq_smallness_assumption} we observe that
\begin{align}
&\sum_{\hat{l}=N_{\beta}+2}^{N_{h}+1}C_R(\hat{l})\le
 \sum_{\hat{l}=N_{\beta}+2}^{N_{h}+1}\frac{1}{\beta}
\prod_{l=\hat{l}-1}^{N_h+1}\left(\sum_{m_l=0}^{N_{L,h}}\sum_{\bX^{m_l},\bY^{m_l}\in
 (I_{L,h})^{m_l}_o}\right)\left|\frac{\partial F^{\ge
 N_h+1}_{m_{N_h+1}}(\bX^{m_{N_h+1}},\bY^{m_{N_h+1}})}{\partial
 \la_1}\right|\notag\\
&\quad\cdot\prod_{j=\hat{l}}^{N_h}\left|\frac{\partial
 F_{m_j}^{\ge j}(\bX^{m_j},\bY^{m_j})}{\partial F^{\ge
 j+1}_{m_{j+1}}(\bX^{m_{j+1}},\bY^{m_{j+1}})}\right|\cdot
\left|\frac{\partial
 T^{\ge \hat{l}-1}_{m_{\hat{l}-1}}(\bX^{m_{\hat{l}-1}},\bY^{m_{\hat{l}-1}})}{\partial F^{\ge
 \hat{l}}_{m_{\hat{l}}}(\bX^{m_{\hat{l}}},\bY^{m_{\hat{l}}})}\right|\notag\\
&\quad\cdot
 h^{-2m_{\hat{l}-1}}(m_{\hat{l}-1}!)^2(2^{2(\hat{l}-2-N_{\beta})}2^5c_0)^{m_{\hat{l}-1}}\notag\\
&\le \sum_{\hat{l}=N_{\beta}+2}^{N_{h}+1}\frac{1}{\beta}\prod_{l=\hat{l}}^{N_h+1}
\left(\sum_{m_l=0}^{N_{L,h}}\sum_{\bX^{m_l},\bY^{m_l}\in
 (I_{L,h})^{m_l}_o}\right)\left|\frac{\partial F^{\ge
 N_h+1}_{m_{N_h+1}}(\bX^{m_{N_h+1}},\bY^{m_{N_h+1}})}{\partial
 \la_1}\right|\notag\\
&\quad\cdot\prod_{j=\hat{l}}^{N_h}\left|\frac{\partial
 F_{m_j}^{\ge j}(\bX^{m_j},\bY^{m_j})}{\partial F^{\ge
 j+1}_{m_{j+1}}(\bX^{m_{j+1}},\bY^{m_{j+1}})}\right|\notag\\
&\quad\cdot
 h^{-2m_{\hat{l}}}(m_{\hat{l}}!)^2(2^{2(\hat{l}-1-N_{\beta})}2^5c_0)^{m_{\hat{l}}}2^{13}\alpha^{-1}M^{-(\hat{l}-1-N_{\beta})}2^{2(\hat{l}-2-N_{\beta})}\notag\\
&\le
2
 \sum_{\hat{l}=N_{\beta}+2}^{N_h+1}\left(\sum_{l=\hat{l}}^{N_h}(M^{l-N_h}+M^{N_{\beta}-l})c_0+2^{2(\hat{l}-1-N_{\beta})}2^5c_0\right)^2(2^2M^{-1})^{\hat{l}-1-N_{\beta}}.\label{eq_free_tree_chain_rule}
\end{align}

Finally by putting
 \eqref{eq_chain_rule},\eqref{eq_free_chain_rule},\eqref{eq_free_tree_chain_rule}
 together,
\begin{align*}
&\frac{1}{\beta}\left|\frac{\partial J_0^{\ge N_{\beta}}}{\partial \la_1}\right|\le
 2\sum_{\hat{l}=N_{\beta}+1}^{N_h+1}\left(\sum_{l=\hat{l}}^{N_h}(M^{l-N_h}+M^{N_{\beta}-l})c_0+2^{2(\hat{l}-1-N_{\beta})}2^5c_0\right)^2
(2^2M^{-1})^{\hat{l}-1-N_{\beta}}\\
&\le
 2(4c_0+2^5c_0)^2\sum_{\hat{l}=N_{\beta}+1}^{N_h+1}(2^6M^{-1})^{\hat{l}-1-N_{\beta}}\le 2(4+2^5)^2c_0^2\left(1-\frac{1}{4}\right)^{-1}\le 2^{12}c_0^2.
\end{align*}
\end{proof}

Here we can give the proof of Theorem \ref{thm_exponential_decay}.
\begin{proof}[Proof of Theorem \ref{thm_exponential_decay}]
Assume that $M=\max\{78E_{max}^2,2^8\}$ and $\alpha=2^{10}M^2$. Then, if
 $|\la_1|,|\la_{-1}|,|U_c|,|U_o|<2^{-4}\alpha^{-2}c_0^{-2}M^{N_{\beta}}$,
 the condition \eqref{eq_smallness_assumption} holds.
 
By Lemma \ref{lem_grassmann_formulation} \eqref{item_real_part}, for any
 sufficiently large $h\in 2\N/\beta$ there exists a domain $\cO_h\subset
 \C$ containing the interval $[-2^{-4}\alpha^{-2}c_0^{-2}M^{N_{\beta}},$
 $2^{-4}\alpha^{-2}c_0^{-2}M^{N_{\beta}}]$
 inside such that $(U_c,U_o)\mapsto (\partial/\partial \la)\log(\int
 e^{V_{(\la,\la)}(\psi)}d\mu_{\cC}(\psi))|_{\la =0}$ is analytic in
 $\cO_h\times\cO_h$. Let us fix such a large $h\in 2\N/\beta$. 

By the construction of $G_0^{\ge
 N_{\beta}}$ and $J_0^{\ge N_{\beta}}$ and Corollary
 \ref{cor_analyticity_variables} \eqref{item_complex_variables} there
 exists $U_{small}>0$ such
 that $J_0^{\ge N_{\beta}}=G_0^{\ge N_{\beta}}$ holds and
 $(\la_1,\la_{-1},U_c,U_o,w)\mapsto J_0^{\ge N_{\beta}}$ is analytic in
 $D_{small}^i\times D^i_{R}$. In order to indicate the dependency on the
 variable $w$, let us write $J_0^{\ge N_{\beta}}(w\be_p)$, $G_0^{\ge
 N_{\beta}}(w\be_p)$ instead of $J_0^{\ge N_{\beta}}$, $G_0^{\ge
 N_{\beta}}$. Then for
 any $n\in \N$ with $2\pi n/L + \cF_{t,\beta}(8/\pi^2) < R$ and $(\la_1,\la_{-1},U_c,U_o)\in
 D_{small}^i$,
\begin{align*}
&\sum_{a\in \{1,-1\}}\prod_{j=1}^n\left(\frac{L}{2\pi}\int_0^{2\pi
 a/L}d\theta_{a,j}\frac{1}{2\pi
 i}\oint_{|w_{a,j}-\theta_{a,j}|=\cF_{t,\beta}(8/\pi^2)/n}dw_{a,j}\frac{1}{(w_{a,j}-\theta_{a,j})^2}\right)\\
&\quad\cdot\left(\frac{\partial}{\partial \la_a}G_0^{\ge
 N_{\beta}}\left(\sum_{j=1}^nw_{a,j}\be_p\right)-
\frac{\partial}{\partial \la_a}J_0^{\ge
 N_{\beta}}\left(\sum_{j=1}^nw_{a,j}\be_p\right)\right)=0.
\end{align*}

On the other hand, Corollary \ref{cor_analyticity_variables}
 \eqref{item_complex_variables} implies that 
\begin{align*}
(U_c,U_o)\mapsto &\sum_{a\in \{1,-1\}}\prod_{j=1}^n\left(\frac{L}{2\pi}\int_0^{2\pi
 a/L}d\theta_{a,j}\frac{1}{2\pi
 i}\oint_{|w_{a,j}-\theta_{a,j}|=\cF_{t,\beta}(8/\pi^2)/n}dw_{a,j}\frac{1}{(w_{a,j}-\theta_{a,j})^2}\right)
\\
&\cdot\frac{\partial}{\partial \la_a}J_0^{\ge
 N_{\beta}}\left(\sum_{j=1}^nw_{a,j}\be_p\right)\Big|_{\la_1=\la_{-1}=0}
\end{align*}
is analytic in $\{(U_c,U_o)\in \C^2\ |\
 |U_c|,|U_o|<2^{-4}\alpha^{-2}c_0^{-2}M^{N_{\beta}}\}$. Therefore, by
 Lemma \ref{lem_contour_integral_formulation}
 \eqref{item_contour_integral_transform}, the identity theorem for
 analytic functions ensures that 
\begin{align}
&\left(\frac{L}{2\pi}\left(e^{i\frac{2\pi}{L}\<\sum_{j=1}^2(\hat{s}(\hs_j)\hbx_j-\hat{s}(\htau_j)\hby_j),\be_p\>}-1\right)\right)^n\frac{1}{\beta}\frac{\partial}{\partial
 \la }\log\left(\int
 e^{V_{(\la,\la)}(\psi)}d\mu_{\cC}(\psi)\right)\Big|_{\la=0}\notag\\
&=\sum_{a\in \{1,-1\}}\prod_{j=1}^n\left(\frac{L}{2\pi}\int_0^{2\pi
 a/L}d\theta_{a,j}\frac{1}{2\pi
 i}\oint_{|w_{a,j}-\theta_{a,j}|=\cF_{t,\beta}(8/\pi^2)/n}dw_{a,j}
\frac{1}{(w_{a,j}-\theta_{a,j})^2}\right)\notag
\\
&\quad\cdot\frac{1}{\beta}\frac{\partial}{\partial \la_a}J_0^{\ge
 N_{\beta}}\left(\sum_{j=1}^nw_{a,j}\be_p\right)\Big|_{\la_1=\la_{-1}=0}\label{eq_consequence_identity}
\end{align}
for all $U_c,U_o\in\R$ with
 $|U_c|,|U_o|<2^{-4}\alpha^{-2}c_0^{-2}M^{N_{\beta}}$. Then by using
 Proposition \ref{prop_bound_schwinger} and $n^n\le n!e^n$ we can
 estimate \eqref{eq_consequence_identity} as follows.
\begin{align}
&\left|\frac{L}{2\pi}\left(e^{i\frac{2\pi}{L}\<\sum_{j=1}^2(\hat{s}(\hs_j)\hbx_j-\hat{s}(\htau_j)\hby_j),\be_p\>}-1\right)\right|^n
\left|\frac{1}{\beta}\frac{\partial}{\partial
 \la }\log\left(\int
 e^{V_{(\la,\la)}(\psi)}d\mu_{\cC}(\psi)\right)\Big|_{\la=0}\right|\notag\\
&\le 2^{13}c_0^2n!e^n\cF_{t,\beta}(8/\pi^2)^{-n}.\label{eq_bound_still_discrete}
\end{align}
Note that the inequality \eqref{eq_bound_still_discrete} for $n=0$ can
 be derived in the same way. By Lemma \ref{lem_grassmann_formulation}
 \eqref{item_grassmann_formulation} we can send $h\to \infty$ in
 \eqref{eq_bound_still_discrete} so that 
\begin{align}
&\left|\frac{L}{2\pi}\left(e^{i\frac{2\pi}{L}\<\sum_{j=1}^2(\hat{s}(\hs_j)\hbx_j-\hat{s}(\htau_j)\hby_j),\be_p\>}-1\right)\right|^n|\<\psi_{\hcX_1}^*\psi_{\hcX_2}^*\psi_{\hcY_2}\psi_{\hcY_1}+\text{
 h.c}\>_L|\notag\\
&\le 2^{13}c_0^2n!e^n\cF_{t,\beta}(8/\pi^2)^{-n}.
\label{eq_bound_not_discrete}
\end{align}
As we have fixed the parameters arbitrarily in the beginning of
 Subsection \ref{subsec_sketch_multiscale}, we can claim \eqref{eq_bound_not_discrete} for all
 $n\in\N\cup\{0\}$, $p\in \{1,2\}$, $U_c,U_o\in\R$ with $|U_c|,|U_o|<
 2^{-4}\alpha^{-2}c_0^{-2}M^{N_{\beta}}$ and sufficiently large
 $L\in\N$.

Set 
\begin{equation*}
f_1(E_{max}):=2^5\alpha^2\left(\frac{\max\{c,1\}}{(1-\eps)\eps^2}M^9\right)^2,\quad 
f_2(E_{max}):=2^{14}\left(\frac{\max\{c,1\}}{(1-\eps)\eps^2}M^9\right)^2.
\end{equation*}
By remarking \eqref{eq_beta_M} and \eqref{eq_final_coefficient} we can
 confirm that 
\begin{equation*}
f_1(E_{max})\max\{1,\beta^{16}\}\beta>2^4\alpha^2c_0^2M^{-N_{\beta}},\quad
 f_2(E_{max})\max\{1,\beta^{16}\}=2^{14}c_0^2,
\end{equation*}
$f_1(E_{max})$, $f_2(E_{max})$ are non-decreasing with respect to
 $E_{max}\in\R_{\ge 1}$ and $f_1(E_{max})=O(E_{max}^{44})$,
 $f_2(E_{max})=O(E_{max}^{36})$ as $E_{max}\to \infty$. 

It is straightforward to derive the following inequality from \eqref{eq_bound_not_discrete}.   
\begin{align*}
&|\<\psi_{\hcX_1}^*\psi_{\hcX_2}^*\psi_{\hcY_2}\psi_{\hcY_1}+\text{ h.c
 }\>_{L}|\le f_2(E_{max})\max\{1,\beta^{16}\}\\
&\quad\cdot
 \left(\frac{1}{\max\{1,t^2\}\max\{\beta,\beta^2\}}+1\right)^{-\frac{1}{8e}\sum_{p=1}^2\left|\frac{e^{i2\pi\<\sum_{j=1}^2(\hat{s}(\hs_j)\hbx_j-\hat{s}(\htau_j)\hby_j),\be_p\>/L}-1}{2\pi/L}\right|}
\end{align*}
for any $U_c, U_o\in\R$ with $|U_c|,|U_o|\le
 (f_1(E_{max})\max\{1,\beta^{16}\}\beta)^{-1}$ and sufficiently large
 $L\in \N$. Finally by Lemma \ref{lem_thermodynamic_limit} proved in
 Appendix \ref{app_thermodynamic_limit} we can take the limit
 $L\to \infty$ and complete the proof.
\end{proof}

\appendix
\section{Derivation of the covariance}\label{app_covariance}
In this part of Appendix we derive the representation of the covariance
\eqref{eq_covariance_characterization},
\eqref{eq_covariance_inside_function}.  Define the $3\times 3$ matrix
$M_{t,\bk}^{\s}=(M_{t,\bk}^{\s}(\rho,\eta))_{1\le \rho,\eta\le 3}$
$(\bk=(k_1,k_2)\in\G^*,t\in\R,\s\in\spin)$ by 
$$
M_{t,\bk}^{\s}:=\left(\begin{array}{ccc}\ec & t(1+e^{-ik_1}) &
		 t(1+e^{-ik_2}) \\
t(1+e^{ik_1}) & \eo & 0 \\
t(1+e^{ik_2}) & 0 & \eo\end{array}\right).
$$
We see that $H_0=\sum_{(\rho,\bx),(\eta,\by)\atop \in \{1,2,3\}\times\G}\sum_{\s\in
 \spin}\frac{1}{L^2}\sum_{\bk\in\G^*}e^{i\<\bx-\by,\bk\>}M_{t,\bk}^{\s}(\rho,\eta)\psi^*_{\rho\bx\s}\psi_{\eta\by\s}$.
For $\bk\in\G^*$, $t\in\R$,
$\s\in\spin$ and $\rho\in\{2,3\}$ set 
\begin{equation}\label{eq_eigen_values}
\begin{split}
&A_1^{\s}(t,\bk):=\left\{\begin{array}{ll}\ec&\text{ if
		   }\bk=(\pi,\pi)\text{ in }\G^*\text{ or }t=0,\\
       \eo&\text{ otherwise,}\end{array}
\right.\\
&A_{\rho}^{\s}(t,\bk):=\left\{\begin{array}{l}\eo\quad\text{  if
		   }\bk=(\pi,\pi)\text{ in }\G^*\text{ or }t=0,\\
       \frac{1}{2}(\ec+\eo)+(-1)^{\rho}\frac{t}{2}\sqrt{\frac{1}{t^2}(\ec-\eo)^2+8\sum_{j=1}^2(1+\cos k_j)}
\quad\text{  otherwise.}\end{array}
\right.
\end{split}
\end{equation}
Recall \eqref{eq_first_def_E}, i.e, $E(t,\bk)=2t^2\sum_{j=1}^2(1+\cos
k_j)$ $(t\in\R,\bk\in\G^*)$, and
define the $3\times 3$ matrix
$\cU_{t,\bk}^{\s}=(\cU_{t,\bk}^{\s}(\rho,\eta))_{1\le \rho,\eta\le 3}$ by
$\cU_{t,\bk}^{\s}(\rho,\eta):=\delta_{\rho,\eta}$ if
$\bk=(\pi,\pi)$ in $\G^*$ or $t=0$,
\begin{equation*}
\cU_{t,\bk}^{\s}:=\left(\begin{array}{ccc} 0 &
		    \frac{A_2^{\s}(t,\bk)-\eo}{((A_2^{\s}(t,\bk)-\eo)^2+E(t,\bk))^{1/2}}
		     &
		     \frac{A_3^{\s}(t,\bk)-\eo}{((A_3^{\s}(t,\bk)-\eo)^2+E(t,\bk))^{1/2}} \\
\frac{t(1+e^{-ik_2})}{E(t,\bk)^{1/2}} &
 \frac{t(1+e^{ik_1})}{((A_2^{\s}(t,\bk)-\eo)^2+E(t,\bk))^{1/2}} &
 \frac{t(1+e^{ik_1})}{((A_3^{\s}(t,\bk)-\eo)^2+E(t,\bk))^{1/2}} \\
\frac{-t(1+e^{-ik_1})}{E(t,\bk)^{1/2}} &
 \frac{t(1+e^{ik_2})}{((A_2^{\s}(t,\bk)-\eo)^2+E(t,\bk))^{1/2}} &
 \frac{t(1+e^{ik_2})}{((A_3^{\s}(t,\bk)-\eo)^2+E(t,\bk))^{1/2}} 
\end{array}
\right)
\end{equation*}
otherwise. One can check that $\cU_{t,\bk}^{\s}$ is unitary and 
\begin{equation}\label{eq_hopping_diagonal}
(\cU_{t,\bk}^{\s})^*M_{t,\bk}^{\s}\cU_{t,\bk}^{\s}=\left(\begin{array}{ccc}A_1^{\s}(t,\bk)
						      & 0 & 0 \\
0 & A_2^{\s}(t,\bk) & 0\\
0 & 0 & A_3^{\s}(t,\bk)\end{array}\right).
\end{equation}

By using $\cU_{t,\bk}^{\s}$ let us define the matrix $W_t=(W_t(\rho\bx
\s,\eta\by\tau))_{(\rho,\bx,\s),(\eta,\by,\tau)\in
\{1,2,3\}\times\G\times\spin}$ by
$W_t(\rho\bx\s,\eta\by\tau):=\frac{\delta_{\s,\tau}}{L^2}\sum_{\bk\in
\G^*}e^{-i\<\bx-\by,\bk\>}\overline{\cU_{t,\bk}^{\s}(\rho,\eta)}$. One
can also verify that\\
$(W_t^*W_t)(\rho\bx\s,\eta\by\tau)=1_{(\rho,\bx,\s)=(\eta,\by,\tau)}$.
With the matrix $W_t$ define the operator
$G(W_t):F_f(L^2(\{1,2,3\}\times\G\times \spin))\to
F_f(L^2(\{1,2,3\}\times\G\times \spin))$ by 
\begin{align*}
&G(W_t)\O:=\O,\\
&G(W_t)\psi_{\rho_1\bx_1\s_1}^*\psi_{\rho_2\bx_2\s_2}^*\cdots\psi_{\rho_n\bx_n\s_n}^*\O:=(W_t\psi^*)_{\rho_1\bx_1\s_1}(W_t\psi^*)_{\rho_2\bx_2\s_2}\cdots(W_t\psi^*)_{\rho_n\bx_n\s_n}\O\\
&\quad(n\in\N,(\rho_j,\bx_j,\s_j)\in\{1,2,3\}\times\G\times\spin\
 (j=1,\cdots,n)),
\end{align*}
and by linearity. Here the notation $\O$ represents the vacuum of
$F_f(L^2(\{1,2,3\}\times\G\times \spin))$ and
$(W_t\psi^*)_{\rho\bx\s}:=\sum_{(\eta, \by,\tau)\atop
\in\{1,2,3\}\times\G\times\spin }W_t(\rho\bx\s,\eta\by\tau)\psi_{\eta\by\tau}^*
$. The operator $G(W_t)$ is unitary. By letting
$(\overline{W_t}\psi)_{\rho\bx\s}$ denote $\sum_{(\eta, \by,\tau)\atop
\in\{1,2,3\}\times\G\times\spin
}\overline{W_t(\rho\bx\s,\eta\by\tau)}\psi_{\eta\by\tau}$, we observe
that $G(W_t)H_0\phi=\tilde{H}_0G(W_t)\phi$ for any $\phi\in
F_f(L^2(\{1,2,3\}\times\G\times \spin))$, where
\begin{equation*}
\tilde{H}_0:=\sum_{(\rho,\bx),(\eta, \by)\atop
\in\{1,2,3\}\times\G}\sum_{\s\in \spin}\frac{1}{L^2}\sum_{\bk\in\G^*}e^{i\<\bx-\by,\bk\>}M_{t,\bk}^{\s}(\rho,\eta)(W_t\psi^*)_{\rho\bx\s}(\overline{W_t}\psi)_{\eta\by\s}.
\end{equation*}
By using \eqref{eq_hopping_diagonal} we have 
\begin{equation*}
\tilde{H}_0=\sum_{\rho\in\{1,2,3\}}\sum_{\bx,\by\in\G}\sum_{\s\in
 \spin}\left(\frac{1}{L^2}\sum_{\bk\in\G^*}e^{i\<\bx-\by,\bk\>}A_{\rho}^{\s}(t,\bk)\right)\psi_{\rho\bx\s}^{*}\psi_{\rho\by\s}.
\end{equation*}

For $(\rho,\bx,\s,x)$,
$(\eta,\by,\tau,y)\in\{1,2,3\}\times\G\times\spin\times[0,\beta)$ let
\begin{align*}
&\tilde{\psi}_{\rho\bx\s}^*(x):=e^{x\tilde{H}_0}\psi_{\rho\bx\s}^*e^{-x\tilde{H}_0},\quad \tilde{\psi}_{\eta\by\tau}(y):=e^{y\tilde{H}_0}\psi_{\eta\by\tau}e^{-y\tilde{H}_0},\\ 
&T(\tilde{\psi}_{\rho\bx\s}^*(x)\tilde{\psi}_{\eta\by\tau}(y)):=1_{x\ge y
}\tilde{\psi}_{\rho\bx\s}^*(x)\tilde{\psi}_{\eta\by\tau}(y)-1_{x<y}\tilde{\psi}_{\eta\by\tau}(y)\tilde{\psi}_{\rho\bx\s}^*(x).
\end{align*}
The unitary property of $G(W_t)$ implies that
\begin{align}
\cC(\rho\bx\s x,\eta\by\tau
 y)=&\sum_{(\rho',\bx',\s'),(\eta',\by',\tau')\atop\in \{1,2,3\}\times
 \G\times\spin}W_t(\rho\bx\s,\rho'\bx'\s')\overline{W_t(\eta\by\tau,\eta'\by'\tau')}\notag\\
&\cdot\frac{\Tr(e^{-\beta
 \tilde{H}_0}T(\tilde{\psi}_{\rho'\bx'\s'}^*(x)\tilde{\psi}_{\eta'\by'\tau'}(y)))}{\Tr
 e^{-\beta\tilde{H}_0}}.\label{eq_covariance_pre_form}
\end{align}
Since $\tilde{H}_0$ is diagonal with respect to $\rho\in\{1,2,3\}$, the
characterization of \\
$\Tr(e^{-\beta
 \tilde{H}_0}T(\tilde{\psi}_{\rho'\bx'\s'}^*(x)\tilde{\psi}_{\eta'\by'\tau'}(y)))/\Tr
 e^{-\beta\tilde{H}_0}$ can be carried out by a standard argument.
See, e.g, \cite[\mbox{Appendix B}]{K1} for the derivation of the
covariance governed by a free Hamiltonian defined on $F_f(L^2(\G\times
\spin))$. As the result we obtain
\begin{align}
&\frac{\Tr(e^{-\beta
 \tilde{H}_0}T(\tilde{\psi}_{\rho'\bx'\s'}^*(x)\tilde{\psi}_{\eta'\by'\tau'}(y)))}{\Tr
 e^{-\beta\tilde{H}_0}}\notag\\
&=\frac{\delta_{\rho',\eta'}\delta_{\s',\tau'}}{L^2}\sum_{\bk\in\G^*}e^{-i\<\bx'-\by',\bk\>}e^{(x-y)A_{\rho'}^{\s'}(t,\bk)}\left(\frac{1_{x\ge y}}{1+e^{\beta A_{\rho'}^{\s'}(t,\bk)}}-\frac{1_{x<y}}{1+e^{-\beta A_{\rho'}^{\s'}(t,\bk)}}\right).\label{eq_covariance_ordered_temperature}
\end{align}
Substituting \eqref{eq_covariance_ordered_temperature} into
\eqref{eq_covariance_pre_form} yields that for $(\rho,\bx,\s,x)$,
$(\eta,\by,\tau,y)\in\{1,2,3\}\times\G\times\spin\times[0,\beta)$,
\begin{align}
\cC(\rho\bx\s x,\eta\by \tau y)
=&\frac{\delta_{\s,\tau}}{L^2}\sum_{\g\in\{1,2,3\}}\sum_{\bk\in\G^*}e^{-i\<\bx-\by,\bk\>}e^{(x-y)A_{\g}^{\s}(t,\bk)}\notag\\
&\cdot\left(\frac{1_{x\ge y}}{1+e^{\beta A_{\g}^{\s}(t,\bk)}}-\frac{1_{x<y
 }}{1+e^{-\beta A_{\g}^{\s}(t,\bk)}}\right)\overline{\cU_{t,\bk}^{\s}(\rho,\g)}\cU_{t,\bk}^{\s}(\eta,\g).\label{eq_covariance_ordered_combination}
\end{align}
Moreover by applying \cite[\mbox{Lemma C.3}]{K1} to the right-hand side of
\eqref{eq_covariance_ordered_combination} one reaches the equality that
for $(\rho,\bx,\s,x)$, $(\eta,\by,\tau,y)\in I_{L,h}$,
\begin{equation*}
\cC(\rho\bx\s x,\eta\by \tau y)
=\frac{\delta_{\s,\tau}}{\beta L^2}\sum_{(\bk,\o)\in\G^*\times\cM_h}
e^{-i\<\bx-\by,\bk\>}e^{i(x-y)\o}\sum_{\g\in\{1,2,3\}}\frac{\overline{\cU_{t,\bk}^{\s}(\rho,\g)}\cU_{t,\bk}^{\s}(\eta,\g)}{h(1-e^{-i\o/h+A_{\g}^{\s}(t,\bk)/h})}.
\end{equation*}

We need to show that for any $(\bk,\o)\in\G^*\times
 \cM_h$, $\rho,\eta\in\{1,2,3\}$, $t\in\R$, $\s\in\spin$,
\begin{equation}\label{eq_required_equality}
\sum_{\g\in\{1,2,3\}}\frac{\overline{\cU_{t,\bk}^{\s}(\rho,\g)}\cU_{t,\bk}^{\s}(\eta,\g)}{h(1-e^{-i\o/h+A_{\g}^{\s}(t,\bk)/h})}=\cB_{\rho,\eta}^{\s}(\bk,\o),
\end{equation}
where $\cB_{\rho,\eta}^{\s}(\bk,\o)$ is written in
\eqref{eq_covariance_inside_function}. The equality
\eqref{eq_required_equality} can be confirmed by direct calculation. To
assist the readers' verification, we present some intermediate results
appearing in the calculation. The functions $O_j^{\s}(\cdot):\C^2\to
\C$ $(j\in\{1,\cdots,5\},\s\in\spin)$ in \eqref{eq_covariance_inside_function}
are in fact given
as follows.
\begin{align}
&O_1^{\s}(\bk):=-2\sum_{n=2}^{\infty}\frac{1}{(2n)!h^{2n-2}}\left(\frac{(\ec-\eo)^2}{4}+E(t,\bk)\right)^n,\notag\\
&O_2^{\s}(\bk):=-\sum_{n=1}^{\infty}\frac{1}{(2n)!h^{2n-1}}\left(\frac{(\ec-\eo)^2}{4}+E(t,\bk)\right)^n\notag\\
&\qquad\qquad+\frac{\ec-\eo}{2}\sum_{n=1}^{\infty}\frac{1}{(2n+1)!h^{2n}}\left(\frac{(\ec-\eo)^2}{4}+E(t,\bk)\right)^n,\notag\\
&O_3^{\s}(\bk):=\sum_{n=1}^{\infty}\frac{1}{(2n+1)!h^{2n}}\left(\frac{(\ec-\eo)^2}{4}+E(t,\bk)\right)^n,\notag\\
&O_4^{\s}(\bk):=\sum_{n=2}^{\infty}\frac{1}{(2n)!h^{2n-2}}\sum_{m=1}^n\left(\begin{array}{c}n\\
									   m\end{array}\right)E(t,\bk)^{m-1}\left(\frac{(\ec-\eo)^2}{4}\right)^{n-m},\notag\\
&O_5^{\s}(\bk):=-\frac{\ec-\eo}{2}\sum_{n=1}^{\infty}\frac{1}{(2n+1)!h^{2n-1}}\sum_{m=1}^n\left(\begin{array}{c}n\\
									   m\end{array}\right)E(t,\bk)^{m-1}\left(\frac{(\ec-\eo)^2}{4}\right)^{n-m}.\label{eq_vanishing_functions}
\end{align}
From \eqref{eq_vanishing_functions} one can see that
\eqref{eq_decay_extra_terms} holds.

First assume that $\bk=(\pi,\pi)$ in $\G^*$ or $t=0$. In this case
$E(t,\bk)=0$ and thus $\cD^{\s}(\bk,\o)$ and
$\cN^{\s}_{\rho,\eta}(\bk,\o)$ given in
\eqref{eq_covariance_inside_function} are simplified as follows. 
\begin{align*}
&\cD^{\s}(\bk,\o)=h^2(1-e^{-i\o/h+\ec/h})(1-e^{-i\o/h+\eo/h}),\quad
 \cN_{1,1}^{\s}(\bk,\o)=h(1-e^{-i\o/h+\eo/h}),\\
&\cN_{\rho,\eta}^{\s}(\bk,\o)=0\quad (\forall
 (\rho,\eta)\in\{1,2,3\}^2\backslash\{(1,1)\}).
\end{align*}
By using these, the equality \eqref{eq_required_equality} can be
confirmed in this case.

Next consider the case that $\bk\neq (\pi,\pi)$ in $\G^*$ and $t\neq
0$. To organize the calculation, set
$f(w,A):=h(1-e^{-i\o/h+A/h})$. Remark that for $\rho\in \{2,3\}$, 
\begin{align*}
&f(\o,A_{\rho}^{\s}(t,\bk))=h-he^{-\frac{i}{h}\o+\frac{\ec+\eo}{2h}}\\
&\quad\cdot\Bigg(\sum_{n=0}^{\infty}\frac{1}{(2n)!h^{2n}}\Bigg(\frac{t}{2}\Bigg(\frac{(\ec-\eo)^2}{t^2}+8\sum_{j=1}^2(1+\cos
 k_j)\Bigg)^{1/2}\Bigg)^{2n}\\
&\quad +(-1)^{\rho}\sum_{n=0}^{\infty}\frac{1}{(2n+1)!h^{2n+1}}\Bigg(\frac{t}{2}\Bigg(\frac{(\ec-\eo)^2}{t^2}+8\sum_{j=1}^2(1+\cos
 k_j)\Bigg)^{1/2}\Bigg)^{2n+1}\Bigg),\\
&f(\o,A_2^{\s}(t,\bk))f(\o,A_3^{\s}(t,\bk))=\cD^{\s}(\bk,\o),\\
&((A_2^{\s}(t,\bk)-\eo)^2+E(t,\bk))((A_3^{\s}(t,\bk)-\eo)^2+E(t,\bk))\\
&=4E(t,\bk)^2+E(t,\bk)(\ec-\eo)^2.
\end{align*}
By using these equalities we observe that 
\begin{align}
&(\text{The left-hand side of }\eqref{eq_required_equality}\text{ for }(\rho,\eta)=(1,1))\notag\\
&=\big(f(\o,A_3^{\s}(t,\bk))(A_2^{\s}(t,\bk)-\eo)^{2}((A_3^{\s}(t,\bk)-\eo)^{2}+E(t,\bk))\notag\\
&\quad+f(\o,A_2^{\s}(t,\bk))(A_3^{\s}(t,\bk)-\eo)^{2}((A_2^{\s}(t,\bk)-\eo)^{2}+E(t,\bk))\big)\notag\\
&\quad\cdot\big/\big((4E(t,\bk)^2+E(t,\bk)(\ec-\eo)^2)\cD^{\s}(\bk,\o)\big)\notag\\
&=\frac{(4E(t,\bk)^2+E(t,\bk)(\ec-\eo)^2)\cN^{\s}_{1,1}(\bk,\o)}{(4E(t,\bk)^2+E(t,\bk)(\ec-\eo)^2)\cD^{\s}(\bk,\o)}=\cB_{1,1}^{\s}(\bk,\o),\notag\\
&(\text{The left-hand side of }\eqref{eq_required_equality}\text{ for }(\rho,\eta)=(1,2))\notag\\
&=t(1+e^{ik_1})\big(f(\o,A_3^{\s}(t,\bk))(A_2^{\s}(t,\bk)-\eo)((A_3^{\s}(t,\bk)-\eo)^{2}+E(t,\bk))\notag\\
&\quad+f(\o,A_2^{\s}(t,\bk))(A_3^{\s}(t,\bk)-\eo)((A_2^{\s}(t,\bk)-\eo)^{2}+E(t,\bk))\big)\notag\\
&\quad\cdot\big/\big((4E(t,\bk)^2+E(t,\bk)(\ec-\eo)^2)\cD^{\s}(\bk,\o)\big)\notag\\
&=\frac{(4E(t,\bk)^2+E(t,\bk)(\ec-\eo)^2)\cN^{\s}_{1,2}(\bk,\o)}{(4E(t,\bk)^2+E(t,\bk)(\ec-\eo)^2)\cD^{\s}(\bk,\o)}=\cB_{1,2}^{\s}(\bk,\o),\notag\\
&(\text{The left-hand side of }\eqref{eq_required_equality}\text{ for }(\rho,\eta)=(2,2))=\frac{1}{f(\o,A_1^{\s}(t,\bk))}+2t^2(1+\cos
 k_1)\notag\\
&\quad\cdot\left(\frac{-1}{E(t,\bk)f(\o,A_1^{\s}(t,\bk))}+\sum_{j=2}^3\frac{1}{((A_j^{\s}(t,\bk)-\eo)^2+E(t,\bk))f(\o,A_j^{\s}(t,\bk))}\right).\label{eq_intermediate_1_1}
\end{align}
Note that
\begin{align}
&\frac{-1}{E(t,\bk)f(\o,A_1^{\s}(t,\bk))}+\sum_{j=2}^3\frac{1}{((A_j^{\s}(t,\bk)-\eo)^2+E(t,\bk))f(\o,A_j^{\s}(t,\bk))}\notag\\
&=\frac{-1}{E(t,\bk)f(\o,A_1^{\s}(t,\bk))}+\Bigg(h-he^{-\frac{i}{h}\o+\frac{\ec+\eo}{2h}}\sum_{n=0}^{\infty}\frac{1}{(2n)!h^{2n}}\left(\frac{(\ec-\eo)^2}{4}+E(t,\bk)\right)^n\notag\\
&\quad -\frac{\ec-\eo}{2}e^{-\frac{i}{h}\o+\frac{\ec+\eo}{2h}}\sum_{n=0}^{\infty}\frac{1}{(2n+1)!h^{2n}}\left(\frac{(\ec-\eo)^2}{4}+E(t,\bk)\right)^n\Bigg)\big/(E(t,\bk)\cD^{\s}(t,\bk))\notag\\
&=\Bigg(e^{-\frac{i}{h}\o+\frac{\ec+\eo}{2h}}\notag\\
&\quad\cdot\Bigg(\sum_{n=0}^{\infty}\left(\frac{1}{(2n)!h^{2n-2}}-\frac{\ec-\eo}{2(2n+1)!h^{2n-1}}\right)\left(\frac{(\ec-\eo)^2}{4}+E(t,\bk)\right)^n-h^2e^{-\frac{\ec-\eo}{2h}}\Bigg)\notag\\
&\quad+
 e^{-\frac{2i}{h}\o+\frac{\ec+3\eo}{2h}}\notag\\
&\quad\cdot\Bigg(\sum_{n=0}^{\infty}\left(\frac{1}{(2n)!h^{2n-2}}+\frac{\ec-\eo}{2(2n+1)!h^{2n-1}}\right)\left(\frac{(\ec-\eo)^2}{4}+E(t,\bk)\right)^n-h^2e^{\frac{\ec-\eo}{2h}}\Bigg)\Bigg)\notag\\&\quad\cdot \big/ (E(t,\bk)f(\o,A_1^{\s}(t,\bk))\cD^{\s}(t,\bk))\notag\\
&=\Big(\frac{1}{2}e^{-\frac{i}{h}\o+\frac{\ec+\eo}{2h}}+\frac{1}{2}e^{-\frac{2i}{h}\o+\frac{\ec+3\eo}{2h}}+(e^{-\frac{i}{h}\o+\frac{\ec+\eo}{2h}}+e^{-\frac{2i}{h}\o+\frac{\ec+3\eo}{2h}})O_4^{\s}(\bk)\notag\\
&\quad+
(e^{-\frac{i}{h}\o+\frac{\ec+\eo}{2h}}-e^{-\frac{2i}{h}\o+\frac{\ec+3\eo}{2h}})O_5^{\s}(\bk)\Big)\notag\\
&\quad\cdot\big/(f(\o,A_1^{\s}(t,\bk))\cD^{\s}(t,\bk)).\label{eq_intermediate_next_1_1}
\end{align}
By inserting \eqref{eq_intermediate_next_1_1} into
\eqref{eq_intermediate_1_1} we obtain \eqref{eq_required_equality} for
$(\rho,\eta)=(2,2)$. Moreover, by using \eqref{eq_intermediate_next_1_1},
\begin{align*}
&(\text{The left-hand side of
 }\eqref{eq_required_equality}\text{ for }(\rho,\eta)=(2,3))=t^2(1+e^{-ik_1})(1+e^{ik_2})\\
&\quad\cdot\left(\frac{-1}{E(t,\bk)f(\o,A_1^{\s}(t,\bk))}+\sum_{j=2}^3\frac{1}{((A_j^{\s}(t,\bk)-\eo)^2+E(t,\bk))f(\o,A_j^{\s}(t,\bk))}\right)\\
&=\frac{\cN_{2,3}^{\s}(\bk,\o)}{f(\o,A_1^{\s}(t,\bk))\cD^{\s}(t,\bk)}=\cB_{2,3}^{\s}(\bk,\o).
\end{align*}

By using the results for $(\rho,\eta)=(1,1),(1,2),(2,2),(2,3)$ and
symmetries, \eqref{eq_required_equality} for
$(\rho,\eta)=(1,3),(2,1),(3,1),(3,2),(3,3)$ can be immediately
proved. Thus, the representations \eqref{eq_covariance_characterization},
\eqref{eq_covariance_inside_function} have been derived.

\section{Convergence of the Grassmann integral formulation}
\label{app_formulation_convergence}
In this section we sketch how to prove Lemma
\ref{lem_grassmann_formulation}. With a parameter $\la\in\C$ let us
introduce the modified Hamiltonian $H_{\la}$ by
$H_{\la}:=H+\la(\psi_{\hcX_1}^*\psi_{\hcX_2}^*\psi_{\hcY_2}\psi_{\hcY_1}+\text{
h.c})$. 
It follows that $H_{\la}=H_0+\sum_{X_1,X_2,Y_1,Y_2\atop \in
			  \{1,2,3\}\times\G\times\spin}U_{(\la,\la)}(X_1,X_2,Y_1,Y_2)\psi_{X_1}^*\psi_{X_2}^*\psi_{Y_1}\psi_{Y_2}$, where $U_{(\la,\la)}$ is introduced in
\eqref{eq_original_bi_anti_symmetric}. The partition function $\Tr
e^{-\beta H_{\la}}/\Tr e^{-\beta H_0}$ can be expanded as a perturbation
series by straightforwardly following \cite[\mbox{Appendix B}]{K1}.
\begin{align}
&\frac{\Tr e^{-\beta H_{\la}}}{\Tr e^{-\beta
 H_0}}\notag\\
&=1+\sum_{n=1}^{\infty}\frac{1}{n!}\prod_{m=1}^n\Bigg(\sum_{X_{2m-1},X_{2m},Y_{2m-1},Y_{2m}\atop\in
 \{1,2,3\}\times \G\times \spin
 }\int_{0}^{\beta}ds_{2m-1}U_{(\la,\la)}(X_{2m-1},X_{2m},Y_{2m-1},Y_{2m})\Bigg)\notag\\
&\quad\cdot \det(\cC(X_ps_p,Y_qs_q))_{1\le p,q\le
 2n}\Big|_{s_{2j}=s_{2j-1}\atop \forall j\in \{1,\cdots,n\}}.\label{eq_perturbation_partition}
\end{align}

Let the function $P(\la,U_c,U_o)$ $(:\C^3\to \C)$ be defined by the
right-hand side of \eqref{eq_perturbation_partition}. Moreover, by
replacing the integral over $[0,\beta)$ in the right-hand side of
\eqref{eq_perturbation_partition} by the Riemann sum we
can define the discrete analogue of $P$.
\begin{align*}
&P_{h}(\la,U_c,U_o):=\\
&1+\sum_{n=1}^{N_{L,h}/2}\frac{1}{n!}\prod_{m=1}^n\Bigg(\sum_{X_{2m-1},X_{2m},Y_{2m-1},Y_{2m}\atop\in
 \{1,2,3\}\times \G\times \spin
 }\frac{1}{h}\sum_{s_{2m-1}\in [0,\beta)_h}U_{(\la,\la)}(X_{2m-1},X_{2m},Y_{2m-1},Y_{2m})\Bigg)\\
&\quad\cdot \det(\cC(X_ps_p,Y_qs_q))_{1\le p,q\le
 2n}\Big|_{s_{2j}=s_{2j-1}\atop \forall j\in \{1,\cdots,n\}}.
\end{align*}
The function $P_h$ uniformly converges to $P$ in the following
sense. For any $U>0$,
\begin{equation}\label{eq_convergence_partition}
\lim_{h\to \infty\atop h\in 2\N/\beta}\sup_{(\la,U_c,U_o)\in\C^3\atop
 |\la|,|U_c|,|U_o|\le U}|P_h(\la,U_c,U_o)-P(\la,U_c,U_o)|=0.
\end{equation}
To prove the convergence property \eqref{eq_convergence_partition}
 we need to use the determinant bound of the following form.
\begin{equation}\label{eq_volume_dependent_bound}
|\det (\cC(\rho_p\bx_p\s_px_p,\eta_q\by_q\tau_qy_q))_{1\le p,q\le
 n}|\le C_1(L)\cdot C_2(L)^n,
\end{equation}
where the constants $C_1(L)$, $C_2(L)>0$ may depend on $L$, but are independent of $n$
and how to choose $(\rho_p,\bx_p,\s_p,x_p)$,
$(\eta_p,\by_p,\tau_p,y_p)\in\{1,2,3\}\times\G\times\spin\times[0,\beta)$
 $(p=1,\cdots,n)$. 
The bound \eqref{eq_volume_dependent_bound} can be verified as follows.
We can choose the operators $A_1,A_2,\cdots,A_{2n}$ from
$\{e^{x_pH_0}\psi_{\rho_p\bx_p\s_p}^*e^{-x_pH_0},e^{y_pH_0}\psi_{\eta_p\by_p\tau_p}e^{-y_pH_0}\}_{p=1}^n$
so that
\begin{align*}
|\det (\cC(\rho_p\bx_p\s_px_p,\eta_q\by_q\tau_qy_q))_{1\le p,q\le n}|&=\big|\Tr (e^{-\beta
 H_0}A_1A_2\cdots A_{2n})\big|\big/\Tr e^{-\beta H_0}\\
&\le \frac{2^{6L^2}}{\Tr e^{-\beta
 H_0}}\left(e^{\beta\|H_0\|}\right)^{2n+1},
\end{align*}
where $\|H_0\|$ denotes the operator norm of $H_0$. 

Let us recall that in \cite[\mbox{Lemma 3.4}]{K2} Pedra-Salmhofer's determinant
bound \cite[\mbox{Theorem 2.4}]{PS} was applied to prove the essentially
same statements as Lemma \ref{lem_grassmann_formulation}. Though we do
not have a volume-independent determinant bound like \\
\cite[\mbox{Theorem
2.4}]{PS} on our covariance $\cC$ at hand, the crude bound
\eqref{eq_volume_dependent_bound} sufficiently works to show
\eqref{eq_convergence_partition} in the argument parallel to the proof
of \cite[\mbox{Lemma 3.4}]{K2}.

The following equality directly follows from the definition of the
Grassmann Gaussian integral and $P_h$.
\begin{equation}\label{eq_grassmann_partition}
\int e^{V_{(\la,\la)}(\psi)}d\mu_{\cC}(\psi)=P_h(\la,U_c,U_o)\quad
 (\forall (\la,U_c,U_o)\in\C^3).
\end{equation}
Since $\inf_{(\la,U_c,U_o)\in\R^3, |\la|,|U_c|,|U_o|\le U}P(\la,U_c,U_o)>0$, the uniform convergence property
\eqref{eq_convergence_partition} and the equality \eqref{eq_grassmann_partition}
ensure the claim \eqref{item_real_part} of Lemma
\ref{lem_grassmann_formulation}.

By using \cite[\mbox{Lemma 2.3}]{K1} and
\eqref{eq_convergence_partition} we have for any $U_c$, $U_o\in\R$ and $\delta >0$,
\begin{align}
\<\psi_{\hcX_1}^*\psi_{\hcX_2}^*\psi_{\hcY_2}\psi_{\hcY_1}+\text{h.c}\>_L&=-\frac{1}{\beta}\frac{\partial}{\partial\la}\log
 P(\la,U_c,U_o)\Big|_{\la=0}\notag\\
&=-\frac{1}{\beta}\frac{1}{P(0,U_c,U_o)}\frac{1}{2\pi
 i}\oint_{|\la|=\delta}d\la \frac{P(\la,U_c, U_o)}{\la^2}\notag\\
&=-\frac{1}{\beta}\lim_{h\to \infty\atop h\in
 2\N/\beta}\frac{1}{P_h(0,U_c,U_o)}\frac{1}{2\pi
 i}\oint_{|\la|=\delta}d\la \frac{P_h(\la,U_c, U_o)}{\la^2}\notag\\
&=-\frac{1}{\beta}\lim_{h\to \infty\atop h\in
 2\N/\beta}\frac{\partial}{\partial\la}\log P_h(\la,U_c,U_o)\Big|_{\la =0}.\label{eq_cauchy_limit}
\end{align}
Substituting \eqref{eq_grassmann_partition} into the right-hand side of
\eqref{eq_cauchy_limit} yields the claim
\eqref{item_grassmann_formulation} of Lemma \ref{lem_grassmann_formulation}.

\section{Logarithm of Grassmann polynomials}
\label{app_log_grassmann}
The aim of this section is to extend the notion of logarithm of
Grassmann polynomials summarized in \cite{FKT} to be available for
Grassmann polynomials with complex constant terms. In the following
let $f_0,g_0\in \C$ denote the constant term of $f,g\in\bigwedge
\cV$, respectively.
\begin{definition}\label{def_log_grassmann}
For $f\in\bigwedge \cV$ with $\Re f_0>0$, $\log f\in \bigwedge \cV$
 is defined by 
$$
\log f:=\log(f_0)+\sum_{n=1}^{2N_{L,h}}\frac{(-1)^{n-1}}{n}\left(\frac{f-f_0}{f_0}\right)^n,
$$
where $\log z:=\log |z|+i\Arg z$, $\Arg z\in (-\pi/2,\pi/2)$ for 
 $z\in\C$ with $\Re z>0$. 
\end{definition}
Recall that for $f\in\bigwedge \cV$, $e^f\in\bigwedge \cV$ is
defined by 
\begin{equation}\label{eq_def_exponential_grassmann}
e^f:=e^{f_0}\sum_{n=0}^{2N_{L,h}}\frac{1}{n!}(f-f_0)^n.
\end{equation}
It was proved in \cite[\mbox{Problem I.2}]{FKT} that for any $f,g\in
\bigwedge \cV$ satisfying $fg=gf$,
\begin{equation}\label{eq_exponential_decomposition}
e^f\cdot e^g=e^g\cdot e^f=e^{f+g}.
\end{equation}
The following equality was also shown in \cite[\mbox{Problem I.4
b)}]{FKT}. For any $f\in \bigwedge \cV$ with $f_0\in \R_{>0}$,
\begin{equation}\label{eq_exponential_log_previous}
e^{\log f}=f.
\end{equation}
The multi-scale analysis in this paper needs an extension of
\eqref{eq_exponential_log_previous}.
\begin{lemma}\label{lem_exponential_log_equality}
For any $f\in \bigwedge \cV$ with $\Re f_0\in \R_{>0}$, $e^{\log f}=f$.
\end{lemma}
\begin{proof}
Take $f\in \bigwedge \cV$ with $\Re f_0>0$. Since
 $\log(|f_0|^2)=\log(f_0)+\log(\overline{f_0})$, 
\begin{equation}\label{eq_log_decomposition_preparation}
\log(\overline{f_0}\cdot
 f)=\log(|f_0|^2)+\sum_{n=1}^{2N_{L,h}}\frac{(-1)^{n-1}}{n}\left(\frac{\overline{f_0}\cdot
 f-|f_0|^2}{|f_0|^2}\right)^n=\log(\overline{f_0})+\log f.
\end{equation}
It follows from \eqref{eq_exponential_log_previous} that
\begin{equation}\label{eq_exp_log_preparation}
e^{\log(\overline{f_0}\cdot f)}=\overline{f_0}\cdot f.
\end{equation}
By using \eqref{eq_exponential_decomposition},
 \eqref{eq_log_decomposition_preparation} and
 \eqref{eq_exp_log_preparation} we observe that
\begin{equation*}
e^{\log f}=e^{-\log(\overline{f_0})+\log(\overline{f_0}\cdot
 f)}=e^{-\log (\overline{f_0})}\cdot e^{\log (\overline{f_0}\cdot f)}=\frac{1}{\overline{f_0}}\cdot \overline{f_0}\cdot f=f.
\end{equation*}
\end{proof}

\section{Existence of the thermodynamic limit}
\label{app_thermodynamic_limit}
Here we show that the correlation function
$\<\psi_{\hcX_1}^*\psi_{\hcX_2}^*\psi_{\hcY_2}\psi_{\hcY_1}+\text{h.c}\>_L$
converges to a finite value as $L\to \infty$ if $|U_c|$, $|U_o|$ are
smaller than certain value. The idea of the proof is
similar to \cite[\mbox{Appendix B}]{K2} and based on the perturbative
expansion of logarithm of the Grassmann Gaussian
integral. We also use the following lemma.
\begin{lemma}\label{lem_covariance_thermodynamic_limit}
\begin{enumerate}[(i)]
\item\label{item_full_covariance_decay}
For any $(\rho,\bx,\s,x)$,
     $(\eta,\by,\tau,y)\in\{1,2,3\}\times\Z^2\times\spin\times
     [0,\beta)$ with $x\neq y$,
$$
|\cC(\rho\bx\s x,\eta\by\tau
     y)|\le\frac{c(E_{max},\beta)}{1+\sum_{p=1}^2\left(\frac{L}{2\pi}\right)^3\left|e^{i2\pi\<\bx-\by,\be_p\>/L}-1\right|^3},
$$ 
where the constant $c(E_{max},\beta)>0$ depends only on $E_{max}$ and $\beta$.
\item\label{item_covariance_limit}
For any $(\rho,\bx,\s,x)$,
     $(\eta,\by,\tau,y)\in\{1,2,3\}\times\Z^2\times\spin\times
     [0,\beta)$, \\
$\lim_{L\to\infty,L\in\N}\cC(\rho\bx\s x,\eta\by\tau
     y)$ exists.
\end{enumerate}
\end{lemma}
\begin{proof}
\eqref{item_full_covariance_decay}: Take any $(\rho,\bx,\s,x)$,
     $(\eta,\by,\tau,y)\in\{1,2,3\}\times\Z^2\times\spin\times
     [0,\beta)$. By using the notations introduced in Appendix
 \ref{app_covariance}, set
\begin{equation*}
\begin{split}
g_{L,(\rho,\s,x),(\eta,\tau,y)}(\bk)
:=&\delta_{\s,\tau}\sum_{\g\in\{1,2,3\}}e^{(x-y)A_{\g}^{\s}(t,\bk)}\\
&\cdot\left(\frac{1_{x\ge y}}{1+e^{\beta A_{\g}^{\s}(t,\bk)}}-\frac{1_{x<y}}{1+e^{-\beta A_{\g}^{\s}(t,\bk)}}\right)\overline{\cU_{t,\bk}^{\s}(\rho,\g)}\cU_{t,\bk}^{\s}(\eta,\g).
\end{split}
\end{equation*}
By \eqref{eq_covariance_ordered_combination},
$\cC(\rho\bx\s x,\eta\by\tau
 y)=\frac{1}{L^2}\sum_{\bk\in\G^*}e^{-i\<\bx-\by,\bk\>}g_{L,(\rho,\s,x),(\eta,\tau,y)}(\bk)$. Since $|\cU_{t,\bk}^{\s}(\rho,\g)|$,\\
$|\cU_{t,\bk}^{\s}(\eta,\g)|\le 1$, $|g_{L,(\rho,\s,x),(\eta,\tau,y)}(\bk)|\le 3$. This implies that $|\cC(\rho\bx\s x,\eta\by\tau y)|\le 3$.

Let us additionally assume that $x\neq y$. In this case we can
 expand $g_{L,(\rho,\s,x),(\eta,\tau,y)}(\bk)$ as a sum over $\pi
 (2\Z+1)/\beta$ so that  
\begin{equation*}
\cC(\rho\bx\s x,\eta\by \tau y)
=\frac{\delta_{\s,\tau}}{\beta L^2}\sum_{\bk\in\G^*}\sum_{\o\in
\pi(2\Z+1)/\beta}e^{-i\<\bx-\by,\bk\>}e^{i(x-y)\o}\cB_{\rho,\eta}^{\s,\infty}(\bk,\o),
\end{equation*}
where 
$$
\cB_{\rho,\eta}^{\s,\infty}(\bk,\o):=\sum_{\g\in\{1,2,3\}}\frac{\overline{\cU_{t,\bk}^{\s}(\rho,\g)}\cU_{t,\bk}^{\s}(\eta,\g)}{i\o-A_{\g}^{\s}(t,\bk)}.
$$
We can see from \eqref{eq_required_equality} that
 $\cB_{\rho,\eta}^{\s,\infty}(\bk,\o)=\lim_{h\to \infty,h\in
 2\N/\beta}\cB_{\rho,\eta}^{\s}(\bk,\o)$. Thus by setting
$$
\cD^{\s,\infty}(\bk,\o):=\left(i\o-\frac{1}{2}(\ec+\eo)\right)^2-\frac{1}{4}(\ec-\eo)^2-2t^2\sum_{j=1}^2(1+\cos
 k_j),$$
it follows from \eqref{eq_covariance_inside_function} that for any
 $\bk=(k_1,k_2)\in\G^*$ and $\o\in \pi(2\Z+1)/\beta$,
\begin{align*}
&\cB_{1,1}^{\s,\infty}(\bk,\o)=\frac{i\o-\eo}{\cD^{\s,\infty}(\bk,\o)},\
\cB_{1,2}^{\s,\infty}(\bk,\o)=\frac{t(1+e^{ik_1})}{\cD^{\s,\infty}(\bk,\o)},\
 \cB_{1,3}^{\s,\infty}(\bk,\o)=\cB_{1,2}^{\s,\infty}((k_2,k_1),\o),\\
&\cB_{2,1}^{\s,\infty}(\bk,\o)=\cB_{1,2}^{\s,\infty}(-\bk,\o),\ \cB_{2,2}^{\s,\infty}(\bk,\o)=\frac{1}{i\o-\eo}+\frac{2t^2(1+\cos
 k_1)}{(i\o-\eo)\cD^{\s,\infty}(\bk,\o)},\\
&\cB_{2,3}^{\s,\infty}(\bk,\o)=\frac{t^2(1+e^{-ik_1})(1+e^{ik_2})}{(i\o-\eo)\cD^{\s,\infty}(\bk,\o)},\
 \cB_{3,1}^{\s,\infty}(\bk,\o)=
 \cB_{1,2}^{\s,\infty}(-(k_2,k_1),\o),\\
&\cB_{3,2}^{\s,\infty}(\bk,\o)=\cB_{2,3}^{\s,\infty}(-\bk,\o),\ \cB_{3,3}^{\s,\infty}(\bk,\o)=\cB_{2,2}^{\s,\infty}((k_2,k_1),\o).
\end{align*}

Periodicity with respect to $\bk\in\G^*$ guarantees that for $p\in\{1,2\}$,
\begin{align}
&\left(\frac{L}{2\pi}\left(e^{i\frac{2\pi}{L}\<\bx-\by,\be_p\>}-1\right)\right)^3\cC(\rho\bx\s
 x,\eta\by\tau y)=\frac{\delta_{\s,\tau}}{\beta
 L^2}\sum_{\bk\in\G^*}\sum_{\o\in\pi(2\Z+1)/\beta}e^{-i\<\bx-\by,\bk\>}e^{i(x-y)\o}\notag\\
&\quad\cdot\prod_{j=1}^3\left(\frac{L}{2\pi}\int_0^{2\pi/L}d\theta_j\right)\left(\frac{\partial}{\partial
 k_p}\right)^3\cB_{\rho,\eta}^{\s,\infty}\left(\bk+\sum_{j=1}^3\theta_j\be_p,\o\right).\label{eq_periodicity_infty}
\end{align}
Note that for any $\bk\in\R^2$, $\o\in \pi(2\Z+1)/\beta$,
\begin{align}
|\cD^{\s,\infty}(\bk,\o)|&\ge \max\{|\Re \cD^{\s,\infty}(\bk,\o)|,|\Im
 \cD^{\s,\infty}(\bk,\o)|\}\notag\\
&\ge \max\{\o^2-\ec\eo,|\o(\ec+\eo)|\}\ge \frac{1}{2}\o^2.\label{eq_D_infty_lb}
\end{align}
By using \eqref{eq_D_infty_lb} we can estimate the equality \eqref{eq_periodicity_infty} and deduce that
\begin{equation}\label{eq_full_covariance_decay_bound}
\left|\frac{L}{2\pi}\left(e^{i\frac{2\pi}{L}\<\bx-\by,\be_p\>}-1\right)\right|^3|\cC(\rho\bx\s
 x,\eta\by\tau y)|\le
 \frac{1}{\beta}\sum_{\o\in \pi(2\Z+1)/\beta}\frac{c(E_{max},\beta)}{\o^2}\le c(E_{max},\beta).
\end{equation}
By coupling \eqref{eq_full_covariance_decay_bound} with the bound
 $|\cC(\rho\bx\s x,\eta\by \tau y)|\le 3$ we obtain the inequality in \eqref{item_full_covariance_decay}.

\eqref{item_covariance_limit}: Remark that for any $(\rho,\bx,\s,x)$,
 $(\eta,\by,\tau,y)\in\{1,2,3\}\times\G\times \spin\times [0,\beta)$,
\begin{equation*}
\cC(\rho\bx\s x,\eta \by\tau
 y)=\frac{1}{(2\pi)^2}\int_{[-\pi,\pi)}dp_1\int_{[-\pi,\pi)}dp_2\tilde{g}_{L,(\rho,\bx,\s,x),(\eta,\by,\tau,y)}(p_1,p_2),
\end{equation*}
where
 $\tilde{g}_{L,(\rho,\bx,\s,x),(\eta,\by,\tau,y)}(p_1,p_2):=e^{-i\<\bx-\by,\bk\>}{g}_{L,(\rho,\s,x),(\eta,\tau,y)}(k_1,k_2)$
 with $k_j\in$
 $\{-\pi,-\pi+2\pi/L,\cdots,\pi-2\pi/L\}$ satisfying that
 $p_j\in [k_j,k_j+2\pi/L)$ $(j=1,2)$. Since $\bk\mapsto g_{L,(\rho,\s,x),(\eta,\tau,y)}(\bk)$ is continuous in $(-\pi,\pi)^2$
 by definition, $\lim_{L\to
 \infty,L\in\N}$$\tilde{g}_{L,(\rho,\bx,\s,x),(\eta,\by,\tau,y)}(\bp)$ exists
 for any $\bp\in (-\pi,\pi)^2$. As we have seen above,
 $|\tilde{g}_{L,(\rho,\bx,\s,x),(\eta,\by,\tau,y)}(\bp)|=$\\
 $|{g}_{L,(\rho,\s,x),(\eta,\tau,y)}(\bk)|\le
 3$. Therefore, the dominated convergence theorem concludes that 
\begin{equation*}
\lim_{L\to \infty\atop L\in\N}\cC(\rho\bx\s x,\eta\by\tau
 y)=\frac{1}{(2\pi)^2}\int_{[-\pi,\pi)^2}d\bp \lim_{L\to \infty\atop L\in\N}\tilde{g}_{L,(\rho,\bx,\s,x),(\eta,\by,\tau,y)}(\bp).
\end{equation*}
\end{proof}

\begin{lemma}\label{lem_thermodynamic_limit}
Assume that $U_c, U_o\in\R$ and \eqref{eq_smallness_assumption} holds with $c_0$ defined in
 \eqref{eq_final_coefficient}. Then,
 $\<\psi_{\hcX_1}^*\psi_{\hcX_2}^*\psi_{\hcY_2}\psi_{\hcY_1}+\text{h.c}\>_L$
 converges to a finite value as $L\to \infty$ ($L\in\N$).
\end{lemma}
\begin{proof}
Fix $U_c,U_o\in\R$ with $|U_c|,|U_o|<
 2^{-4}\alpha^{-2}c_0^{-2}M^{N_{\beta}}$. It follows from Lemma
 \ref{lem_grassmann_formulation} \eqref{item_grassmann_formulation} and
 \eqref{eq_consequence_identity} for $n=0$ that
$$
\<\psi_{\hcX_1}^*\psi_{\hcX_2}^*\psi_{\hcY_2}\psi_{\hcY_1}+\text{h.c}\>_L=-\frac{1}{\beta}\lim_{h\to
 \infty\atop h\in
 2\N/\beta}\sum_{a\in\{1,-1\}}\frac{\partial}{\partial\la_a}J_0^{\ge N_{\beta}}(\b0)\Big|_{(\la_1,\la_{-1})=(0,0)}.
$$
Thus, it suffices to prove the convergence of 
\begin{equation}\label{eq_target_convergence}
\lim_{L\to \infty\atop L\in\N}\lim_{h\to \infty\atop h\in 2\N/\beta}\sum_{a\in\{1,-1\}}\frac{\partial}{\partial\la_a}J_0^{\ge N_{\beta}}(\b0)\Big|_{(\la_1,\la_{-1})=(0,0)}.
\end{equation}

In order to make clear the dependency on $U_c$, $U_o$ we write 
$$\frac{\partial}{\partial\la_a}J_0^{\ge
 N_{\beta}}(\b0)\Big|_{(\la_1,\la_{-1})=(0,0)}(U_c,U_o)$$ 
in place of $(\partial/\partial\la_a)J_0^{\ge
 N_{\beta}}(\b0)|_{(\la_1,\la_{-1})=(0,0)}$. We can take $\eps>0$ such that
 $(1+\eps)|U_c|$,
 $(1+\eps)|U_o|<2^{-4}\alpha^{-2}c_0^{-2}M^{N_{\beta}}$. By Corollary
 \ref{cor_analyticity_variables} \eqref{item_complex_variables} there is
 a domain $D_o\subset\C$ containing the disk $\{z\in \C\ |\ |z|\le
 1+\eps\}$ inside such that
$$
z\mapsto \sum_{a\in\{1,-1\}}\frac{\partial}{\partial\la_a}J_0^{\ge N_{\beta}}(\b0)\Big|_{(\la_1,\la_{-1})=(0,0)}(zU_c,zU_o)
$$
is analytic in $D_o$. Thus,
\begin{align*}
&\sum_{a\in\{1,-1\}}\frac{\partial}{\partial\la_a}J_0^{\ge
 N_{\beta}}(\b0)\Big|_{(\la_1,\la_{-1})=(0,0)}(U_c,U_o)\\
&=\sum_{n=0}^{\infty}\frac{1}{n!}\left(\frac{d}{dz}\right)^n\left(
\sum_{a\in\{1,-1\}}\frac{\partial}{\partial\la_a}J_0^{\ge
 N_{\beta}}(\b0)\Big|_{(\la_1,\la_{-1})=(0,0)}(zU_c,zU_o)\right)\Big|_{z=0}.
\end{align*}
Moreover, by Proposition \ref{prop_bound_schwinger}, for any
 $n\in\N\cup\{0\}$,
\begin{align*}
&\left|\frac{1}{n!}\left(\frac{d}{dz}\right)^n\left(
\sum_{a\in\{1,-1\}}\frac{\partial}{\partial\la_a}J_0^{\ge
 N_{\beta}}(\b0)\Big|_{(\la_1,\la_{-1})=(0,0)}(zU_c,zU_o)\right)\Big|_{z=0}\right|\\
&=\left|\frac{1}{2\pi i}\oint_{|z|=1+\eps}dz\cdot z^{-n-1}
\sum_{a\in\{1,-1\}}\frac{\partial}{\partial\la_a}J_0^{\ge
 N_{\beta}}(\b0)\Big|_{(\la_1,\la_{-1})=(0,0)}(zU_c,zU_o)\right|\\
&\le 2^{13}\beta c_0^2(1+\eps)^{-n}.
\end{align*}
Since $(1+\eps)^{-n}$ is summable over $\N\cup\{0\}$, the dominant
 convergence theorem guarantees that \eqref{eq_target_convergence} 
converges if 
\begin{equation}\label{eq_wanted_convergence}
\lim_{L\to \infty\atop L\in\N}\lim_{h\to \infty\atop h\in 2\N/\beta}
\left(\frac{d}{dz}\right)^n\left(
\sum_{a\in\{1,-1\}}\frac{\partial}{\partial\la_a}J_0^{\ge
 N_{\beta}}(\b0)\Big|_{(\la_1,\la_{-1})=(0,0)}(zU_c,zU_o)\right)\Big|_{z=0}
\end{equation}
exists for all $n\in\N\cup \{0\}$. 

Again by \eqref{eq_consequence_identity} for $n=0$ we can write for any
 $x\in \R$ with $|x|\le 1+\eps$ that
\begin{align}
&\sum_{a\in\{1,-1\}}\frac{\partial}{\partial\la_a}J_0^{\ge
 N_{\beta}}(\b0)\Big|_{(\la_1,\la_{-1})=(0,0)}(xU_c,xU_o)\notag\\
&=\frac{\partial}{\partial \la}\log\left(\int
 e^{V_{(\la,\la)}(\psi)}d\mu_{\cC}(\psi)\right)\Big|_{\la=0}(xU_c,xU_o),\label{eq_identity_application}
\end{align}
which implies that
\begin{align*}
&\sum_{a\in\{1,-1\}}\frac{\partial}{\partial\la_a}J_0^{\ge
 N_{\beta}}(\b0)\Big|_{(\la_1,\la_{-1})=(0,0)}(0,0)\\
&=-\beta (\det(\cC(\hcX_p0,\hcY_q0))_{1\le p,q\le
 2}+\det(\cC(\hcY_p0,\hcX_q0))_{1\le p,q\le 2}).
\end{align*}
Thus, Lemma \ref{lem_covariance_thermodynamic_limit}
 \eqref{item_covariance_limit} proves the existence of
 \eqref{eq_wanted_convergence} for $n=0$.

It follows from \eqref{eq_identity_application} that for any $n\in\N$,
\begin{align*}
&\frac{1}{n!}\left(\frac{d}{dz}\right)^n\left(
\sum_{a\in\{1,-1\}}\frac{\partial}{\partial\la_a}J_0^{\ge
 N_{\beta}}(\b0)\Big|_{(\la_1,\la_{-1})=(0,0)}(zU_c,zU_o)\right)\Big|_{z=0}\\
&=\frac{\partial}{\partial \la}\left(\frac{1}{(n+1)!}\left(\frac{d}{dx}\right)^{n+1}\log\left(\int
 e^{x
 V_{(\la,\la)}(\psi)}d\mu_{\cC}(\psi)\right)\Big|_{x=0}\right)\Big|_{\la=0}\\
&=\frac{\partial}{\partial
 \la}\cP_0T_{ree}(n+1,\cC,V_{(\la,\la)})\Big|_{\la=0}.
\end{align*}
Recall that $T_{ree}(\cdot,\cdot,\cdot)$ is defined in
 \eqref{eq_tree_expansion}. In the expansion of
 $\cP_0T_{ree}(n+1,\cC,V_{(\la,\la)})$ we apply the operator
 $\prod_{\{q,r\}\in T}(\D_{q,r}(\cC)+\D_{r,q}(\cC))$ first and then
 erase the rest of Grassmann polynomials by the operator
 $e^{\sum_{q,r=1}^{n+1}M_{at}(T,\xi,\bs)_{q,r}\D_{q,r}(\cC)}$. By
 recalling the notation \eqref{eq_tree_subset} we observe that
\begin{align}
&\frac{\partial}{\partial
 \la}\cP_0T_{ree}(n+1,\cC,V_{(\la,\la)})\Big|_{\la=0}\notag\\
&=\sum_{a\in\{1,-1\}}\sum_{T\in\T_{n+1}}\frac{1}{h}\sum_{x_1\in
 [0,\beta)_h}\notag\\
&\quad\cdot\prod_{j=2}^{n+1}\left(\sum_{\rho_j\in\{1,2,3\}}(1_{\rho_j=1}U_c+1_{\rho_j=2,3}U_o)\sum_{\s_1^j,\s_2^j,\tau_1^j,\tau_2^j\atop\in
 \spin}1_{(\s_1^j,\s_2^j,\tau_1^j,\tau_2^j)=(\ua,\da,\da,\ua)}\frac{1}{h}\sum_{(\bx_j,x_j)\in\G\times[0,\beta)_h}\right)\notag\\
&\quad\cdot\prod_{\{1,r\}\in
 L_1^1(T)}\left(\sum_{k_{\{1,r\}}=1}^2\sum_{l_{\{1,r\}}=1}^2\sum_{b_{\{1,r\}}\in\{1,-1\}}\cC_{\{1,r\},a}^{k_{\{1,r\}},l_{\{1,r\}},b_{\{1,r\}}}(x_1,\bx_r
 x_r)\right)\notag\\
&\quad\cdot\prod_{q=2}^{n+1}\prod_{\{q,r\}\in L_q^1(T)}\left(\sum_{k_{\{q,r\}}=1}^2\sum_{l_{\{q,r\}}=1}^2\sum_{b_{\{q,r\}}\in\{1,-1\}}\cC_{\{q,r\}}^{k_{\{q,r\}},l_{\{q,r\}},b_{\{q,r\}}}(\bx_qx_q,\bx_r
 x_r)\right)\notag\\
&\quad\cdot f(T,a,\{k_{\{q,r\}},l_{\{q,r\}},b_{\{q,r\}}\}_{\{q,r\}\in T},\cC),\label{eq_tree_expansion_limit}
\end{align}
where 
\begin{align*}
&\cC_{\{1,r\},a}^{k_{\{1,r\}},l_{\{1,r\}},b_{\{1,r\}}}(x_1,\bx_r
 x_r):=\left\{\begin{array}{ll}\cC(\hcX_{k_{\{1,r\}}}x_1,\rho_r\bx_r\tau_{l_{\{1,r\}}}^rx_r)&\text{
	if }a=1,\ b_{\{1,r\}}=1,\\
\cC(\rho_r\bx_r\s_{l_{\{1,r\}}}^rx_r,\hcY_{k_{\{1,r\}}}x_1)&\text{
	if }a=1,\ b_{\{1,r\}}=-1,\\
\cC(\hcY_{k_{\{1,r\}}}x_1,\rho_r\bx_r\tau_{l_{\{1,r\}}}^rx_r)&\text{
	if }a=-1,\ b_{\{1,r\}}=1,\\
\cC(\rho_r\bx_r\s_{l_{\{1,r\}}}^rx_r,\hcX_{k_{\{1,r\}}} x_1)&\text{
	if }a=-1,\ b_{\{1,r\}}=-1,
\end{array}\right.\\
&\cC_{\{q,r\}}^{k_{\{q,r\}},l_{\{q,r\}},b_{\{q,r\}}}(\bx_qx_q,\bx_r
 x_r):=\left\{\begin{array}{ll}\cC(\rho_q\bx_q\s_{k_{\{q,r\}}}^qx_q,\rho_r\bx_r\tau_{l_{\{q,r\}}}^rx_r)&\text{
	if }b_{\{q,r\}}=1,\\
\cC(\rho_r\bx_r\s_{l_{\{q,r\}}}^rx_r,\rho_q\bx_q\tau_{k_{\{q,r\}}}^qx_q)&\text{
	if }b_{\{q,r\}}=-1,
\end{array}
\right.\\
&f(T,a,\{k_{\{q,r\}},l_{\{q,r\}},b_{\{q,r\}}\}_{\{q,r\}\in T},\cC)\\
&:=\frac{1}{n!}\int_{[0,1]^n}d\bs\sum_{\xi\in\S_{n+1}(T)}\varphi(T,\xi,\bs)e^{\sum_{u,v=1}^{n+1}M_{at}(T,\xi,\bs)_{u,v}\D_{u,v}(\cC)}\\
&\quad\cdot\prod_{\{1,r\}\in
 L_1^1(T)}\cL_{\{1,r\},a}^{k_{\{1,r\}},l_{\{1,r\}},b_{\{1,r\}}}(x_1,\bx_rx_r)\prod_{q=2}^{n+1}\prod_{\{q,r\}\in L_q^1(T)}\cL_{\{q,r\}}^{k_{\{q,r\}},l_{\{q,r\}},b_{\{q,r\}}}(\bx_qx_q,\bx_rx_r)\\
&\quad\cdot(-1_{a=1}\opsi_{\hcX_1x_1}^1\opsi_{\hcX_2x_1}^1\psi_{\hcY_2x_1}^1\psi_{\hcY_1x_1}^1-1_{a=-1}\opsi_{\hcY_1x_1}^1\opsi_{\hcY_2x_1}^1\psi_{\hcX_2x_1}^1\psi_{\hcX_1x_1}^1)\\
&\quad\cdot\prod_{s=2}^{n+1}(-\opsi_{\rho_s\bx_s\s_1^s
 x_s}^s\opsi_{\rho_s\bx_s\s_2^s x_s}^s\psi_{\rho_s\bx_s\tau_1^s
 x_s}^s\psi_{\rho_s\bx_s\tau_2^s x_s}^s)\Big|_{\psi^j=\b0\atop \forall
 j\in\{1,\cdots,n+1\}},\\
&\cL_{\{1,r\},a}^{k_{\{1,r\}},l_{\{1,r\}},b_{\{1,r\}}}(x_1,\bx_r
 x_r)\\
&:=\left\{\begin{array}{ll}-(\partial/\partial\opsi_{\hcX_{k_{\{1,r\}}}x_1}^1)(\partial/\partial\psi_{\rho_r\bx_r\tau_{l_{\{1,r\}}}^rx_r}^r)&\text{
	if }a=1,\ b_{\{1,r\}}=1,\\
-(\partial/\partial\opsi_{\rho_r\bx_r\s_{l_{\{1,r\}}}^rx_r}^r)(\partial/\partial \psi_{\hcY_{k_{\{1,r\}}}x_1}^1)&\text{
	if }a=1,\ b_{\{1,r\}}=-1,\\
-(\partial/\partial \opsi_{\hcY_{k_{\{1,r\}}}x_1}^1)(\partial/\partial\psi_{\rho_r\bx_r\tau_{l_{\{1,r\}}}^rx_r}^r)&\text{
	if }a=-1,\ b_{\{1,r\}}=1,\\
-(\partial/\partial \opsi_{\rho_r\bx_r\s_{l_{\{1,r\}}}^rx_r}^r)(\partial/\partial \psi_{\hcX_{k_{\{1,r\}}} x_1}^1)&\text{
	if }a=-1,\ b_{\{1,r\}}=-1,
\end{array}\right.\\
&\cL_{\{q,r\}}^{k_{\{q,r\}},l_{\{q,r\}},b_{\{q,r\}}}(\bx_qx_q,\bx_r
 x_r)\\
&:=\left\{\begin{array}{ll}-(\partial/\partial
	\opsi_{\rho_q\bx_q\s_{k_{\{q,r\}}}^qx_q}^q)(\partial/\partial \psi_{\rho_r\bx_r\tau_{l_{\{q,r\}}}^rx_r}^r)&\text{
	if }b_{\{q,r\}}=1,\\
-(\partial/\partial
 \opsi_{\rho_r\bx_r\s_{l_{\{q,r\}}}^rx_r}^r)(\partial/\partial \psi_{\rho_q\bx_q\tau_{k_{\{q,r\}}}^qx_q}^q)&\text{
	if }b_{\{q,r\}}=-1.
\end{array}
\right.
\end{align*}
By the translation invariance and the periodicity of $\cC(\rho\bx\s
 x,\eta \by \tau y)$ with respect to $\bx,\by\in\G$,
\begin{align}
&\frac{\partial}{\partial \la}\cP_0T_{ree}(n+1,\cC,V_{(\la,\la)})\Big|_{\la=0}\notag\\
&=\frac{1}{h}\sum_{x_1\in
 [0,\beta)_h}\prod_{j=2}^{n+1}\left(\frac{1}{h}\sum_{(\bx_j,x_j)\in\G\times[0,\beta)_h}\right)F_L(x_1,\bx_2x_2,\cdots,\bx_{n+1}x_{n+1}),\label{eq_result_translation_invariance}
\end{align}
where
\begin{align}
&F_L(x_1,\bx_2x_2,\cdots,\bx_{n+1}x_{n+1}):=\notag\\
&\sum_{a\in\{1,-1\}}\sum_{T\in\T_{n+1}}\prod_{j=2}^{n+1}\left(\sum_{\rho_j\in\{1,2,3\}}(1_{\rho_j=1}U_c+1_{\rho_j=2,3}U_o)\sum_{\s_1^j,\s_2^j,\tau_1^j,\tau_2^j\atop\in
 \spin}1_{(\s_1^j,\s_2^j,\tau_1^j,\tau_2^j)=(\ua,\da,\da,\ua)}\right)\notag\\
&\quad\cdot\prod_{\{1,r\}\in
 L_1^1(T)}\left(\sum_{k_{\{1,r\}}=1}^2\sum_{l_{\{1,r\}}=1}^2\sum_{b_{\{1,r\}}\in\{1,-1\}}\tilde{\cC}_{\{1,r\},a}^{k_{\{1,r\}},l_{\{1,r\}},b_{\{1,r\}}}(x_1,\bx_r x_r)\right)\notag\\
&\quad\cdot\prod_{q=2}^{n+1}\prod_{\{q,r\}\in L_q^1(T)}\left(\sum_{k_{\{q,r\}}=1}^2\sum_{l_{\{q,r\}}=1}^2\sum_{b_{\{q,r\}}\in\{1,-1\}}\cC_{\{q,r\}}^{k_{\{q,r\}},l_{\{q,r\}},b_{\{q,r\}}}(\b0x_q,\bx_r
 x_r)\right)\notag\\
&\quad\cdot f(T,a,\{k_{\{q,r\}},l_{\{q,r\}},b_{\{q,r\}}\}_{\{q,r\}\in
 T},\cC),\label{eq_tree_expansion_limit_next}\\
&\tilde{\cC}_{\{1,r\},a}^{k_{\{1,r\}},l_{\{1,r\}},b_{\{1,r\}}}(x_1,\bx_r
 x_r)\notag\\
&:=\left\{\begin{array}{ll}\cC(\hat{\rho}_{k_{\{1,r\}}}\b0\hat{\s}_{k_{\{1,r\}}}x_1,\rho_r\bx_r\tau_{l_{\{1,r\}}}^rx_r)&\text{ if }a=1,\quad b_{\{1,r\}}=1,\\
\cC(\rho_{r}\bx_r\s_{l_{\{1,r\}}}^rx_r,\hat{\eta}_{k_{\{1,r\}}}\b0\hat{\tau}_{k_{\{1,r\}}}x_1)&\text{
 if }a=1,\quad b_{\{1,r\}}=-1,\\
\cC(\hat{\eta}_{k_{\{1,r\}}}\b0\hat{\tau}_{k_{\{1,r\}}}x_1,\rho_r\bx_r\tau_{l_{\{1,r\}}}^rx_r)&\text{
 if }a=-1,\quad b_{\{1,r\}}=1,\\
\cC(\rho_r\bx_r\s_{l_{\{1,r\}}}^rx_r,\hat{\rho}_{k_{\{1,r\}}}\b0\hat{\s}_{k_{\{1,r\}}}x_1)&\text{
 if }a=-1,\quad b_{\{1,r\}}=-1.
\end{array}\right.\notag
\end{align}
Though we do not explicitly write for simplicity, we should remark that
 the dependency of $f(T,a,\{k_{\{q,r\}},l_{\{q,r\}},b_{\{q,r\}}\}_{\{q,r\}\in T},\cC)$ on
 the variables $x_1\in [0,\beta)_h$, $(\bx_j,x_j)\in \G\times [0,\beta)_h$
$(j=2,\cdots,n+1)$ in \eqref{eq_tree_expansion_limit_next} is different
 from that in \eqref{eq_tree_expansion_limit}.

For $s_1\in [0,\beta)$, $(\bx_j,s_j)\in\Z^2\times
 [0,\beta)$ $(j=2,\cdots,n+1)$ set
 $$F_{L,h}(s_1,\bx_2s_2,\cdots,\bx_{n+1}s_{n+1}):=F_L(x_1,\bx_2
 x_2,\cdots,\bx_{n+1}x_{n+1}),$$
where $x_j\in [0,\beta)_h$ satisfies
 that $s_j\in [x_j,x_j+h^{-1})$ $(\forall j\in \{1,\cdots,n+1\})$. Since
 $(x,y)\mapsto \cC(\rho\bx\s x,\eta\by \tau y)$ is continuous a.e. in
 $[0,\beta)^2$, 
$$\lim_{h\to \infty\atop h\in
 2\N/\beta}F_{L,h}(s_1,\bx_2s_2,\cdots,\bx_{n+1}s_{n+1})=F_{L}(s_1,\bx_2s_2,\cdots,\bx_{n+1}s_{n+1})$$
 for a.e. $(s_1,s_2,\cdots,s_{n+1})\in [0,\beta)^{n+1}$, and thus
\begin{align*}
&\lim_{h\to \infty\atop h\in 2\N/\beta}\frac{\partial}{\partial \la}\cP_0T_{ree}(n+1,\cC,V_{(\la,\la)})\Big|_{\la=0}\\
&=\lim_{h\to \infty\atop h\in
 2\N/\beta}\int_0^{\beta}ds_1\prod_{j=2}^{n+1}\left(\int_0^{\beta}ds_j\sum_{\bx_j\in\G}\right)F_{L,h}(s_1,\bx_2s_2,\cdots,\bx_{n+1}s_{n+1})\\
&=\int_0^{\beta}ds_1\prod_{j=2}^{n+1}\left(\int_0^{\beta}ds_j\sum_{\bx_j\in\G}\right) F_{L}(s_1,\bx_2s_2,\cdots,\bx_{n+1}s_{n+1}).
\end{align*}

Lemma \ref{lem_covariance_thermodynamic_limit} \eqref{item_covariance_limit}
implies that $\lim_{L\to
 \infty,L\in\N}F_{L}(s_1,\bx_2s_2,\cdots,\bx_{n+1}s_{n+1})$ exists for
 any $s_1\in [0,\beta)$, $(\bx_j,s_j)\in\Z^2\times
 [0,\beta)$ $(j=2,\cdots,n+1)$.

By Lemma \ref{lem_covariance_thermodynamic_limit}
 \eqref{item_full_covariance_decay}, \eqref{eq_nice_property_phi} and
 the fact that $|M_{at}(T,\xi,\bs)_{u,v}|\le 1$ $(\forall u,v\in
 \{1,\cdots,n+1\})$ (see the proof of \cite[\mbox{Lemma 4.5}]{K1})
 there exists $c(n,T,E_{max},\beta)>0$
 depending only on $n$, $T$, $E_{max}$ and $\beta$ such that $|f(T,a,\{k_{\{q,r\}},l_{\{q,r\}},b_{\{q,r\}}\}_{\{q,r\}\in
 T},\cC)|\le c(n,T,E_{max},\beta)$ for
 a.e. $(s_1,\cdots,s_{n+1})\in [0,\beta)^{n+1}$. Therefore, by setting
 $U_{max}:=\max\{|U_c|,|U_o|\}$ and using Lemma \ref{lem_covariance_thermodynamic_limit} \eqref{item_full_covariance_decay},
\begin{align}
&1_{\bx_j\in\{-\lfloor
 L/2\rfloor,-\lfloor L/2\rfloor+1,\cdots,-\lfloor L/2\rfloor+L-1
 \}^2\ (\forall j\in
 \{2,\cdots,n+1\})}|F_{L}(s_1,\bx_2s_2,\cdots,\bx_{n+1}s_{n+1})|\notag\\
&\le U_{max}^n\sum_{T\in \T_{n+1}}c(n,T,E_{max},\beta)\prod_{\{1,r\}\in
 L_1^1(T)}\frac{1}{1+\left(\frac{2}{\pi}\right)^3\sum_{p=1}^2|\<\bx_r,\be_p\>|^3}\notag\\
&\quad\cdot \prod_{q=2}^{n+1}\prod_{\{q,r\}\in
 L_q^1(T)}\frac{1}{1+\left(\frac{2}{\pi}\right)^3\sum_{p=1}^2|\<\bx_r,\be_p\>|^3}\notag\\
&=U_{max}^n\prod_{j=2}^{n+1}\frac{1}{1+\left(\frac{2}{\pi}\right)^3\sum_{p=1}^2|\<\bx_j,\be_p\>|^3}\sum_{T\in
 \T_{n+1}}c(n,T,E_{max},\beta)
\label{eq_final_dominant_bound}
\end{align}
for a.e. $(s_1,\cdots,s_{n+1})\in [0,\beta)^{n+1}$ and any $\bx_j\in\Z^2$
 $(j=2,\cdots,n+1)$. The right-hand side of \eqref{eq_final_dominant_bound} is in
 $L^1([0,\beta)\times (\Z^2\times [0,\beta))^n)$. Thus, the dominated
 convergence theorem proves that
\begin{align*}
&\lim_{L\to \infty\atop L\in \N}\lim_{h\to\infty\atop h\in 2\N/\beta}\frac{\partial}{\partial\la}\cP_0T_{ree}(n+1,\cC,V_{(\la,\la)})\Big|_{\la=0}\\
&=\lim_{L\to \infty\atop L\in \N}\int_0^{\beta}ds_1\prod_{j=2}^{n+1}\left(\int_0^{\beta}ds_j\sum_{\bx_j\in\Z^2}\right)1_{\bx_j\in\{-\lfloor
 L/2\rfloor,-\lfloor L/2\rfloor+1,\cdots,-\lfloor L/2\rfloor+L-1
 \}^2\ (\forall j\in \{2,\cdots,n+1\})}\\
&\quad\cdot F_{L}(s_1,\bx_2s_2,\cdots,\bx_{n+1}s_{n+1})\\
&=\int_0^{\beta}ds_1\prod_{j=2}^{n+1}\left(\int_0^{\beta}ds_j\sum_{\bx_j\in\Z^2}\right)\lim_{L\to
 \infty\atop L\in \N}F_{L}(s_1,\bx_2s_2,\cdots,\bx_{n+1}s_{n+1}).
\end{align*}
This implies the existence of \eqref{eq_wanted_convergence}
 for $n\in \N$ and completes the proof.
\end{proof}

\section*{Acknowledgments}
The author wishes to thank the referees for their critical reading of
the manuscript.

\section*{Notation}
\subsection*{Parameters and constants}
\begin{tabular}{l|l|l}
Notation & Description & Reference \\
\hline
$L$ & size of lattice of the position variable & Subsection \ref{subsec_model_hamiltonian}\\
$t$ & hopping amplitude & Subsection \ref{subsec_model_hamiltonian}\\
$U_c$ & coupling constant on the Cu sites & Subsection
 \ref{subsec_model_hamiltonian}\\
$U_o$ & coupling constant on the O sites & Subsection
 \ref{subsec_model_hamiltonian}\\
$\ec$, $\eo$ & spin-dependent on-site energies & Subsection
 \ref{subsec_model_hamiltonian}\\
$(\s\in\spin)$ & & \\
$\beta$ & proportional to the inverse of temperature & Subsection
 \ref{subsec_model_hamiltonian}\\ 
$E_{max}$ & $\max_{\s\in\spin}\{1,|t|,|\ec|,|\eo|\}$ & beginning of Section \ref{sec_formulation}\\
$\hcX_j$, $\hcY_j$ & same as $(\hrho_j,\hbx_j,\hs_j)$,
$(\heta_j,\hby_j,\htau_j)$ $(j=1,2)$, & beginning of Section
\ref{sec_formulation}\\
$(j=1,2)$ & fixed sites to define the correlation function& \\
$h$ & step size of the discretization  &  beginning of Section
 \ref{sec_formulation}\\
 & of $[0,\beta)$, $[-\beta,\beta)$ & \\
$N_{L,h}$ & $6L^2\beta h$, cardinality of $I_{L,h}$ & beginning of Section
\ref{sec_formulation}\\
$\la_1$, $\la_{-1}$ & used to modify the interaction & Subsection \ref{subsec_grassmann_formulation}\\
$c$ & generic constant depending & beginning of Section
 \ref{sec_preliminaries}\\
 & only on a fixed smooth function & \\
$M$ & parameter to control the size of the support & Subsection
 \ref{subsec_cut_off}\\
 & of the cut-off function & \\
$N_h$ & $\lfloor \log(2h)/\log (M)\rfloor$ & Subsection \ref{subsec_cut_off}\\
$N_{\beta}$ & $\max\{\lfloor \log(1/\beta)/\log (M)\rfloor +1,1\}$ & Subsection \ref{subsec_cut_off}\\
$c_0$ & constant depending on $M$ and $\beta$ &\eqref{eq_final_coefficient}\\
$\alpha$ & additional parameter used &
 before Proposition \ref{prop_inductive_bound}\\
 &  in the multi-scale integration &  
\end{tabular}

\subsection*{Sets}
\begin{tabular}{l|l|l}
Notation & Description & Reference \\
\hline
$\G$ & $(\Z/L\Z)^2$ & Subsection
\ref{subsec_model_hamiltonian}\\
$[0,\beta)_h$ & $\{0,1/h,\cdots,\beta-1/h\}$ & beginning of Section
\ref{sec_formulation}\\
$[-\beta,\beta)_h$ & $\{-\beta,-\beta+1/h,\cdots,-1/h\}\cup [0,\beta)_h$ & beginning of Section
\ref{sec_formulation}\\
$\G^*$& $(\frac{2\pi}{L}\Z/(2\pi\Z))^2$  & beginning of
Section \ref{sec_formulation}\\
$\cM_h$ & $\{\o\in \pi (2\Z+1)/\beta\ |\ |\o|<\pi h\}$ & beginning of Section \ref{sec_formulation}\\
$I_{L,h}$ & $\{1,2,3\}\times \G\times \spin \times [0,\beta)_h$ & beginning of
Section \ref{sec_formulation}\\
$\tilde{I}_{L,h}$ & $I_{L,h}\times\{1,-1\}$ & Subsection
\ref{subsec_multi_notations}\\
$(I_{L,h})_o^m$ & subset of $I_{L,h}^m$ & Subsection
\ref{subsec_multi_notations}\\
$D_{small}$ & subset of $\C^4$ & Subsection
\ref{subsec_sketch_multiscale}\\
$D_R$ & subset of $\C$ & Subsection
\ref{subsec_sketch_multiscale}
\end{tabular}

\subsection*{Functions}
\begin{tabular}{l|l|l}
Notation & Description & Reference \\
\hline
$\cF_{t,\beta}(\cdot)$ & used to specify the domain & beginning of
 Section \ref{sec_formulation}\\
 & of analyticity of the covariance & \\
$\hat{s}(\cdot)$ & fixed function of spin & beginning of
Section \ref{sec_formulation}\\
$\cC(\cdot,\cdot)$ & covariance of full scale & Subsection
\ref{subsec_covariance}\\
$\chi_l(\cdot)$ & cut-off function of $l$-th scale & Subsection
\ref{subsec_cut_off}\\
$\cC_l(\cdot,\cdot)$ & covariance of $l$-th scale & beginning of Subsection \ref{subsec_sliced_covariance}
\end{tabular}

\end{document}